\documentclass[sigconf, nonacm]{acmart}

\usepackage{subcaption}
\usepackage[ruled,linesnumbered]{algorithm2e}

\usepackage{threeparttable}
\usepackage{booktabs}

\begin{document}
\title{Differential Privacy for Network Assortativity}

\author{Fei Ma}
\affiliation{%
  \institution{School of Computer Science\\Northwestern Polytechnical University}
  \city{Xi'an}
  \country{China}
  \postcode{710072}
}
\email{feima@nwpu.edu.cn}

\author{Jinzhi Ouyang}
\affiliation{%
  \institution{School of Computer Science\\Northwestern Polytechnical University}
  \city{Xi'an}
  \country{China}
  \postcode{710072}
}
\email{ouyangjinzhi@mail.nwpu.edu.cn}

\author{Xincheng Hu}
\affiliation{%
  \institution{School of Computer Science\\Northwestern Polytechnical University}
  \city{Xi'an}
  \country{China}
  \postcode{710072}
}
\email{xinchenghu@mail.nwpu.edu.cn}

\begin{abstract}
The analysis of network assortativity is of great importance for understanding the structural characteristics of and dynamics upon networks. Often, network assortativity is quantified using the assortativity coefficient that is defined based on the Pearson correlation coefficient between vertex degrees (see Eq.(\ref{eq:3-2-1}) for concrete expression). It is well known that a network may contain sensitive information, such as the number of friends of an individual in a social network (which is abstracted as the degree of vertex.). So, the computation of the assortativity coefficient leads to privacy leakage, which increases the urgent need for privacy-preserving protocol. However, there has been no scheme addressing the concern above. 


To bridge this gap, in this work, we are the first to propose approaches based on differential privacy (DP for short). Specifically, we design three DP-based algorithms: $\mathbf{Local_{ru}}$, $\mathbf{Shuffle_{ru}}$, and $\mathbf{Decentral_{ru}}$. The first two algorithms, based on Local DP (LDP) and Shuffle DP respectively, are designed for settings where each individual only knows his/her direct friends. In contrast, the third algorithm, based on Decentralized DP (DDP), targets scenarios where each individual has a broader view, i.e., also knowing his/her friends' friends. Theoretically, we prove that each algorithm enables an unbiased estimation of the assortativity coefficient of the network. We further evaluate the performance of the proposed algorithms using mean squared error (MSE), showing that $\mathbf{Shuffle_{ru}}$ achieves the best performance, followed by $\mathbf{Decentral_{ru}}$, with $\mathbf{Local_{ru}}$ performing the worst. Note that these three algorithms have different assumptions, so each has its applicability scenario. Lastly, we conduct extensive numerical simulations, which demonstrate that the presented approaches are adequate to achieve the estimation of network assortativity under the demand for privacy protection.

\end{abstract}

\maketitle



\section{Introduction}

Nowadays, network analysis plays a crucial role in understanding a great variety of complex systems such as social networks \cite{cinelli2021echo}, transportation networks \cite{ganin2017resilience}, biological networks \cite{gosak2018network}, and so on. It is well known that network analysis is often performed using various measures including degree distribution, clustering coefficient, diameter, and assortativity coefficient, to name but a few \cite{bhattacharya2020impact}. Among them, the assortativity coefficient (see subsection 3.2 for more details) \cite{newman2002assortative} is used to reflect the tendency of nodes to connect to other nodes with similar attributes or characteristics in networks under consideration, and has received increasing attention from various science communities \cite{noldus2015assortativity}. One of the main reasons for this is that the assortativity coefficient, as a fundamental measure, plays a key role in understanding the structural characteristics of and dynamics upon networks. Figure \ref{fig:example} shows an example of assortativity (see Eq.\ref{eq:3-2-1} for specific expression). Hereinafter, two terms network and graph are used interchangeably. 


\begin{figure}
  \centering
  \includegraphics[width=0.76\linewidth]{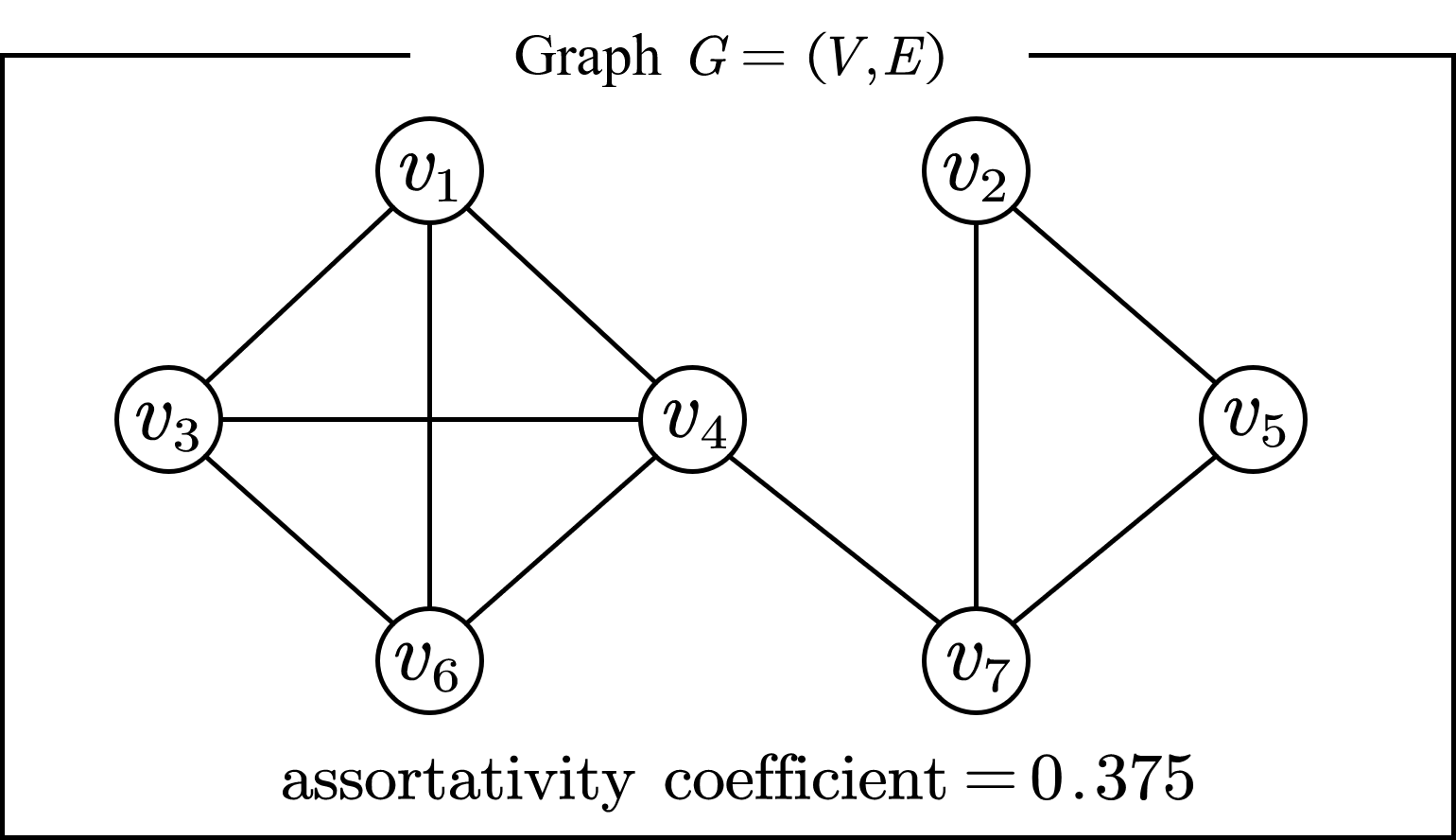}
  \caption{An example graph and its assortativity coefficient.}
  \label{fig:example}
\end{figure}

Specifically, in the assortative network, nodes with similar attributes tend to connect. Essentially, assortativity can be understood as `birds of a feather flock together,' where individuals with similar experiences, backgrounds, or knowledge tend to cluster together. The opposite connection between nodes is found in the disassortative network \cite{newman2003mixing}. Besides providing structural information about networks, network assortativity is also related to the dynamic behavior on networks \cite{melamed2018cooperation}. Assortativity of the network may influence information propagation pathways and speeds, as nodes preferentially connect with others sharing similar attributes, which potentially accelerates information dissemination within homogeneous communities. In addition, network assortativity can impact the stability and resilience of a network \cite{d2012robustness}. Thus, it is of great interest to analyze network assortativity in order to unravel mixing patterns on the network, predict the evolution behavior of the network, and design proper strategies for information diffusion on the network. 

As tried in the literature \cite{jiang2021applications}, the data involved in networks may contain sensitive information when carrying out the analysis of network assortativity, which potentially leads to privacy breaches. For example, the network in question is a social one, where friendships and the number of friends used to calculate its assortativity are deemed sensitive information for each individual. Therefore, it is necessary to develop a framework that is suitable for the analysis of network assortativity and provides a guarantee for protecting the privacy of individuals within the network simultaneously. Assortativity is structurally sensitive and involves joint statistics over node degrees and their connectivity, making its estimation under strong privacy constraints particularly challenging. To the best of our knowledge, there is no such scheme by far. 
To bridge this gap, we propose available approaches based on differential privacy in this paper.

Differential Privacy (DP) \cite{dwork2006differential} has become the gold standard for privacy-preserving network analysis, with broad applications in degree distribution estimation \cite{hay2009accurate, day2016publishing, ma2024ptt}, triangle and k-star counting \cite{imola2022communication, imola2022differentially}. While the traditional Central DP (CDP) model offers strong utility by assuming a trusted data curator, it poses significant risks under server compromise or insider threats \cite{yan2021monitoring, edmonds2020power}. To address this limitation, Local DP (LDP) shifts the trust boundary to individual users by requiring them to perturb their data before submission \cite{erlingsson2014rappor}. On the other hand, this is achieved at the cost of high noise and poor utility, especially for graph statistics relying on edge correlations or higher-order structure \cite{kasiviswanathan2013analyzing}.

To balance privacy and utility, intermediate models such as Shuffle DP \cite{cheu2019distributed, bittau2017prochlo} and Decentralized DP (DDP) have been proposed. Shuffle DP enhances LDP by anonymizing locally randomized reports via an intermediate server called shuffler, enabling amplified privacy guarantees and improved estimation accuracy \cite{erlingsson2019amplification}. DDP, firstly introduced by Sun et al. \cite{sun2019analyzing}, is a decentralized model tailored for network settings where each user holds an extended local view (e.g., 1-hop or 2-hop subgraph). Users apply randomized algorithms to their local views and send only privatized summaries to the server. This allows more expressive and accurate network analysis than standard LDP and, at the same time, preserves strong user-level privacy.

Motivated by this, the goal of this work is to investigate the private estimation of the assortativity coefficient—a critical metric measuring degree correlation between nodes—under LDP, Shuffle DP, and the DDP framework.  The main contribution of this work is as follows:
\begin{itemize}
    \item  We are the first to study the estimation of network assortativity under differential privacy. 
    \item  We propose two DP-algorithms $\mathbf{Local_{ru}}$ and $\mathbf{Shuffle_{ru}}$ for accurate estimation of network assortativity in the setting where each individual is assumed to know only his/her friends. In a mathematically rigorous manner, we prove that the two algorithms output unbiased estimator, and also determine their steadiness by calculating MSE. In addition, we analyze both time and space complexity of the proposed algorithms.   
    \item We also propose another DP-algorithm $\mathbf{Decentral_{ru}}$ suitable for a common situation in which each individual has a broader view, i.e., also knowing his/her friends' friends. Accordingly, we study the unbiasedness, MSE, time and space complexity related to this algorithm. The result shows that in sparse networks, the consequence produced by $\mathbf{Decentral_{ru}}$ is overwhelmingly better than that by $\mathbf{Local_{ru}}$. This implies that a broader view is helpful to obtain a much steadier estimation.
    \item We conduct experimental evaluation on both synthetic and real-world datasets, which demonstrates that empirical analysis is in good agreement with theoretical consequences. In the meantime, $\mathbf{Shuffle_{ru}}$ shows better performance than the other two algorithms due to its own privacy amplification.
\end{itemize}

The rest of this paper is organized as below. We showcase the related work in Section 2. Section 3 introduces terminologies and notations including graph, assortativity coefficient, differential privacy and shuffle model. The main results, namely, the proposed algorithms, are shown in Section 4. We conduct extensive experiments, and the results are displayed in Section 5. Some future research directions are listed in Section 6. Lastly, we close this work in Section 7.

\section{Related Work}

Here, we briefly introduce related work. Roundly speaking, it contains two parts, i.e., non-private network assortativity analysis and private graph statistics. Note that two terms network and graph are used interchangeablely in the remainder of this work.   

\subsection{Non-private network assortativity analysis}
Assortativity, an important network property, has been extensively studied since its introduction \cite{newman2002assortative}. In \cite{newman2002assortative}, Newman introduced the concept of assortativity in networks in 2002 and proposed the assortativity coefficient as a measure for structure mixing of networks. By applying this metric to various real-world networks, he showed that social networks tend to be assortative, whereas technological and biological networks tend to be disassortative. In 2003, Newman \cite{newman2003mixing} further analyzed mixing patterns in networks, distinguished between assortative and disassortative mixing, and explored their effects on network structure and dynamics. In the follow-up work, Piraveenan et al. \cite{piraveenan2008local} investigated the relationship between assortativity and the information content of networks, discovering that assortative and disassortative networks contain more information than neutral networks, which shows neither assortativity nor disassortativity. Chang et al. \cite{chang2007assortativity} focused on assortativity in weighted networks by defining the node's strength as the sum of the weights of its connecting edges. Leung et al. \cite{leung2007weighted} also studied the assortativity of weighted networks and introduced the concept of weighted assortativity. Estrada et al. \cite{estrada2011combinatorial} presented a method to determine network assortativity by examining the relationships among the clustering coefficient, modular connectivity, and branching. The results indicate that the clustering coefficient and modular connectivity positively influence assortativity, while branching has a negative impact.


\subsection{Private graph statistics}
Graph analysis under the differential privacy framework has been extensively studied, covering areas such as degree distribution \cite{hay2009accurate,day2016publishing,ma2024ptt}, subgraph counting \cite{karwa2011private,kasiviswanathan2013analyzing}, and synthetic graph generation \cite{sala2011sharing,xiao2014differentially}. These studies typically operate within either a centralized or a local model. The centralized model assumes a centralized data curator, posing risks of data leakage. Consequently, the local model (LDP) has become increasingly favored by researchers.

In recent years, there has been a surge in research on Local Differential Privacy (LDP) for graph data, with an increasing number of notable advancements. For example, Sun et al. \cite{sun2019analyzing} proposed a subgraph counting algorithm under the assumption that each user knows all friends of his/her friends. Qin et al. \cite{qin2017generating} devised a method for generating synthetic graphs based on LDP. Ye et al. \cite{ye2020towards} presented a local one-round algorithm to estimate graph metrics including the clustering coefficient. Zhang et al. \cite{zhang2020differentially} developed a software usage analysis algorithm under LDP. Imola et al. \cite{imola2022differentially} proposed an exact triangle and 4-cycle counting algorithm under LDP by introducing the Shuffle model.

However, to date, there has been no research conducted on assortativity analysis within the context of privacy protection. Hence, the goal of this work is to present assortativity analysis algorithms based on Differential Privacy which is widely recognized as a robust privacy protection framework.

\section{Preliminary}

In this section, we introduce some preliminaries for our work. Subsection 3.1 defines some basic notations related to networks. Subsections 3.2, 3.3 and 3.4 introduce the assortivity coefficient of networks, DP on graph and the shuffle model, respectively. Subsection 3.5 shows the utility metrics used in this work.

\subsection{Notations}
Let $\mathbb{N}$, $\mathbb{R}$ and $\mathbb{R}_{\ge 0}$ be the sets of natural numbers, real numbers and non-negative
real numbers, respectively. We consider an undirected, non-weighted graph $G=\left( V,E\right)$, where $V$ represents a set of nodes (users) and $E\subseteq V\times V$ represents a set of edges. Let $n\in \mathbb{N}$ be the number of nodes (users) in $G$, and $M\in \mathbb{N}$ be the number of edges in $G$. We use $v_i$ to represent the $i$-th node (user), and then have $V=\left\{ v_1,v_2,\dots,v_n\right\}$. Let $d_i\in \mathbb{N}$ be the degree of node $v_i$ (i.e. the number of edges connected to $v_i$), and $d_{max}\in \mathbb{N}$ be the maximum degree of nodes in $G$, i.e., $d_{max}=\text{max}\left\{ d_1,\dots,d_n\right\}$. We denote by $\mathcal{G}$ the set of graphs with $n$ nodes.  It is a convention to represent a graph $G$ using the adjacency matrix $\mathbf{A}=\left( a_{ij}\right) \in \left\{ 0,1\right\} ^{n\times n}$ in which if $\left( v_i,v_j\right) \in E$, then $a_{ij}=1$; and $a_{ij}=0$ otherwise. Accordingly, the $i$-th row of $\mathbf{A}$ is $\mathbf{a_i}=(a_{i1},a_{i2},...,a_{in})$. Table \ref{tab1:notation} shows some basic notations used in this paper.

\begin{table}[h]
    \centering
    \caption{Basic notations}
    \label{tab1:notation}
    \renewcommand\arraystretch{1.2}
    \begin{tabular}{c|l}
        \hline
        Symbol & Description \\
        \hline
        $G$ & Undirected, non-weighted graph with $n$ nodes. \\
        $M$ & Number of edges in $G$. \\
        $\mathbf{A}$ & Adjacency matrix corresponding to graph $G$. \\
        $v_i$ & $i$-th node (user) in $G$. \\
        $\mathbf{a}_i$ & Neighbor list of $v_i$ (i.e., the $i$-th row of $\mathbf{A}$). \\
        $d_i$ & Degree of $v_i$. \\
        $d_{max}$ & Maximum degree of nodes in $G$. \\
        $d_{avg}$ & Average degree of nodes in $G$. \\
        $q_{ru}$ & Assortativity factor query of $G$. \\
        \hline
    \end{tabular}
\end{table}

We adhere to standard notation conventions in our study. Concretely speaking, we employ the normal font $\theta$ to represent the true value of a statistic (e.g., $d_i$ for the true degree of node $v_i$) and use the tilde symbol $\tilde{\theta}$ to indicate the noisy version of the quantity (e.g., $\tilde{d}_i$ is the noisy degree of $v_i$ after applying some operation, for example, the Laplace mechanism). Additionally, the hat symbol $\hat{\theta}$ is used to signify the estimated value (e.g., $\hat{d}_{i}$ is the estimation of degree of node $v_i$ in $G$).

\subsection{Assortativity Coefficient}
Assortativity is a significant property of networks, measuring the preference of nodes to attach to others that are alike in some way. To quantitatively study network assortativity, Newman proposed an index $r$ called assortativity coefficient \cite{newman2003mixing}. Specifically, the concept of assortativity coefficient is based on Pearson correlation coefficient, which is defined as
\begin{equation}\label{eq:3-2-1}
    r=\frac{M^{-1}\Sigma_{e_{ij} \in E}{d_id_j}-\left[ M^{-1}\Sigma_{e_{ij} \in E}{\frac{1}{2}\left( d_i+d_j \right)} \right] ^2}{M^{-1}\Sigma_{e_{ij} \in E}{\frac{1}{2}\left( d_{i}^{2}+d_{j}^{2} \right)}-\left[ M^{-1}\Sigma_{e_{ij} \in E}{\frac{1}{2}\left( d_i+d_j \right)} \right] ^2},
\end{equation}
where $e_{ij}$ represents an edge between vertices $v_i$ and $v_j$. If $r>0$, the network in question is referred to as assortative. In this setting, it is believed that on average, nodes with similar degrees tend to connect to each other. If $r<0$, the network is considered disassortative, where nodes with different degrees are more likely to connect to one another. If $r=0$, it indicates that nodes in the network are connecting randomly, without any preference for nodes with similar or dissimilar degrees.

Ma et al. have proven that the denominator in Eq.(\ref{eq:3-2-1}) consistently remains non-negative \cite{mafei2024}. This implies that whether a network exhibits assortativity or disassortativity depends solely on the value of the numerator, denoted by $r_u$ for convenience, in Eq.(\ref{eq:3-2-1}). Namely, we write 
\begin{equation}\label{eq:3-2-2}
   r_u=M^{-1}\Sigma_{e_{ij} \in E}{d_id_j}-\left[ M^{-1}\Sigma_{e_{ij} \in E}{\frac{1}{2}\left( d_i+d_j \right)} \right] ^2.
\end{equation}
Also, we refer to $r_u$ as the assortativity factor. As shall be shown below, we focus on how to compute the accurate estimate of $r_u$ under DP. Let $q_{ru}:\mathcal{G}\rightarrow \left[ -1,1\right]$ be a assortativity factor query that takes $G\in \mathcal{G}$ as input and outputs the assortativity factor $q_{ru}(G)$ of $G$.

\subsection{Differential Privacy}
\textbf{Local Differential Privacy.} Local Differential Privacy (LDP) \cite{duchi2013local} is a privacy measure that protects the personal information of each user locally. Due to its ability to protect sensitive information without relying on trusted third-party servers, LDP has garnered widespread attention in the field of network analysis \cite{qin2017generating,wei2020asgldp,imola2021locally}. In LDP, each user first obfuscates his/her personal data by himself/herself and sends the obfuscated data to a data collector. 

LDP on graphs includes edge LDP and node LDP \cite{qin2017generating}. The former conceals the presence of any edge in a graph, while the latter hides both a node and all its adjacent edges. Node LDP is a stronger privacy definition but is more challenging to implement since it requires algorithms to hide more information about the input graph. This paper focuses on edge LDP because it can provide sufficient protection in many scenarios including subgraph counting \cite{imola2021locally}, synthetic graphs \cite{qin2017generating} and maximizes the availability of perturbed data. Additionally, assortativity represents the likelihood of forming edges between nodes with similar attributes, and the primary focus of its study is the network’s edges. In other words, its calculation only requires edge information. From this perspective, edge LDP is sufficient.

\begin{definition}[$\left(\varepsilon,\delta\right)$-edge LDP \cite{qin2017generating}]
    Let $n\in\mathbb{N}$, $\varepsilon\in\mathbb{R}_{\ge 0}$, and $\delta\in\left[ 0,1\right]$. For $i\in\left\{ 1,2,\dots,n\right\}$, let $\mathcal{R}_i$ be a local randomizer of user $v_i$ that takes $\mathbf{a_i}$ as input. $\mathcal{R}_i$ provides $\left( \varepsilon,\delta \right)$-edge LDP if for any two neighbor lists $\mathbf{a}_i,\mathbf{a}_i^\prime \in \left\{ 0,1\right\} ^n$ that differ in one bit and any $S\subseteq \mathrm{Range}\left( \mathcal{R}_i\right)$, we have 
    \begin{equation}\label{eq:3-3-1}
        \mathrm{Pr}\left( \mathcal{R}_i\left( \mathbf{a}_i\right) \in S\right) \le e^\varepsilon \mathrm{Pr}\left( \mathcal{R}_i\left( \mathbf{a}_i^\prime \right) \in S\right) + \delta.
    \end{equation}
    \label{def:edge LDP}
\end{definition}
The parameter $\varepsilon$ is referred to as the privacy budget, which reflects the level of privacy protection offered by algorithm $\mathcal{R}_i$. When $\varepsilon$ is small (e.g. $\varepsilon \le 1$ \cite{li2017differential}), each bit of $\mathbf{a}_i$ is strongly protected by edge LDP. The parameter $\delta$ represents the privacy failure probability and is typically set to a value much smaller than $\frac{1}{n}$ \cite{dwork2014algorithmic}. If $\delta=0$, $\mathcal{R}_i$ provides $\varepsilon$-edge LDP. 

\textbf{Randomized Response (RR).} Randomized response is the mainstream obfuscation mechanism in the study of categorical data under LDP. The classic Warner's RR (Randomized Response) \cite{warner1965randomized} is defined as follows. Applying Warner's RR to neighbor lists provides $\varepsilon$-edge LDP. 
\begin{definition}[Randomized Response \cite{warner1965randomized}]
    Given $\varepsilon \in \mathbb{R}_{\ge 0}$, the Randomized Response Mechanism $\mathcal{R}_\varepsilon^W:\left\{ 0,1\right\} \rightarrow \left\{ 0,1\right\} $ maps $x\in \left\{ 0,1\right\} $ to $y\in \left\{ 0,1\right\}$ with the probability
    \begin{align}
        \mathrm{Pr}\left( \mathcal{R}_{\varepsilon}^{W}\left( x \right) =y \right) =\left\{ \begin{matrix}
	   \frac{e^{\varepsilon}}{e^{\varepsilon}+1}&		\,\,\text{if\,\,\,}x=y,\\
	   \frac{1}{e^{\varepsilon}+1}&		\text{otherwise}.\\
    \end{matrix} \right. 
    \end{align}
    \label{def:RR}
\end{definition}

Note that two users in a network share information about the existence of an edge between them. That is to say, both users' outputs may reveal sensitive data. Based on this, the following algorithms (which will be proposed in this work) are designed for the lower triangular part of the adjacency matrix $\mathbf{A}$ to enhance the level of user privacy protection. Specifically, given $\varepsilon \in \mathbb{R}_{\ge 0}$ and a neighbor list $\mathbf{a}_i\in \left\{ 0,1\right\} ^n$, the local randomizer $\mathcal{R}_i$ outputs noisy bits $\left( \tilde{a}_{i,1},\dots,\tilde{a}_{i,i-1}\right) \in \left\{ 0,1\right\} ^n$ for users with
smaller IDs, i.e., for each $j\in \left\{ 1,\dots,i-1\right\}$, $\tilde{a}_{i,j}=1-a_{i,j}$ with probability $p=\frac{1}{e^\varepsilon +1}$ and $\tilde{a}_{i,j}=a_{i,j}$ with probability $1-p$.

\textbf{Laplacian Mechanism.} The Laplace mechanism \cite{dwork2014algorithmic} is a data obfuscation method introduced by Dwork et al for numerical data to guarantee differential privacy. The mechanism protects privacy by injecting random noise independently sampled from a Laplace distribution into the query statistic, where the level of noise added depends on the global sensitivity of the statistic. The greater the global sensitivity of the statistic, the more noise is added for privacy preservation.
\begin{definition}[Global Sensitivity under LDP \cite{dwork2014algorithmic}]
    In edge LDP, the global sensitivity of a function $f:\left\{ 0,1\right\} ^n\rightarrow \mathbb{R}$ is given by
    \begin{align}
        \Delta_f = \max _{\mathbf{a}_{i}, \mathbf{a}_{i}^{\prime} \in\{0,1\}^{n}, \left| \mathbf{a}_{i}-\mathbf{a}_{i}^{\prime}\right| _1=1}\left|f\left(\mathbf{a}_{i}\right)-f\left(\mathbf{a}_{i}^{\prime}\right)\right|,
    \end{align}
where $\left| \mathbf{a}_{i}-\mathbf{a}_{i}^{\prime}\right| _1=1$ represents that $\mathbf{a}_{i}$ and $\mathbf{a}_{i}^{\prime}$ differ in one bit.
\label{def:global sensitivity}
\end{definition}
\begin{definition}[Local Laplacian Mechanism \cite{dwork2014algorithmic}]
    Given $\varepsilon \in \mathbb{R}_{\ge 0}$ and any query function $f:\left\{ 0,1\right\} ^n\rightarrow \mathbb{R}$ with a global sensitivity $\Delta_f$, the Local Laplacian Mechanism is defined as
    \begin{align}
        \mathcal{R}_\varepsilon^L\left( f\left( \mathbf{a}_i\right)\right) =f\left( \mathbf{a}_i\right) +\mathrm{Lap} \left( \frac{\Delta_f}{\varepsilon}\right).
    \end{align}
    \label{def:local Lap}
\end{definition}

\textbf{Extended Local Views (ELV).} In the LDP mechanism, each user's local input is his/her neighbor list $\mathbf{a}_i$. However, in many real-world scenarios, users typically have a broader local view. For example, in social networks like Twitter and LinkedIn, even with privacy settings enabled, the number of friends of a user's friends is usually accessible. In this paper, in addition to the user's 1-hop local view (i.e., neighbor list) we also focus on the 2-hop Extended Local View, which is defined as follows.

\begin{definition}[2-hop Extended Local View \cite{sun2019analyzing}]
    Given a node $v_i$, its 2-hop Extended Local View (ELV) consists of all nodes reachable within 2 hops and all corresponding edges.
    \label{def:ELV}
\end{definition}

Conducting local differential privacy graph statistics in a local setting with 2-hop ELVs can significantly increase privacy risks due to the substantial overlap between different users' ELVs. Specifically, if ELVs from different users overlap, a single edge may contribute to multiple local reports. In other words, inserting/removing an edge in the global graph would result in changes to multiple local reports. Consequently, the presence of an edge is revealed multiple times, leading to an increased privacy risk.

\textbf{Decentralized Differential Privacy.} To address this privacy risk in the ELV setting, Sun et al. \cite{sun2019analyzing} introduced Decentralized Differential Privacy (DDP), which requires each user to protect not only his/her own privacy but also that of his/her neighbors. DDP is a privacy concept that lies between CDP and LDP. Specifically, DDP follows CDP to define neighboring datasets on the entire graph, but it protects each report locally during data collection. Formally, edge DDP is defined as follows.

\begin{definition}[$\left(\varepsilon,\delta\right)$-edge DDP \cite{sun2019analyzing}]
    Let $n\in\mathbb{N}$, $\varepsilon\in\mathbb{R}_{\ge 0}$, and $\delta\in\left[ 0,1\right]$.  A set of local randomizers $\left\{ \mathcal{R}_1, 1\le i\le n\right\} $ provides $\left( \varepsilon,\delta \right) $-edge DDP if for any two neighboring graphs $G,G^\prime \in \mathcal{G}$ that differ in one edge and any $\left\{ S_i\in \mathrm{Range}\left( \mathcal{R}_i\right) , 1\le i\le n\right\}$,
    \begin{equation}\label{eq:DDP}
        \begin{aligned}
        &\quad \mathrm{Pr}\left( \left( \mathcal{R}_1\left( G_1\right) \in S_1\right) , \dots, \left( \mathcal{R}_n\left( G_n\right) \in S_n\right) \right) \\
        &\le e^\varepsilon \mathrm{Pr}\left( \left( \mathcal{R}_1\left( G_1^\prime \right) \in S_1\right) , \dots, \left( \mathcal{R}_n\left( G_n^\prime \right) \in S_n\right) \right) + \delta.
        \end{aligned}
    \end{equation}
    where $G_i$ (resp. $G_i^\prime$) is the ELV of user $v_i$ in graph $G$ (resp. $G^\prime$).
    \label{def:edge DDP}
\end{definition}

Since DDP protects an edge in the entire graph $G$, the global sensitivity under DDP is defined as

\begin{definition}[Global Sensitivity under DDP \cite{sun2019analyzing}]
    Given a function $f$, the global sensitivity of $f$ is defined as
    \begin{equation}\label{eq:GS under DDP}
        \begin{aligned}
        GS_{DDP}\left( f\right) =\max _{G,G^\prime }{\sum_{i=1}^n{\lvert f(G_i)-f(G_i^\prime)\rvert}},
        \end{aligned}
    \end{equation}
    where $G$ and $G^\prime$ are two arbitrary neighboring graphs, and $G_i$ (resp. $G_i^\prime$) is the ELV of user $v_i$ in graph $G$ (resp. $G^\prime$).
    \label{def:GS under DDP}
\end{definition}

The global sensitivity of a function is determined by the function itself, and some functions may have large global sensitivities. For example, in the 2-hop ELV setting, the global sensitivity of the sum of the degrees of a user's friends under DDP is $4n - 6$. Since $n$ can be very large in real-world graphs (e.g., the number of users in social networks), directly perturbing data using the global sensitivity as the noise scale can significantly compromise the accuracy of the estimates. To improve accuracy, local sensitivity is typically used to reduce the noise scale, which is defined as follows.

\begin{definition}[Local Sensitivity under DDP \cite{sun2019analyzing}]
    Given a global graph $G$ and a function $f$, the local sensitivity of $f$ is defined as
    \begin{equation}\label{eq:LS under DDP}
        \begin{aligned}
        LS_{DDP}\left( f\right) =\max _{G^\prime }{\sum_{i=1}^n{\lvert f(G_i)-f(G_i^\prime)\rvert}},
        \end{aligned}
    \end{equation}
    where $G^\prime$ is an arbitrary neighboring graph of $G$ and $G_i$ (resp. $G_i^\prime$) is the ELV of user $v_i$ in graph $G$ (resp. $G^\prime$).
    \label{def:LS under DDP}
\end{definition}

Since local sensitivity is data-dependent, to prevent it from compromising privacy, we derive a probabilistic upper bound for local sensitivity by Lemma \ref{lemma:TailBound}, and use this upper bound as the noise scale for data perturbations.

\begin{lemma}[Tail Bound for Laplace Distribution \cite{sun2019analyzing}]
    Let $x$ be any real value, and $\tilde{x} = x+\mathrm{Lap}(\alpha )$ for some $\alpha >0$. Then, with probability $1-\delta$,
    \begin{equation}
        \tilde{x}+\alpha \cdot \log \left( \frac{1}{2\delta}\right) \ge x
    \end{equation}
    \label{lemma:TailBound}
\end{lemma}


\subsection{Shuffle Model}
In recent studies \cite{balle2019privacy,cheu2019distributed,erlingsson2019amplification}, the shuffle model has garnered significant interest due to its privacy amplification effect. The shuffle model, also known as Shuffle DP (SDP for short), introduces an intermediary server called shuffler between the user and the data collector based on LDP, and it works as follows. First, each user $v_i$ obfuscates his/her personal data using an LDP mechanism $\mathcal{R}$ common to all users. Then, user $v_i$ encrypts the obfuscated data and sends it to the shuffler. The shuffler randomly shuffles the encrypted data to ensure anonymization and sends the result to a data collector for which we make no trust assumptions. Finally, the data collector decrypts the shuffled data. Under the assumption that the data collector does not collude with the shuffler, the shuffle model amplifies DP guarantees of the obfuscated data because shuffling removes the link between users and the obfuscated data. Shuffling can amplify privacy without loss of data utility for analytical tasks that are insensitive to data order, such as sum, average and histogram queries. Current research on the shuffle model primarily focuses on determining bounds on the level of privacy obtained after shuffling \cite{girgis2021renyi,feldman2022hiding}. 

In this paper, we use the privacy amplification bound provided by Feldman et al., which is the current state of the art.
\begin{lemma}[Privacy Amplification by Shuffling \cite{feldman2022hiding}]
    Let $n\in \mathbb{N}$ and $\varepsilon_0\in \mathbb{R}_{\ge 0}$. Let $\mathcal{X}$ be the set of input data for each user. Let $x_i\in \mathcal{X}$ be input data of the $i$-th user and $\mathbf{x}=\left( x_1,\dots,x_n\right) \in \mathcal{X}^n$. For $i\in \left\{ 1,2,\dots,n\right\}$, let $\mathcal{R}_i:\mathcal{X} \rightarrow \mathcal{Y}$ be a local randomizer of the $i$-th user that provides $\varepsilon_0$-LDP. Let $\mathcal{A}_S:\mathcal{X}^n \rightarrow \mathcal{Y}^n$ be the algorithm that given a dataset $\mathbf{x} \in \mathcal{X}^n$, samples a uniform random permutation $\pi$ over $\left\{ 1,2,\dots,n\right\}$, then sequentially computes $y_{\pi \left( i\right)}=\mathcal{R}_i\left( x_{\pi \left( i\right)}\right)$ and outputs $\mathbf{y}=\left( y_{\pi \left( 1\right)},\dots,y_{\pi \left( n\right)}\right)$. Then for any $\delta \in \left[ 0,1\right]$ such that $\varepsilon_0 \le \log \left( \frac{n}{16\log \left( 2/\delta \right)}\right)$, $\mathcal{A}_S$ provides $\left( \varepsilon,\delta \right)$-DP, where
    \begin{align}
        \varepsilon \le \log \left( 1+\frac{e^{\varepsilon _0}-1}{e^{\varepsilon _0}+1}\left( \frac{8\sqrt{e^{\varepsilon _0}\log \left( 4/\delta \right)}}{\sqrt{n}}+\frac{8e^{\varepsilon _0}}{n} \right) \right).
        \label{eq:shuffle}
    \end{align}
    \label{lemma:shuffle}
\end{lemma}

In the lemma above, the shuffled data $y_{\pi(1)}, \dots, y_{\pi(n)}$ is protected by $\varepsilon$-DP, where $\varepsilon \ll \varepsilon_0$. In addition to providing this closed-form upper bound, Feldman et al. \cite{feldman2022hiding} developed an efficient algorithm to numerically compute a tighter upper bound, which we refer to as the numerical upper bound. In this paper, we use both the closed-form and numerical upper bounds and also distinguish the impact of each term on the performance of the algorithm.

\subsection{Utility Metrics}
We use MSE and Relative Error as utility metrics in theoretical analysis and experimental evaluation, respectively. Let $\hat{q}_{ru}:\mathcal{G}\rightarrow \mathbb{R}$ be the assortativity factor estimator. Let $MSE:\mathbb{R}\rightarrow \mathbb{R}_{\ge 0}$ be the MSE function, which maps the estimate $\hat{q}_{ru}\left( G\right)$ and the true value $q_{ru}\left( G\right)$ to the MSE; i.e., 
\begin{align}
        MSE\left[ \hat{q}_{ru}\left( G\right),q_{ru}\left( G\right)\right] =\mathbb{E}\left[ \left( \hat{q}_{ru}\left( G\right)-q_{ru}\left( G\right) \right)^2 \right]. 
        \label{eq:MSE}
    \end{align}
Note that the MSE can also be large when the $q_{ru} \left( G\right)$ is large. Therefore, in our experiments, we evaluate the relative error given by 
\begin{align}
        RE\left[ \hat{q}_{ru}\left( G\right),q_{ru}\left( G\right)\right] =\frac{ |\hat{q}_{ru}\left( G\right)-q_{ru}\left( G\right)|}{\min\{q_{ru}\left( G\right),\eta\}},
        \label{eq:RE}
    \end{align}
where $\eta\in \mathbb{R}_{\ge 0}$ is a small positive value. In general, $\eta$ is set to $\frac{n}{10^3}$ \cite{xiao2011ireduct, chen2012differentially, bindschaedler2016synthesizing}. If the relative error is much smaller than 1, the estimation is highly accurate.

As shown in Eq.(\ref{eq:3-2-1}), it is clear that the estimation of the assortativity coefficient $r$ involves intricate analyses of some terms reflecting inter-node correlations, which poses a challenge to obtaining an accurate estimate of $r$ under the requirement of privacy protection. Specifically, bias correction is required for the expression $X/Y$ derived from locally differentially private data to obtain an unbiased privacy estimate of the assortativity coefficient. Here, both $X$ and $Y$ denote a random polynomial. Despite recent advances in algorithms for unbiased local privacy estimation of Laplacian variable polynomials \cite{hillebrand2023unbiased}, it remains unresolved to obtain an unbiased estimate of $X/Y$ from locally differentially private data.


The next lemma (i.e., lemma \ref{lemma:expectation of the ratio}) offers an approach for approximating the unbiased privacy estimate of $r$. It is clear that we can approximate the unbiased estimate of the assortativity coefficient $r$ by calculating the ratio of the unbiased estimates of its numerator and denominator. Of course, a more accurate unbiased estimate of $r$ can be approximated by Eq.(\ref{subeq:3-5-1-1}). As shall be seen below, the following method is suitable for a more accurate estimation of $r$ while many tedious computations need to be performed. This is because the expression for $r$ involves a large number of terms representing connections between nodes. In this paper, we use a simple approximate estimator for $r$. Note that the denominator of $r$ (denoted by $r_d$) is
\begin{equation}
    \begin{aligned}
        r_d &= M^{-1}\Sigma_{e_{ij} \in E}{\frac{1}{2}\left( d_{i}^{2}+d_{j}^{2} \right)}-\left[ M^{-1}\Sigma_{e_{ij} \in E}{\frac{1}{2}\left( d_i+d_j \right)} \right] ^2 \\
        & = \frac{1}{2M}\sum_{i=1}^n{d_i^3}-\left[ \frac{1}{2M} \sum_{i=1}^n{d_i^2}\right]^2
    \end{aligned}
\end{equation}
Clearly, the denominator $r_d$ is a polynomial function of node degrees, and its locally differentially private unbiased estimator can be obtained using the debiasing algorithm proposed in \cite{hillebrand2023unbiased}. Therefore, this paper focuses on computing a stable, unbiased estimator of the assortative factor $r_u$ (numerator) under LDP.

\begin{lemma}[\cite{casella2024statistical}]
    Let $X$ and $Y$ be two random variables,
\begin{subequations}
\label{eq:whole}
\begin{eqnarray}
 \mathbb{E}\left( \frac{X}{Y} \right) \approx \frac{\mathbb{E}\left( X \right)}{\mathbb{E}\left( Y \right)}-\frac{\mathrm{Cov}\left( X,Y \right)}{\left[ \mathbb{E}\left( Y \right) \right] ^2}+\frac{\mathbb{E}\left( X \right)}{\left[ \mathbb{E}\left( Y \right) \right] ^3}\cdot \mathbb{V}\left( Y \right) ,\label{subeq:3-5-1-1}
\end{eqnarray}
\begin{equation}
\mathbb{V}\left( \frac{X}{Y} \right) \approx \frac{\mathbb{V}\left( X \right)}{\left[ \mathbb{E}\left( Y \right) \right] ^2}-\frac{2\mathbb{E}\left( X \right) \mathrm{Cov}\left( X,Y \right)}{\left[ \mathbb{E}\left( Y \right) \right] ^3}+\frac{\left[ \mathbb{E}\left( X \right) \right] ^2}{\left[ \mathbb{E}\left( Y \right) \right] ^4}\mathbb{V}\left( Y \right). \label{subeq:3-5-1-2}
\end{equation}
\end{subequations}
    \label{lemma:expectation of the ratio}
\end{lemma}

\section{Schemes}

This section elaborates on our main results. Concretely, we propose three algorithms for estimating the assortativity factor $r_u$ of network on $n$ nodes in the requirement of DP. 

First of all, we consider the common local setting where each individual is assumed to know only his/her friends. Accordingly, two algorithms $\mathbf{Local_{ru}}$ and $\mathbf{Shuffle_{ru}}$ are presented. $\mathbf{Local_{ru}}$ is built based on a simple assumption that there is only one round of interaction between each user and the data collector. In this model, each user (data provider) perturbs his/her data once and then sends it to the data collector. Due to the limited interaction, the data collector can just gather a restricted amount of information. Consequently, the one-round algorithm is generally suitable for simple data analysis tasks where high accuracy is not critical. In order to obtain more accurate consequences, we add a round of communication between users and the data collector. Put this another way, a two-round algorithm $\mathbf{Shuffle_{ru}}$ is designed. Specifically, the data collector sends query requests to each user two times, allowing for more powerful queries. As shown in the previous works, such as more accurate subgraph counting \cite{sun2019analyzing,imola2021locally}, synthetic graph construction \cite{qin2017generating}, a similar thought has been employed to design some two-round algorithms that can handle more complex data analysis tasks and achieve greater precision. However, this increased accuracy comes at the cost of added complexity and higher communication overhead. 

Secondly, we also discuss the other scenario in which each individual has a broader view, i.e., also knowing his/her friends' friends. For instance, with the default setting of Facebook (facebook.com), a user allows each of her friends to see all her connections. The third algorithm $\mathbf{Decentral_{ru}}$ is proposed to get assortativity factor $r_u$ in this situation. 

For convenience and understanding, we summarize the main properties of these three algorithms (i.e., $\mathbf{Local_{ru}}$, $\mathbf{Shuffle_{ru}}$ and $\mathbf{Decentral_{ru}}$) in the next table ahead of time. 

\begin{table}[h]
    \centering
    \begin{threeparttable}
        \caption{main properties of algorithm}
        \label{tab2}
        \renewcommand\arraystretch{1.2}
        \begin{tabular}{|c|c|c|c|c|l}
            \hline
            &  $\mathbf{Local_{ru}}$ & $\mathbf{Shuffle_{ru}}$ & $\mathbf{Decentral_{ru}}$\\
            \hline
            $x$-DP& $\varepsilon$-LDP & $(\varepsilon,\delta)$-DP & $(\varepsilon,\delta)$-DDP   \\
            \hline
            Unbiasedness & yes & yes & yes \\
            \hline
            MSE & $ O\left( K_1 \right) $\tnote{1} & $O\left( K_2 \right) $\tnote{2} & $O\left( \frac{n^3d_{\max}^{6}}{M^4} \right) $  \\
            \hline
            Time complexity & $O(n^2)$& $O(n^2)$ & $O(n^2)$ \\
            \hline
            Space complexity & $O(n^2)$& $O(n^2)$ & $O(n^2)$  \\
            \hline
        \end{tabular}
        \begin{tablenotes}
            \item[1] $K_1=\frac{n^3d_{\max}^{2}+n^2d_{\max}^{4}}{M^2}+\frac{n^3d_{\max}^{6}}{M^4}$.
            \item[2] $K_2=\frac{n^{1+\alpha}d_{\max}^{4}}{M^2}+\frac{n^3d_{\max}^{2}}{\left( \log n \right) ^2M^2}+\frac{n^3d_{\max}^{6}}{\left( \log n \right) ^2M^4}$.
        \end{tablenotes}
    \end{threeparttable}
\end{table}

\emph{Remark} In the rest of this work, we assume that parameter $M$ is known to all. This is mainly because $M$ can be accurately calculated using well-established degree distribution estimation techniques \cite{day2016publishing, wei2020asgldp}. Doing so simplifies the following analysis as well. At the same time, it is clear that parameter $M$ has little influence on the estimation of network assortativity. It is worth noticing that the analysis approach established below is also suitable for the situation where parameter $M$ is unknown in advance. Nonetheless, this causes more tedious calculations when conducting the theoretical analysis, which becomes clear to understand based on the following analysis and lemma \ref{lemma:expectation of the ratio}. The goal of this paper is to build up a methodology that has wide applications. In a nutshell, we only consider the case where parameter $M$ is known.


Now, let us divert our attention to the development of the first algorithm $\mathbf{Local_{ru}}$. It should be noted that all the proofs of both unbiasedness and MSE of three algorithms are deferred to show in Supplementary Material for readability. Also, some notations are abused yet the meanings are clear from the context. 



\subsection{One-Round LDP Algorithm for Assortativity Factor}
As shown in Eq.\eqref{eq:3-2-2}, the calculation of the assortativity factor $r_u$ involves two terms: the degrees of nodes and the edges between them. In other words, we can derive $r_u$ from the degree sequence $d_1, d_2, \dots, d_n$ and the network's adjacency matrix $\mathbf{A}$. Based on this, we propose a one-round algorithm for $r_u$ under LDP. This algorithm adds random noise to the degrees of all nodes using the Laplace mechanism and obfuscates the lower triangular part of the adjacency matrix $\mathbf{A}$ with Randomized Response simultaneously. The data collector then estimates $r_u$ from these noisy data sent by each user.

\begin{algorithm}[h]
    \caption{$\mathbf{Local_{ru}}$. $\left[ \mathrm{v_i}\right] $ and $\left[ \mathrm{d}\right] $ represent that the process is run by user $v_i$ and the data collector, respectively.}\label{algo:1}
    \SetAlgoLined
    \KwData{Graph $G$ represented as neighbor lists $\mathbf{a}_1,\dots,\mathbf{a}_n$ $\in \left\{ 0,1\right\}^n$, privacy budget $\varepsilon_1,\varepsilon_2 \in \mathbb{R}_{\ge 0}$}
    \KwResult{Private estimate $\hat{q}_{ru}\left( G\right)$ of $q_{ru}\left( G\right)$}
   \BlankLine

    \For{$i=1$ \KwTo $n$}{
        $\left[ \mathrm{v_i}\right] $ $\tilde{\mathbf{R}}_i\leftarrow \left( \mathcal{R}_{\varepsilon_1}^W \left( a_{i,1}\right) ,\dots ,\mathcal{R}_{\varepsilon_1}^W \left( a_{i,i-1}\right) \right)$ \tcp*{Apply $RR$ to the lower triangular part of adjacency matrix $\mathbf{A}$.}
        $\left[ \mathrm{v_i}\right] $ $d_i \leftarrow \sum_{j=1}^n{a_{i,j}}$\;
        $\left[ \mathrm{v_i}\right] $ $\tilde{d}_i \leftarrow d_i+\mathrm{Lap}\left( \frac{1}{\varepsilon_2}\right)$\;
        $\left[ \mathrm{v_i}\right] $ Send $\tilde{\mathbf{R}}_i$ and $\tilde{d}_i $ to the data collector\;
    }
    
    $\left[ \mathrm{d}\right] $ $X_1\leftarrow \sum_{i=2}^n{\sum_{j=1}^{i-1}{\frac{\left( \tilde{a}_{i,j}-p \right) \tilde{d}_i\tilde{d}_j}{1-2p}}}$  \tcp*{$\tilde{a}_{i,j}$ is the output after applying $RR$ to obfuscate $a_{i,j}$}
    
    $\left[ \mathrm{d}\right] $ $Y_1\leftarrow \left( \frac{1}{2}\sum_{i=1}^n{\tilde{d}_{i}^{2}}-\frac{n+2}{\varepsilon _{2}^{2}} \right) ^2-\frac{5n+4}{\varepsilon _{2}^{4}}$\;

    $\left[ \mathrm{d}\right] $ $\hat{q}_{ru}\left( G\right)\leftarrow X_1/M+Y_1/M^2$\;
    
    $\left[ \mathrm{d}\right] $ \textbf{return} $\hat{q}_{ru}\left( G\right)$

\end{algorithm}

Algorithm 1 shows our one-round LDP algorithm for computing assortativity factor $r_u$. It takes a network $G$ (represented as neighbor lists $\mathbf{a}_1,\mathbf{a}_2,\dots,\mathbf{a}_n \in \left\{ 0,1\right\} ^n$) and privacy budgets $\varepsilon_1,\varepsilon_2\in \mathbb{R}_{\ge 0}$ as inputs, and outputs a privacy-preserving estimate of the assortativity factor query $q_{ru}$. We denote this algorithm by $\mathbf{Local_{ru}}$ for brevity. 

First, user $v_i$ uses the Randomized Response mechanism $\mathcal{R}_{\varepsilon_1}^W$ that provides $\varepsilon_1$-LDP to obfuscate $a_{i,1}, \ldots, a_{i,i-1}$ (i.e., the bits with smaller user IDs in the neighbor list $\mathbf{a}_i$) (Line 2). In other words, we apply RR to the lower triangular part of the adjacency matrix $\mathbf{A}$. Meanwhile, user $v_i$ calculates his/her degree $d_i$, i.e., the number of ``1"s in $\mathbf{a}_i$ (Line 3), and adds noise $\mathrm{Lap}\left( \frac{1}{\varepsilon_2}\right)$ to $d_i$ (Line 4), where $\mathrm{Lap}\left( \frac{1}{\varepsilon_2}\right)$ is a Laplace random variable with mean 0 and scale $\frac{1}{\varepsilon_2}$. Then, user $v_i$ sends the noisy data $\tilde{\mathbf{R}}_i$ and $\tilde{d}_i$ to the data collector. Finally, the data collector estimates the assortativity factor $q_{ru}\left( G\right)$ using $\tilde{\mathbf{R}}_1, \dots, \tilde{\mathbf{R}}_n$ and $\tilde{d}_1, \dots, \tilde{d}_n$. Specifically, the privacy estimate $\hat{q}_{ru}\left( G\right)$ of $q_{ru}\left( G\right)$ is as follows
\begin{equation}
    \hat{q}_{ru}\left( G \right) =\frac{X_1}{M}-\frac{Y_1}{M^2},
    \label{eq:est1}
\end{equation}
where
\begin{subequations}
\label{eq:whole}
\begin{eqnarray}
X_1=\sum_{i=2}^n{\sum_{j=1}^{i-1}{\frac{\left( \tilde{a}_{i,j}-p \right) \tilde{d}_i\tilde{d}_j}{1-2p}}},\label{subeq:4-1-1-1}
\end{eqnarray}
\begin{equation}
Y_1=\left( \frac{1}{2}\sum_{i=1}^n{\tilde{d}_{i}^{2}}-\frac{n+2}{\varepsilon _{2}^{2}} \right) ^2-\frac{5n+4}{\varepsilon _{2}^{4}}.\label{subeq:4-1-1-2}
\end{equation}
\end{subequations}

Next, we show some theoretical properties of $\mathbf{Local_{ru}}$. First, we prove that $\mathbf{Local_{ru}}$ certainly provides differential privacy. 
\begin{theorem}
    $\mathbf{Local_{ru}}$ provides $\varepsilon$-edge LDP and $2\varepsilon$-edge DDP, where $\varepsilon=\varepsilon_1+\varepsilon_2$.
    \label{theorem:3}
\end{theorem}
\begin{proof}
    $\mathbf{Local}_{ru}$ applies RR to the lower triangular part of the adjacency matrix $\mathbf{A}$, ensuring that $\tilde{\mathbf{R}}_1,\dots,\tilde{\mathbf{R}}_n$ are protected under $\varepsilon_1$-edge LDP. Meanwhile, since $\tilde{d}_i=d_i+\text{Lap}\left( \frac{1}{\varepsilon_2} \right)$, the perturbed degrees $\tilde{d}_1,\dots,\tilde{d}_n$ are protected under $\varepsilon_2$-edge LDP. According to the sequential composition property of differential privacy \cite{dwork2014algorithmic}, the combined sequence $\tilde{\mathbf{R}}_1,\dots,\tilde{\mathbf{R}}_n,\tilde{d}_1,\dots,\tilde{d}_n$ are protected under $(\varepsilon_1+\varepsilon_2)$-edge LDP. Since $\hat{q}_{ru}\left( G \right)$ is derived from post-processing $\tilde{\mathbf{R}}_1, \dots, \tilde{\mathbf{R}}_n$ and $\tilde{d}_1, \dots, \tilde{d}_n$, $\hat{q}_{ru}\left( G \right)$ is protected under $(\varepsilon_1 + \varepsilon_2)$-edge LDP by the post-processing immunity property of differential privacy \cite{dwork2014algorithmic}. In a word, $\mathbf{Local_{ru}}$ provides $\varepsilon$-edge LDP, where $\varepsilon=\varepsilon_1+\varepsilon_2$.

    Under the assumption that each individual knows only his/her friends, the local view $G_i$ of user $v_i$ is equivalent to $\mathbf{a}_i$, where $i=1,2,\dots,n$. Thus, according to the definition of DDP (see Definition \ref{def:edge DDP}), for any two neighboring graphs $G,G^\prime$ that differ in one edge, we have
    \begin{equation}
        \begin{aligned}
        &\quad \frac{Pr\left( \mathcal{R}_1\left( G_1 \right) \in S_1,...,\mathcal{R}_n\left( G_n \right) \in S_n \right)}{Pr\left( \mathcal{R}_1\left( G_{1}^{'} \right) \in S_1,...,\mathcal{R}_n\left( G_{n}^{'} \right) \in S_n \right)} \\
        &=\frac{\prod_{i=1}^n{Pr\left( \mathcal{R}_i\left( G_i \right) \in S_i \right)}}{\prod_{i=1}^n{Pr\left( \mathcal{R}_i\left( G_{i}^{'} \right) \in S_i \right)}} \\
        &=\frac{Pr\left( \mathcal{R}_x\left( \mathbf{a}_x \right) \in S_x \right) \cdot Pr\left( \mathcal{R}_y\left( \mathbf{a}_y \right) \in S_y \right)}{Pr\left( \mathcal{R}_x\left( \mathbf{a}_{x}^{'} \right) \in S_x \right) \cdot Pr\left( \mathcal{R}_y\left( \mathbf{a}_{y}^{'} \right) \in S_y \right)} \\
        &\le e^{2\varepsilon}.
        \end{aligned}
    \end{equation}
    Therefore, $\mathbf{Local_{ru}}$ provides $2\varepsilon$-edge DDP.
\end{proof}

The following theorem demonstrates that $\mathbf{Local_{ru}}$ provides an unbiased estimate of the assortativity factor.
\begin{theorem}
    The estimate $\hat{q}_{ru}\left( G\right)$ produced by $\mathbf{Local_{ru}}$ is unbiased, i.e., $\mathbb{E}\left( \hat{q}_{ru}\left( G\right)\right)=q_{ru}\left( G\right)$.
    \label{theorem:4}
\end{theorem}

As mentioned above, we also determine the steadiness of $\mathbf{local_{ru}}$, which is shown as follows. 
\begin{theorem}
    When $\varepsilon_1$, $\varepsilon_2$ are constants, the estimate $\hat{q}_{ru}\left( G\right)$ produced by $\mathbf{Local_{ru}}$ provides the following utility guarantee:
    \begin{align}
        \mathrm{MSE}\left( \hat{q}_{ru}\left( G \right) \right) &=O\left( \frac{n^3d_{\max}^{2}+n^2d_{\max}^{4}}{M^2}+\frac{n^3d_{\max}^{6}}{M^4} \right).
    \end{align}
    \label{theorem:5}
\end{theorem}
Since almost all real-world networks are sparse, i.e. $M=\Theta \left( n\right) $ \cite{albert2002statistical}, it follows that we have $\mathrm{MSE}\left( \hat{q}_{ru}\left( G \right) \right)\leq O\left( n^2\right)$, if quantity $d_{\max}\leq O( n^{1/2})$. It is well known that this condition $d_{\max}\leq n^{1/2}$ is commonly seen in real-world networks \cite{newman2018networks}. 

Finally, we analyze the time and space complexity of $\mathbf{Local_{ru}}$.
\begin{theorem}
    The time complexity of $\mathbf{Local_{ru}}$ is $O\left( n^2\right)$, and the space complexity is $O\left( n^2\right)$.
    \label{theorem:6}
\end{theorem}

$\mathbf{Local_{ru}}$ requires obfuscating each bit in the lower triangular part of the adjacency matrix of network $G$, resulting in a time complexity of $O\left( n^2\right)$. Additionally, this algorithm needs $O\left( n^2\right)$ space to store the adjacency matrix $\mathbf{A}$, hence the space complexity is also $O\left( n^2\right)$.

By far, we have finished the development of $\mathbf{Local_{ru}}$ and also analyzed it in detail. The next subsection aims to enhance the level of privacy-preserving by bringing the Shuffle Model.

\subsection{Two-Round DP Algorithm for Assortativity Factor with Shuffling}
Generally, multi-round LDP algorithms, which involve multiple rounds of interaction between each user and the data collector, can support more powerful queries and produce more accurate results. Hence, we further design a two-round algorithm for calculating the assortativity factor and, in the meantime, introduce the Shuffle Model to amplify privacy.


\begin{algorithm}[h]
    \caption{$\mathbf{Shuffle_{ru}}$. $\left[ \mathrm{v_i}\right] $, $\left[ \mathrm{s}\right]$ and $\left[ \mathrm{d}\right]$ represent that the process is run by user $v_i$, the shuffler and the data collector, respectively.}\label{algo:2}
    \SetAlgoLined

    \KwData{Neighbor lists $\mathbf{a}_1,\dots,\mathbf{a}_n$ $\in \left\{ 0,1\right\}^n$, privacy budget $\varepsilon \in \mathbb{R}_{\ge 0}$, failure probability $\delta \in \left[ 0,1\right]$, parameter for privacy budget allocation $\alpha \in \left( 0,1\right)$}
    \KwResult{Private estimate $\hat{q}_{ru}\left( G\right)$ of $q_{ru}\left( G\right)$}
    \BlankLine

    $\varepsilon_0 \leftarrow \text{LocalPrivacyBudget}(n,\varepsilon,\delta)$\;
    
    \tcc{First round}
    \For{$i=1$ \KwTo $n$}{
        $\left[ \mathrm{v_i}\right] $ $d_i \leftarrow \sum_{j=1}^n{a_{i,j}}$\;
        $\left[ \mathrm{v_i}\right] $ $\tilde{d}_i \leftarrow d_i+\mathrm{Lap}\left( \frac{1}{\alpha \varepsilon_0}\right)$\;
        $\left[ \mathrm{v_i}\right] $ Send $\tilde{d}_i$ to the data collector\;
    }

    \tcc{Second round}
    $\left[ \mathrm{d}\right] $ Return the sequence $\tilde{d}_1,\dots,\tilde{d}_n$ to each user\;

    $\left[ \mathrm{v_i}\right] $ $\tilde{\mathbf{R}}_i\leftarrow \left( \mathcal{R}_{\varepsilon_1}^W \left( a_{i,1}\right) ,\dots ,\mathcal{R}_{\varepsilon_1}^W \left( a_{i,i-1}\right) \right)$ \tcp*{Apply $RR$ that provides $\varepsilon_1$-LDP to the lower triangular part of adjacency matrix $\mathbf{A}$, where $\varepsilon_1 = \left( 1-\alpha \right) \varepsilon_0$.}
    
    \For{$i=2$ \KwTo $n$}{
        $\left[ \mathrm{v_i}\right] $ $\hat{r}_i \leftarrow d_i\sum_{j=1}^{i-1}{\frac{\left( \tilde{a}_{i,j}-p\right) \tilde{d}_j}{1-2p}}$\;
        $\left[ \mathrm{v_i}\right] $ Send $\hat{r}_i$ to the shuffler\;
    }
    $\left[ \mathrm{s}\right]$ Sample a random permutation $\pi$ over $\{ 1,2,\dots,n\}$\;
    $\left[ \mathrm{s}\right]$ Send $\hat{r}_{\pi\left( 1\right)},\dots,\hat{r}_{\pi\left( n\right)}$ to the data collector\;

    $\left[ \mathrm{d}\right] $ $X_2\leftarrow \sum_{i=2}^n{\hat{r}_i}$\;
    
    $\left[ \mathrm{d}\right] $ $Y_2\leftarrow \left( \frac{1}{2}\sum_{i=1}^n{\tilde{d}_{i}^{2}}-\frac{n+2}{\varepsilon _{0}^{2}} \right) ^2-\frac{5n+4}{\varepsilon _{0}^{4}}$\;
    
    $\left[ \mathrm{d}\right] $ $\hat{q}_{ru}\left( G\right)\leftarrow X_2/M+Y_2/M^2$\;
    
    $\left[ \mathrm{d}\right] $ \textbf{return} $\hat{q}_{ru}\left( G\right)$
\end{algorithm}

Algorithm 2 (denoted by $\mathbf{Shuffle}_{ru}$) shows the two-round DP algorithm for the assortativity factor with shuffling. $\mathbf{Shuffle}_{ru}$ takes as input a graph $G$ (represented as neighbor lists $\mathbf{a}_1,\mathbf{a}_2,\dots,\mathbf{a}_n \in \left\{ 0,1\right\} ^n$), a privacy budget $\varepsilon$, a failure probability $\delta$ and a parameter for privacy budget allocation $\alpha$.

When running $\mathbf{Shuffle}_{ru}$, the first step is to call the function $\text{LocalPrivacyBudget}$ to compute a local privacy budget $\varepsilon_0$ from parameters $n,\varepsilon,\delta$ (line 1). Specifically, this function produces $\varepsilon_0$ such that $\varepsilon$ is a closed-form or numerical upper bound in the shuffle model with $n$ users. Given $\varepsilon_0$, the closed-form bound can be computed by Eq.\eqref{eq:shuffle}, while the numerical bound can be obtained using the open source code in \cite{feldman2022hiding}\footnote{https://github.com/apple/ml-shuffling-amplification}. 
Thus, we can easily compute $\varepsilon_0$ from $\varepsilon$ by computing a lookup table for the privacy budget pairs $(\varepsilon_0, \varepsilon)$ in advance.

In the first round, each user $v_i$ calculates his/her degree $d_i$ (line 3), then computes the noisy degree $\tilde{d}_i = d_i + \text{Lap}\left( \frac{1}{\alpha \varepsilon_0} \right)$ (line 4). Finally, $v_i$ sends $\tilde{d}_i$ to the data collector.

In the second round, the data collector returns the noisy degree $\tilde{d}_1, \dots, \tilde{d}_n$ to each user. Each user $v_i$ then obfuscates the bits with smaller user IDs in his/her neighbor list $\mathbf{a}_i$ using the random response mechanism $\mathcal{R}_{\varepsilon_1}^W$, and calculates $\hat{r}_i$ based on both $\tilde{d}_1, \dots, \tilde{d}_n$ and the obfuscated bits in $\mathbf{a}_i$ (lines 8 to 10). Next, $\hat{r}_i$ is sent to the Shuffler which randomly permutes $\hat{r}_1, \dots, \hat{r}_n$ and sends the permuted data $\hat{r}_{\pi(1)}, \dots, \hat{r}_{\pi(n)}$ to the data collector.

Finally, the data collector computes the estimate $\hat{q}_{ru}\left( G\right)$ of $q_{ru}\left( G\right)$ from $\tilde{d}_1, \dots, \tilde{d}_n$ and $\hat{r}_{\pi(1)}, \dots, \hat{r}_{\pi(n)}$. Note that the estimate \( \hat{q}_{ru} \) only involves summing \( \hat{r}_{\pi(1)}, \ldots, \hat{r}_{\pi(n)} \) while it is not necessary to know the permutation $\pi$. Thus, the expression for $\hat{q}_{ru}\left( G\right)$ is as follows
\begin{equation}
    \hat{q}_{ru}\left( G \right) =\frac{X_2}{M}-\frac{Y_2}{M^2},
\end{equation}
where
\begin{subequations}
\label{eq:whole}
\begin{eqnarray}
X_2=\sum_{i=2}^n{\hat{r}_i}=\sum_{i=2}^n{d_i\sum_{j=1}^{i-1}{\frac{\left( \tilde{a}_{i,j}-p\right) \tilde{d}_j}{1-2p}}},\label{subeq:4-2-1-1}
\end{eqnarray}
\begin{equation}
Y_2=\left( \frac{1}{2}\sum_{i=1}^n{\tilde{d}_{i}^{2}}-\frac{n+2}{\varepsilon _{0}^{2}} \right) ^2-\frac{5n+4}{\varepsilon _{0}^{4}}.\label{subeq:4-2-1-2}
\end{equation}
\end{subequations}

As above, we need to make a detailed analysis of $\mathbf{Shuffle_{ru}}$. First, let us state that $\mathbf{Shuffle_{ru}}$ provides the following privacy guarantee.
\begin{theorem}
    $\mathbf{Shuffle_{ru}}$ provides $\left( \varepsilon,\delta \right)$-edge DP.
    \label{theorem:7}
\end{theorem}

\begin{proof}
    Since $\tilde{d}_i=d_i+\text{Lap}\left( \frac{1}{\alpha \varepsilon_0} \right)$, $\tilde{d}_1,\dots ,\tilde{d}_n$ are protected under $\alpha \varepsilon_0$-edge LDP. The lower triangular part of the adjacency matrix $\mathbf{A}$ is obfuscated using RR, so $\tilde{\mathbf{R}}_1,\dots,\tilde{\mathbf{R}}_n$ are protected under $\varepsilon_1$-edge LDP, where $\varepsilon_1=\left( 1-\alpha \right) \varepsilon_0$. Note that $d_i$ is independent of $\sum_{j=1}^{i-1}{\frac{(\tilde{a}_{i,j}-p)\tilde{d}_j}{1-2p}}$, by the sequential composition and post-processing immunity of DP, $\hat{r}_1,\dots ,\hat{r}_n$ is protected under $\varepsilon_0$-edge LDP. After shuffling, the shuffled sequence $\hat{r}_{\pi \left( 1 \right)},\hat{r}_{\pi \left( 2 \right)},\dots,\hat{r}_{\pi \left( n\right)}$ achieves privacy amplification under $\left( \varepsilon,\delta\right)$-edge DP protection, where $\varepsilon=g\left( n,\varepsilon_0,\delta \right)$ (see Eq.(\ref{eq:shuffle}) for the expression). By the immunity to post-processing, $\hat{q}_{ru}\left( G \right)$ is still protected under $\left( \varepsilon,\delta\right)$-edge LDP. Therefore, $\mathbf{Shuffle}_{ru}$ provides $\left( \varepsilon,\delta\right)$-edge LDP.
\end{proof}

It is easy to see that $\mathbf{Shuffle_{ru}}$ achieves the amplification of privacy in comparison with $\mathbf{Local}_{ru}$ while a failure probability $\delta$ emerges.

\begin{theorem}
    The estimate $\hat{q}_{ru}\left( G\right)$ produced by $\mathbf{Shuffle_{ru}}$ satisfies $\mathbb{E}\left( \hat{q}_{ru}\left( G\right)\right)=q_{ru}\left( G\right)$.
    \label{theorem:8}
\end{theorem}

    

This suggests that $\mathbf{Shuffle_{ru}}$ outputs an unbiased estimation of the assortativity factor. In addition, we have the following statement.

\begin{theorem}
    When $\varepsilon$, $\delta$ are constants, $\alpha \in (0,1)$, $\varepsilon_0 = \log \left( n\right) + O\left( 1\right)$, the estimate $\hat{q}_{ru}\left( G\right)$ produced by $\mathbf{Shuffle_{ru}}$ provides the following utility guarantee:
    \begin{align}
        \mathrm{MSE}\left( \hat{q}_{ru}\left( G \right) \right) &=O\left( \frac{n^{1+\alpha}d_{\max}^{4}}{M^2}+\frac{n^3d_{\max}^{2}}{\left( \log n \right) ^2M^2}+\frac{n^3d_{\max}^{6}}{\left( \log n \right) ^2M^4} \right) . 
    \end{align}
    \label{theorem:9}
\end{theorem}


As before, a similar analysis implies that if the quantity $d_{\max}= O\left( n^{1/2}\right)$, then $\mathrm{MSE}\left( \hat{q}_{ru}\left( G \right) \right)= O\left( n^{1+\alpha}+\frac{n ^2}{(\log n)^2}\right)$. Therefore, the estimation yielded by $\mathbf{Shuffle_{ru}}$ has better steadiness than that of $\mathbf{Local_{ru}}$.

\begin{theorem}
    The time complexity of $\mathbf{Shuffle_{ru}}$ is $O\left( n^2\right)$, and the space complexity is $O\left( n^2\right)$.
    \label{theorem:10}
\end{theorem}
  This result may be proved using a similar analysis as used in proof of theorem \ref{theorem:6}, we thus omit it here.  

In the next subsection, we will focus on the other setting that is somewhat different from that discussed in the two subsections above. As a consequence, the associated algorithm is established and analyzed in detail.

\subsection{Computing Assortativity Factor from Extended Local Views}
In some social networks, individuals have not only knowledge of their own connections but also a broader subgraph within their local neighborhood. Such a subgraph is called an Extended Local View (ELV) \cite{sun2019analyzing}. For example, on Facebook, each user can see all the connections of his/her friends by default. Similarly, in real-life social interactions, we often learn about the relationships among our friends during social events. In these scenarios, individuals can provide information about their neighbors' connections (i.e. 2-hop ELVs) to the data collector.

To accurately compute the assortativity factor while preserving privacy in one such scenario, we propose a DDP algorithm under the assumption that each user knows all the connections of his/her friends. In this algorithm, we apply the Laplace mechanism to perturb each user's degree $d_i$ and the sum $T_i=\sum_{j=1}^n{a_{i,j}d_j}$ of his/her friends' degrees. However, adding an edge can increase $\{T_i,\dots,T_n\}$ by up to $4n-6$. That is to say, the global sensitivity of $\{T_i,\dots,T_n\}$ is $4n-6$ (according to definition \ref{def:GS under DDP}). When $n$ is large, we need to add a large amount of noise to $T_i$ to ensure privacy, which significantly reduces the utility of the algorithm's output.

To address this issue, we adopt a probabilistic upper bound on the local sensitivity in place of the global sensitivity based on Lemma \ref{lemma:TailBound}. Notice that the local sensitivity of $\{T_1,\dots,T_n\}$ under DDP satisfies
\begin{equation}
    LS_{DDP}(T)= \underset{v_i,v_j\in G\land i\ne j}{\max}2\left( d_i+d_j \right).
    \label{eq:LS_T}
\end{equation}
Specifically, when an edge $(v_i,v_j)$ is added or removed in $G$, (i) $T_i$ changes by $d_j$, (ii) $T_j$ changes by $d_i$, and (iii) $\Sigma_{l\neq i,j}T_l$ changes by $d_i+d_j$.

It is easy to see that the probabilistic upper bound of $LS_{DDP}(T)$ can be derived by calculating the probabilistic upper bounds of users' degrees. Specifically, we first perturb the degree of each user using the Laplace mechanism. Then, we use Lemma \ref{lemma:TailBound} to calculate the probabilistic upper bounds of all the noisy degrees. Finally, by substituting the two largest degree upper bounds into Eq.\ref{eq:LS_T}, we can obtain the probabilistic upper bound of $LS_{DDP}(T)$.

\IncMargin{0.2em}
\begin{algorithm}
    \caption{$\mathbf{Decentral_{ru}}$. $\left[ \mathrm{v_i} \right] $ and $\left[ \mathrm{d} \right]$ represent that the process is run by $v_i$ and the data collector, respectively.}\label{algo:3}
    \SetAlgoLined
    \KwData{ELVs of all users ${G_1, \dots, G_n}$, privacy budget $\varepsilon_1,\varepsilon_2 \in \mathbb{R}_{\ge 0}$, failure probability $\delta \in [0,1]$}
    \KwResult{Private estimate $\hat{q}_{ru}\left( G\right)$ of $q_{ru}\left( G\right)$}
    \BlankLine


    $\left[ \mathrm{d}\right] $ $\delta_1 \leftarrow \delta/2$\;
    
    \For{$i=1$ \KwTo $n$}{
        $\left[ \mathrm{v_i}\right] $ $d_i \leftarrow \sum_{j=1}^n{a_{i,j}}$\;
        $\left[ \mathrm{v_i}\right] $ $\tilde{d}_i \leftarrow d_i+\mathrm{Lap}\left( \frac{2}{\varepsilon_1}\right)$\;
        $\left[ \mathrm{v_i}\right] $ $d_i^\ast \leftarrow \tilde{d}_i+\frac{2}{\varepsilon_1}\log \left( \frac{1}{2\delta_1}\right)$\;
        $\left[ \mathrm{v_i}\right] $ Send $\tilde{d}_i$ and $d_i^\ast$ to the data collector\;
    }
    
    $\left[ \mathrm{d}\right] $ Sort $\left\{ d_i^\ast \right\}$ into $\left\{ d_{[1]}^\ast, d_{[2]}^\ast, \dots, d_{[n]}^\ast \right\}$ in descending order\;
    $\left[ \mathrm{d}\right] $ $\Delta \leftarrow 2(d_{[1]}^\ast +d_{[2]}^\ast )$\;
    

    

    \For{$i=1$ \KwTo $n$}{   
        $\tilde{d}_i \leftarrow d_i+\mathrm{Lap}\left( \frac{1}{\varepsilon_0}\right)$\;
        $\left[ \mathrm{v_i}\right] $ $\tilde{T}_i \leftarrow T_i+\mathrm{Lap}\left( \frac{\Delta}{\varepsilon_2}\right)$\;
        $\left[ \mathrm{v_i}\right] $ Send $\tilde{T}_i$ to the data collector\;
    }


    $\left[ \mathrm{d}\right] $ $X_3\leftarrow \frac{1}{2}\sum_{i=1}^n{\tilde{d}_i \tilde{T}_i}$\;
    

    $\left[ \mathrm{d}\right] $ $Y_3\leftarrow \left( \frac{1}{2}\sum_{i=1}^n{\tilde{d}_{i}^{2}}-\frac{4\left( n+2 \right)}{\varepsilon _{1}^{2}} \right) ^2-\frac{16\left( 5n+4 \right)}{\varepsilon _{1}^{4}}$\;
    
    $\left[ \mathrm{d}\right] $ $\hat{q}_{ru}\left( G\right)\leftarrow X_3/M+Y_3/M^2$\;
    
    $\left[ \mathrm{d}\right] $ \textbf{return} $\hat{q}_{ru}\left( G\right)$
\end{algorithm}


By far, we are ready to propose the third algorithm. Algorithm 3 shows our algorithm for computing the assortativity factor from extended local views. We denote this algorithm by $\mathbf{Decentral_{ru}}$. It takes several inputs: ELVs of all users ${G_1, G_2, \dots, G_n}$, privacy budgets $\varepsilon_1, \varepsilon_2$, and a failure probability $\delta$.

The algorithm begins by splitting the privacy parameter $\delta$, which is handled by the data collector. The data collector then broadcasts the parameters to all users, i.e., nodes in the social network. Lines 2–7 are executed by each client, where each user computes his/her noisy degree $\tilde{d}_i$ and a probabilistic upper bound $d_i^\prime$ of the true degree $d_i$, and reports these values to the data collector. Based on the received data $d_1^\ast, d_2^\ast, \dots, d_n^\ast$, the data collector determines the sensitivity required for perturbing $T_i$, and sends it back to each user. Then, each user applies the Laplace mechanism to compute the noisy value $\tilde{T}_i$ and reports it to the data collector. Finally, the data collector estimates $q_{ru}\left( G\right)$ as follows
\begin{equation}
    \hat{q}_{ru}\left( G \right) =\frac{X_3}{M}-\frac{Y_3}{M^2},
\end{equation}
where
\begin{subequations}
\label{eq:whole}
\begin{eqnarray}
X_3=\frac{1}{2} \sum_{i=1}^n{\tilde{d}_i \tilde{T}_i},\label{subeq:4-3-1-1}
\end{eqnarray}
\begin{equation}
Y_3=\left( \frac{1}{2}\sum_{i=1}^n{\tilde{d}_{i}^{2}}-\frac{4\left( n+2 \right)}{\varepsilon _{1}^{2}} \right) ^2-\frac{16\left( 5n+4 \right)}{\varepsilon _{1}^{4}}.
\label{subeq:4-3-1-2}
\end{equation}
\end{subequations}

Following the similar analysis as above, we prove that $\mathbf{Decentral_{ru}}$ has the following properties. 
\begin{theorem}
    $\mathbf{Decentral_{ru}}$ provides $\left( \varepsilon,\delta \right)$-edge DDP, where $\varepsilon=\varepsilon_1+\varepsilon_2$ and $\delta =2\delta_1$.
    \label{theorem:11}
\end{theorem}

\begin{proof}
    Obviously, for any $i\in \left\{ 1,\dots,n\right\}$, $\tilde{d}_i=d_i+\text{Lap}\left( \frac{2}{\varepsilon_1}\right)$ is protected with $\varepsilon_1$-edge DDP. Next, we show that $\Delta \ge LS_{DDP}(T_i)$ holds with probability at least $1 - \delta$.  This is equivalent to proving that $(d_{[1]}^\ast \ge d_{(1)}) \land (d_{[2]}^\ast \ge d_{(2)})$ holds with probability at least $1 - \delta$, where $d_{(i)}$ denotes the $i$-th largest true degree.  

    Let $\xi_i \in [0,1]$ be the total probability that another user's degree other than $d_{(i)}$ becomes $d_{[i]}^\ast$ and meanwhile $d_{[i]}^\ast \ge d_{(i)}$.  It is easy to deduce that for any $i \in \{1, 2, \dots, n\}$, we have
    \begin{equation}
        \begin{aligned}
            \mathrm{Pr}[d_{[i]}^\ast \ge d_{(i)}] &= \mathrm{Pr}[d_{(i)}^\ast \ge d_{(i)}]+\xi_i \\
            &\ge \mathrm{Pr}[d_{(i)}^\ast \ge d_{(i)}] \\
            &= 1-\delta_1.
        \end{aligned}
    \end{equation}
    Thus,
    \begin{equation}
        \begin{aligned}
            &\text{Pr}\left[ \left( d_{\left[ 1 \right]}^{\ast}\ge d_{\left( 1 \right)} \right) \land \left( d_{\left[ 2 \right]}^{\ast}\ge d_{\left( 2 \right)} \right) \right] \\
            &=1-\text{Pr}\left[ \left( d_{\left[ 1 \right]}^{\ast}<d_{\left( 1 \right)} \right) \lor \left( d_{\left[ 2 \right]}^{\ast}<d_{\left( 2 \right)} \right) \right] \\
            &\ge 1-\left[ \text{Pr}\left( d_{\left[ 1 \right]}^{\ast}<d_{\left( 1 \right)} \right) +\text{Pr}\left( d_{\left[ 2 \right]}^{\ast}<d_{\left( 2 \right)} \right) \right] \\
            &=\text{Pr}\left( d_{\left[ 1 \right]}^{\ast}\ge d_{\left( 1 \right)} \right) +\text{Pr}\left( d_{\left[ 2 \right]}^{\ast}\ge d_{\left( 2 \right)} \right) -1 \\
            &\ge 1-2\delta _1=1-\delta.
        \end{aligned}
    \end{equation}
    
    Therefore, $\tilde{T}_i=T_i+\text{Lap}\left( \frac{\Delta}{\varepsilon_2}\right)$ is protected under $\varepsilon_2$-edge DDP with probability $1-\delta$. Overall, $\mathbf{Decentral_{ru}}$ provides $(\varepsilon,\delta)$-edge DDP, where $\varepsilon=\varepsilon_1+\varepsilon_2$ and $\delta=2\delta_1$.
    
\end{proof}

\begin{theorem}
    The estimate $\hat{q}_{ru}\left( G\right)$ produced by $\mathbf{Decentral_{ru}}$ satisfies $\mathbb{E}\left( \hat{q}_{ru}\left( G\right)\right)=q_{ru}\left( G\right)$.
    \label{theorem:12}
\end{theorem}

\begin{theorem}
    When $\varepsilon_1$, $\varepsilon_2$ are constants, the estimate $\hat{q}_{ru}\left( G\right)$ produced by $\mathbf{Decentral_{ru}}$ provides the following utility guarantee:
    \begin{align}
        \mathrm{MSE}\left( \hat{q}_{ru}\left( G \right) \right) &= O\left( \frac{n^3d_{\max}^{6}}{M^4} \right) . 
    \end{align}
    \label{theorem:13}
\end{theorem}

If quantity $\hat{d}_{\max}= \Theta (n^{1/6})$, we find that $\mathrm{MSE}\left( \hat{q}_{ru}\left( G \right) \right)= O\left( 1\right)$. In this setting, the consequence produced by $\mathbf{Decentral_{ru}}$ is overwhelmingly better than that by $\mathbf{Local_{ru}}$. This implies that a broader view is beneficial to a steadier estimation given $\hat{d}_{\max}= \Theta (n^{1/6})$. Furthermore, if $\hat{d}_{\max}\leq O (n^{1/2})$, we have also $\mathrm{MSE}\left( \hat{q}_{ru}\left( G \right) \right)\leq O\left( n^2\right)$.

\begin{theorem}
    The time complexity of $\mathbf{Decentral_{ru}}$ is $O\left( n^2\right)$, and the space complexity is $O\left( n^2\right)$.
    \label{theorem:14}
\end{theorem}
As above, this may also be proved in a similar manner as used in the proof of theorem \ref{theorem:6}.

By far, we finish the development of algorithms $\mathbf{Local_{ru}}$, $\mathbf{Shuffle_{ru}}$ and $\mathbf{Decentral_{ru}}$, and give detailed theoretical analysis. In the next section, we are going to conduct experimental evaluations on synthetic datasets and real-world datasets to further clarify the performance of the proposed algorithms.

\section{Experiment Evaluation}

In this section, we will conduct experiments on synthetic datasets and real-world datasets to evaluate the performance of our algorithms $\mathbf{Local_{ru}}$, $\mathbf{Shuffle_{ru}}$ and $\mathbf{Decentral_{ru}}$ proposed in Section 4.

\subsection{Experimental Set-up}
We perform experiments on six datasets whose parameters are summarized in Table \ref{tab:datasets}. The first three are synthetic datasets with power-law degree distribution based on the BA (Barabási-Albert) graph model \cite{posfai2016network}, where $m$ represents the number of edges of a newly added node. Here we use the NetworkX library \cite{hagberg2008exploring} to generate BA synthetic graphs. The other three are real-world datasets from the Stanford Network Analysis Project (SNAP).

\begin{table}[h]
    \centering
    \caption{datasets}
    \label{tab:datasets}
    \renewcommand\arraystretch{1.2}
    \begin{tabular}{|c|c|c|c|c|c|}
        \hline
        & $n$ & $M$ & $d_{max}$ & $d_{avg}$ & $r_u$ \\
        \hline
        BA($m=10$) & 10000 & 99900 & 464 & 19.98 & -95.22 \\
        \hline
        BA($m=50$) & 10000 & 497500 & 1237 & 99.5 & 49.03 \\
        \hline
        BA($m=100$) & 10000 & 990000 & 1640 & 198 & 736.81 \\
        \hline
        Facebook & 4039 & 88234 & 1045 & 43.69 & 870.36 \\
        \hline
        Deezer & 28281 & 92752 & 172 & 6.56 & 29.19 \\
        \hline
        GitHub & 37700 & 289003 & 9458 & 15.33 & -162127.53 \\
        \hline
    \end{tabular}
\end{table}

\begin{itemize}
    \item \textbf{Facebook.} The Facebook online social network dataset (denoted by Facebook) \cite{leskovec2012learning} contains friend lists from survey participants using the Facebook application. It provides a social graph $G=\left( V,E\right)$ with 4,039 nodes (Facebook users) and 88,234 undirected edges, where each edge $\left( v_i,v_j\right) \in E$ represents an online friendship between $v_i$ and $v_j$.
    \item \textbf{Deezer.} This dataset (denoted by Deezer) \cite{rozemberczki2020characteristic} includes a graph $G$ with 28,281 nodes and 92,752 undirected edges, where nodes represent Deezer users from European countries and edges indicate mutual follower relationships.
    \item \textbf{GitHub.} This dataset (denoted by GitHub) \cite{rozemberczki2021multi} is a social network of GitHub developers with 37,700 nodes and 289,003 edges. Nodes are developers who have starred at least 10 repositories, and edges are the mutual follow relationships between them.
\end{itemize}

For each algorithm, we evaluate the relative error (as described in subsection 3.5) of the estimate as the value of $\varepsilon$ varies. In our experiments, we set the range of $\varepsilon$ to $(0, 2]$, which is acceptable in many practical scenarios \cite{li2017differential}. To ensure reliable results, we run each algorithm 20 times and use the average relative error of these runs as the final experimental result. Additionally, given that the assortativity factor $r_u$ reflects the sign of assortativity, it is useful to report the extent to which the privacy estimate of $r_u$ reflects the true sign of assortativity, in addition to the relative error. Therefore, we also assess the accuracy of the sign of the estimate as $\varepsilon$ varies. We run each algorithm 100 times and report the proportion of correct signs in the outputs as the sign accuracy of each algorithm. Our source code is available on GitHub\footnote{https://github.com/OYJZBYCG/Differential-Privacy-for-Network-Assortativity.git}.

\subsection{Experimental Results}

\subsubsection{Relation between the Relative Error and $\varepsilon$}
We first evaluate the relationship between the relative error and the privacy parameter $\varepsilon$ in edge DP. In this experiment, for $\mathbf{Local_{ru}}$, we divide the total privacy budget $\varepsilon$ into $\varepsilon_1 = 0.6\varepsilon$ for perturbing the adjacency matrix $\mathbf{A}$ and $\varepsilon_2 = 0.4\varepsilon$ for perturbing the users' degrees $d_1,\dots,d_n$. For $\mathbf{Shuffle_{ru}}$, we set the privacy allocation parameter $\alpha$ to 0.4 and the failure probability $\delta$ to $10^{-8}$ ($\ll n$) and use the numerical upper bound in \cite{feldman2022hiding} to calculate the local privacy budget $\varepsilon_0$ in the shuffling model. For $\mathbf{Decentral_{ru}}$, $\delta$ is also set to $10^{-8}$. Privacy budget $\varepsilon$ is divided into $\varepsilon_1 = 0.4\varepsilon$ and $\varepsilon_2 = 0.6\varepsilon$ for perturbation of the user's degree $\{d_1,d_2,\dots,d_n\}$ and number of friends of the user's friends $\{T_1,T _2,\dots,T_n\}$, respectively.


\begin{figure}[thbp!]
    \centering
    \begin{subfigure}[b]{0.5\textwidth}
        \includegraphics[width=\textwidth]{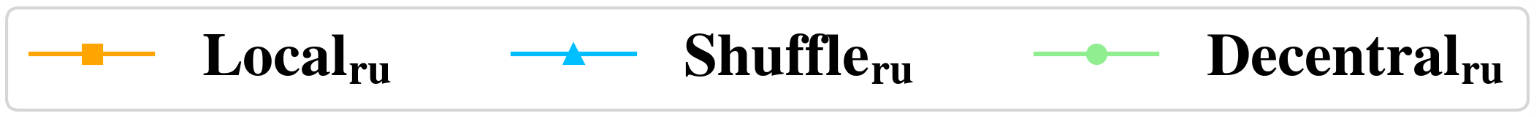}
    \end{subfigure}

    \vspace{1em}
    
    \begin{subfigure}[b]{0.235\textwidth}
        \includegraphics[width=\textwidth]{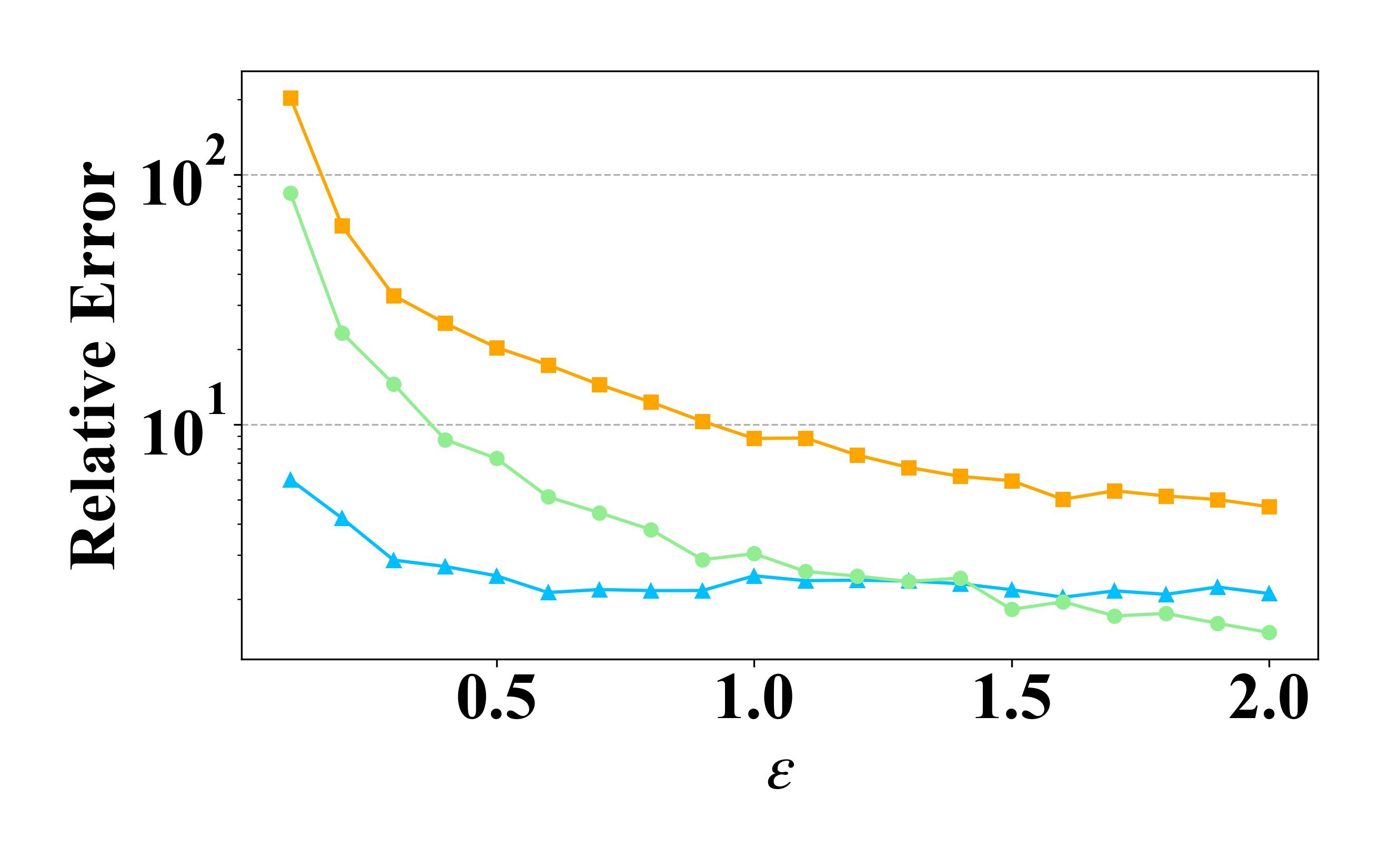}
        \caption{BAGraph-m10}
    \end{subfigure}
    \hfill
    \begin{subfigure}[b]{0.235\textwidth}
        \includegraphics[width=\textwidth]{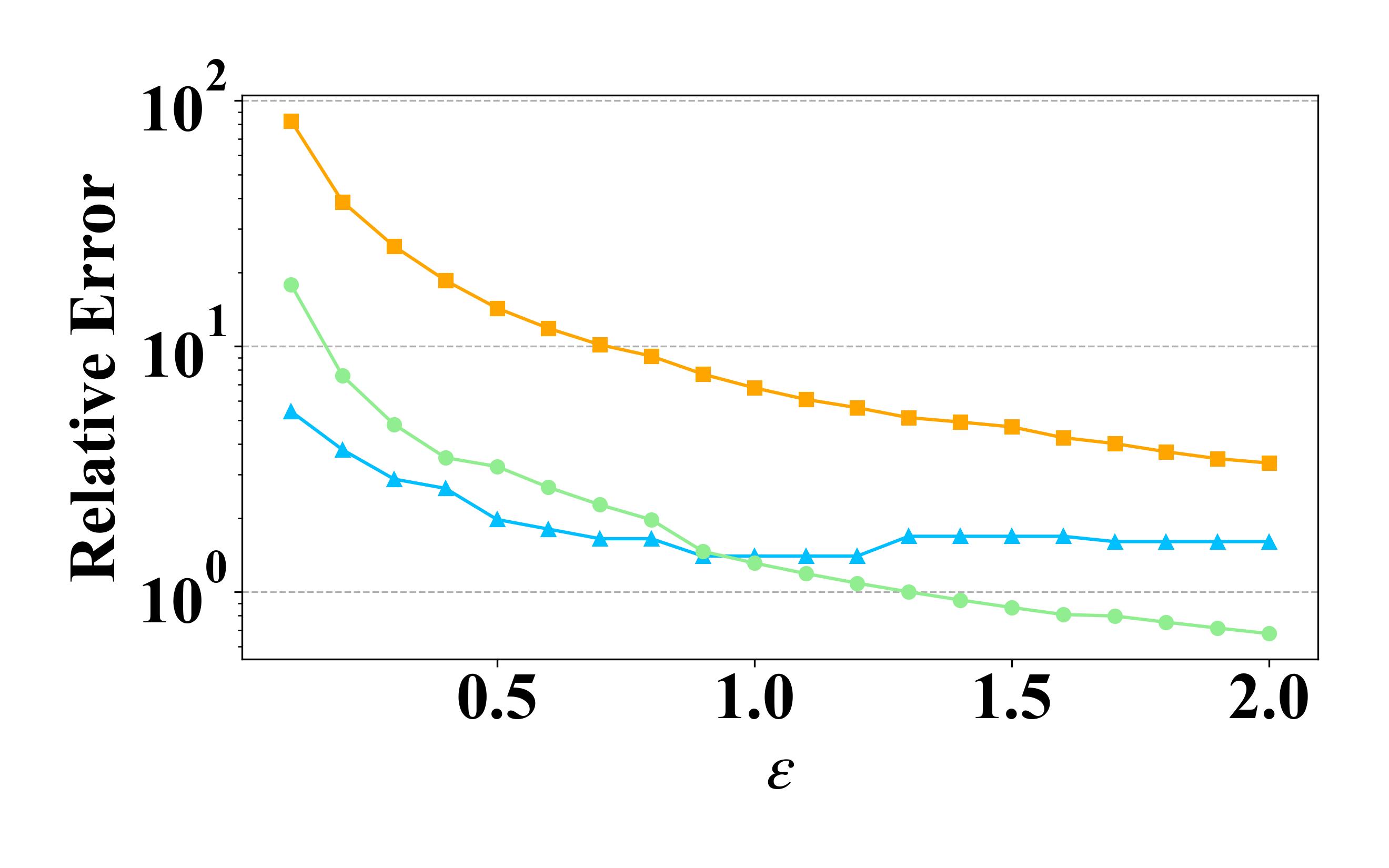}
        \caption{BAGraph-m50}
    \end{subfigure}
    
    \vspace{1em} 
    
    \begin{subfigure}[b]{0.235\textwidth}
        \includegraphics[width=\textwidth]{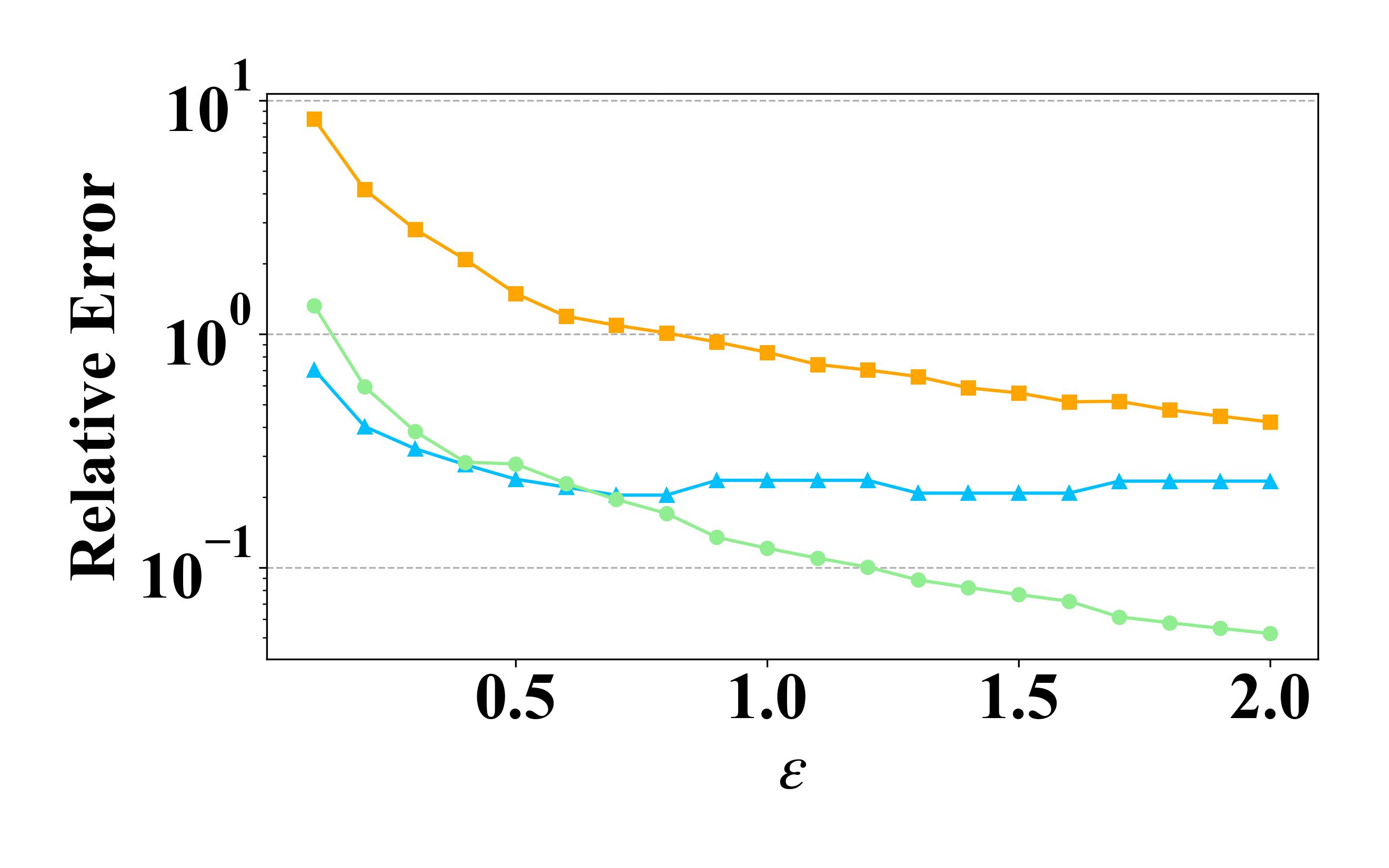}
        \caption{BAGraph-m100}
    \end{subfigure}
    \hfill
    \begin{subfigure}[b]{0.235\textwidth}
        \includegraphics[width=\textwidth]{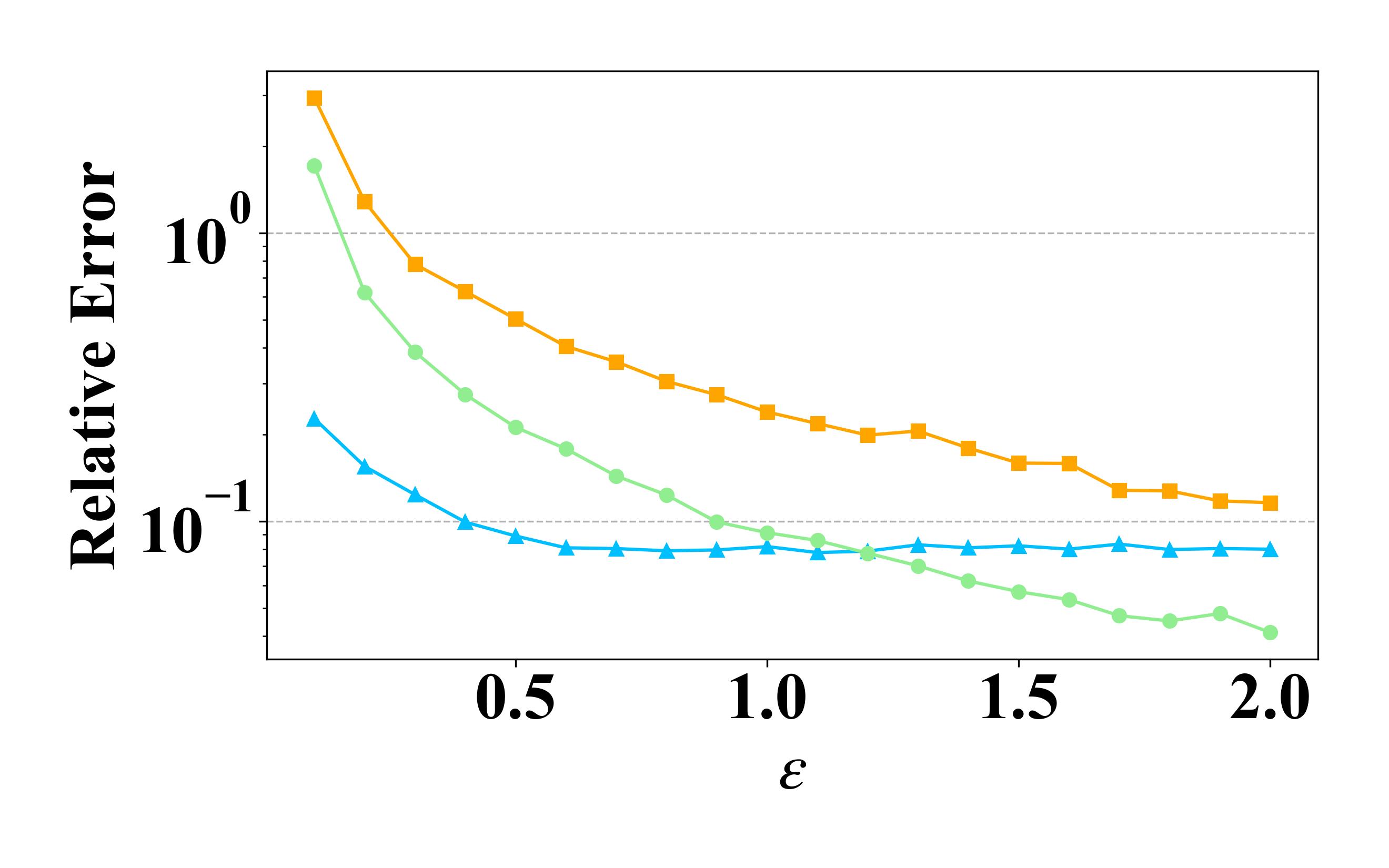}
        \caption{Facebook}
    \end{subfigure}

    \vspace{1em} 
    
    \begin{subfigure}[b]{0.235\textwidth}
        \includegraphics[width=\textwidth]{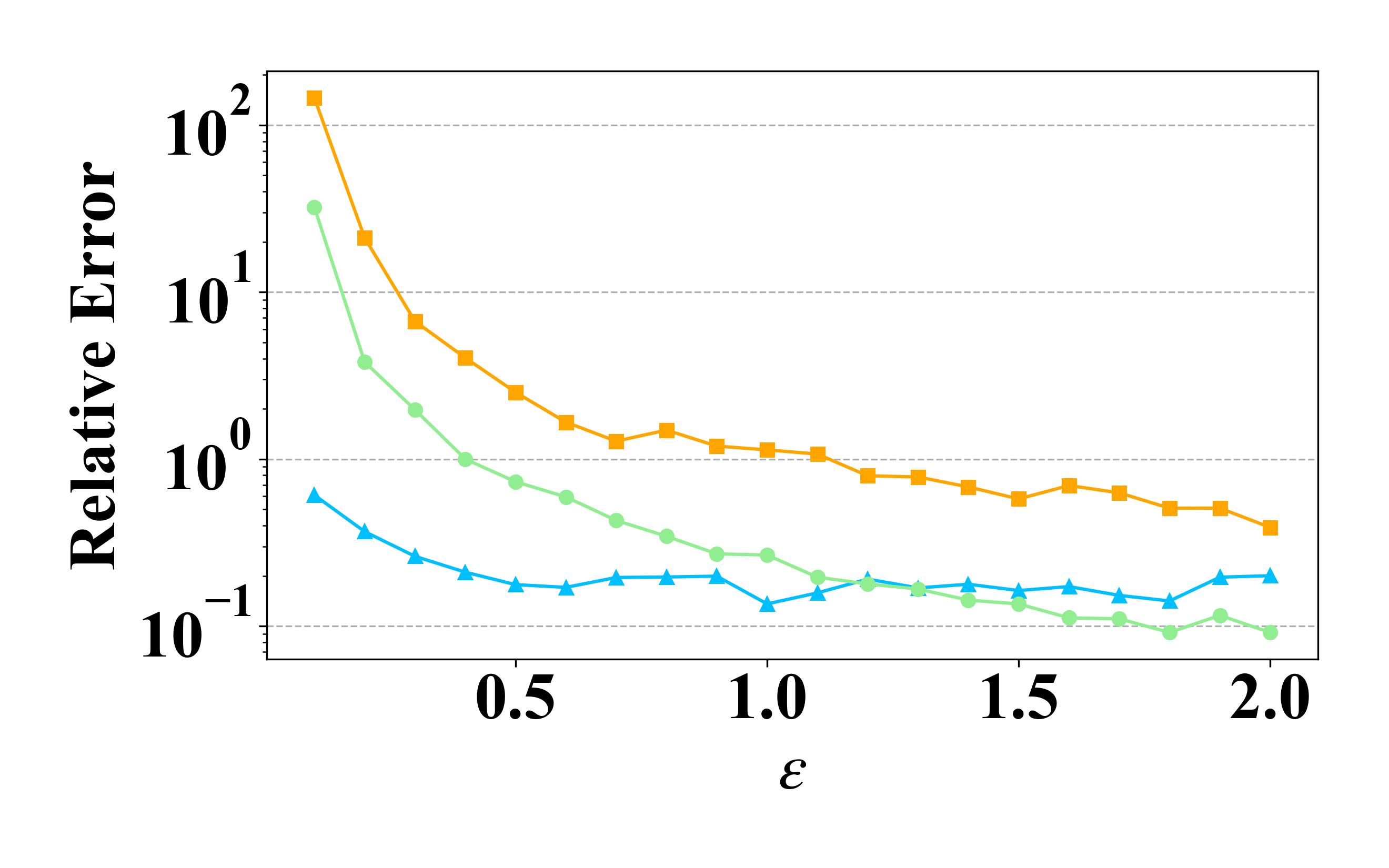}
        \caption{Deezer}
    \end{subfigure}
    \hfill
    \begin{subfigure}[b]{0.235\textwidth}
        \includegraphics[width=\textwidth]{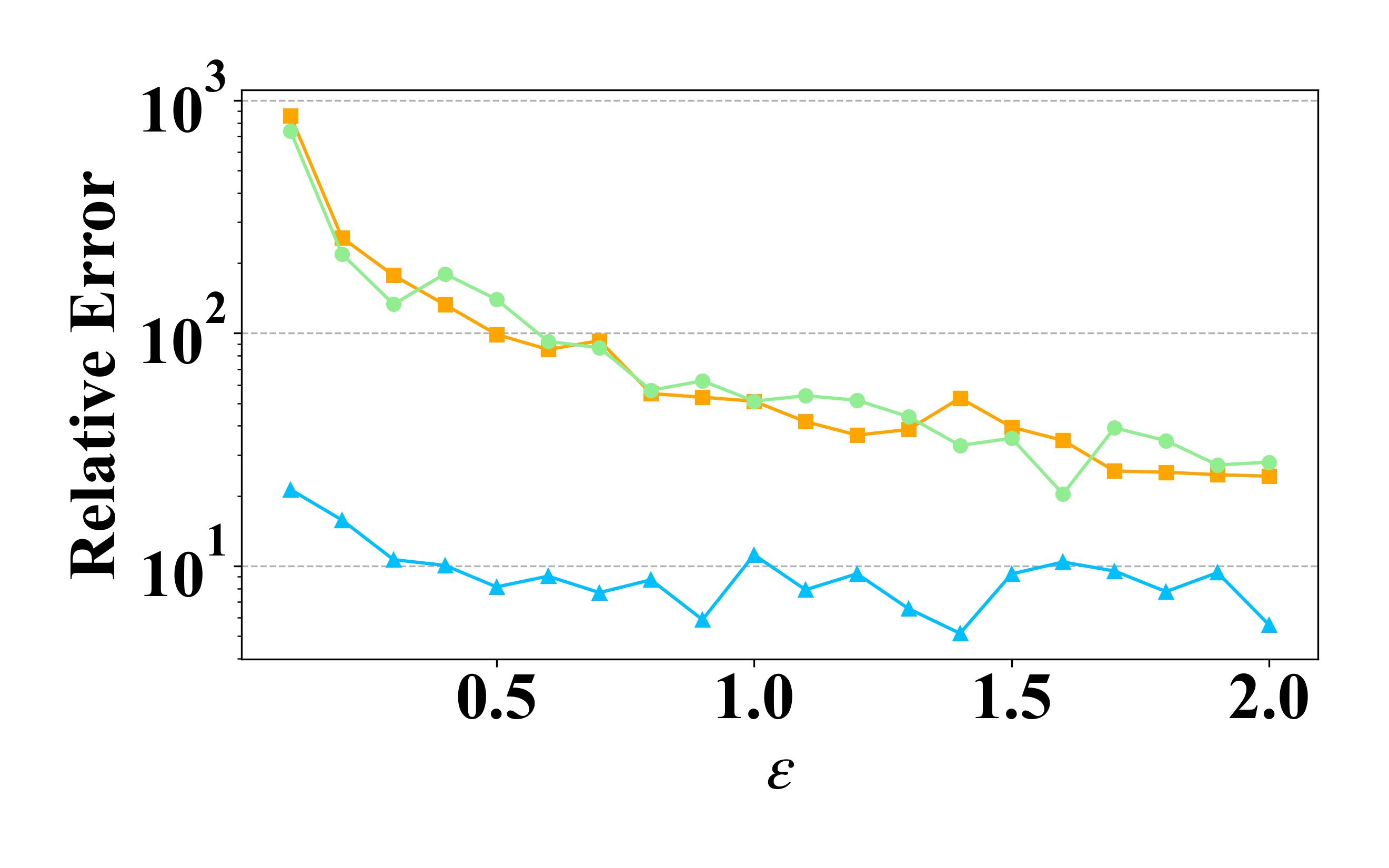}
        \caption{GitHub}
    \end{subfigure}

    \caption{Relation between the relative error and $\varepsilon$ in edge DP.}
    \label{fig1}
\end{figure}

All the experimental results are shown in Fig.\ref{fig1}. It is apparent to see that for $\varepsilon$ in question, $\mathbf{Decentral_{ru}}$ has a better performance than $\mathbf{Local_{ru}}$. This means that a broader view is conducive to obtaining better estimates. At the same time, given the low level of privacy budget $\varepsilon$, $\mathbf{Shuffle_{ru}}$ is always optimal in comparison with the other two algorithms as this algorithm has a strictly smaller relative error than them. This is in line with the theoretical analysis in the preceding section. This is mainly because its ingredient, Shuffle Model, owns its merit, i.e., privacy amplification effect. As encountered in the literature \cite{dwork2006differential,hay2009accurate, day2016publishing, ma2024ptt}, all the previous DP-based schemes prefer an assumption of low privacy budget. In a word, $\mathbf{Shuffle_{ru}}$ is the best choice among these proposed algorithms.  

\subsubsection{Relation between the Sign Accuracy and $\varepsilon$} 
Next, We evaluate the relationship between the sign accuracy and the privacy parameter $\varepsilon$ in edge DP. Fig.\ref{fig2} shows the results of the experiment. It is clear that the accuracy of $\mathbf{Decentral_{ru}}$ is higher than that of $\mathbf{Local_{ru}}$, which reinforces the validity of using a broader view to improve our estimation accuracy. In addition, $\mathbf{Shuffle_{ru}}$ is significantly better than the other two, and the accuracy is almost always 100\%. This suggests that it is highly effective to take advantage of the shuffle model to improve the estimation accuracy under DP.

\begin{figure}[thbp!]
    \centering
    \begin{subfigure}[b]{0.5\textwidth}
        \includegraphics[width=\textwidth]{RE+Sign.jpg}
    \end{subfigure}

    \vspace{1em}
    
    \begin{subfigure}[b]{0.235\textwidth}
        \includegraphics[width=\textwidth]{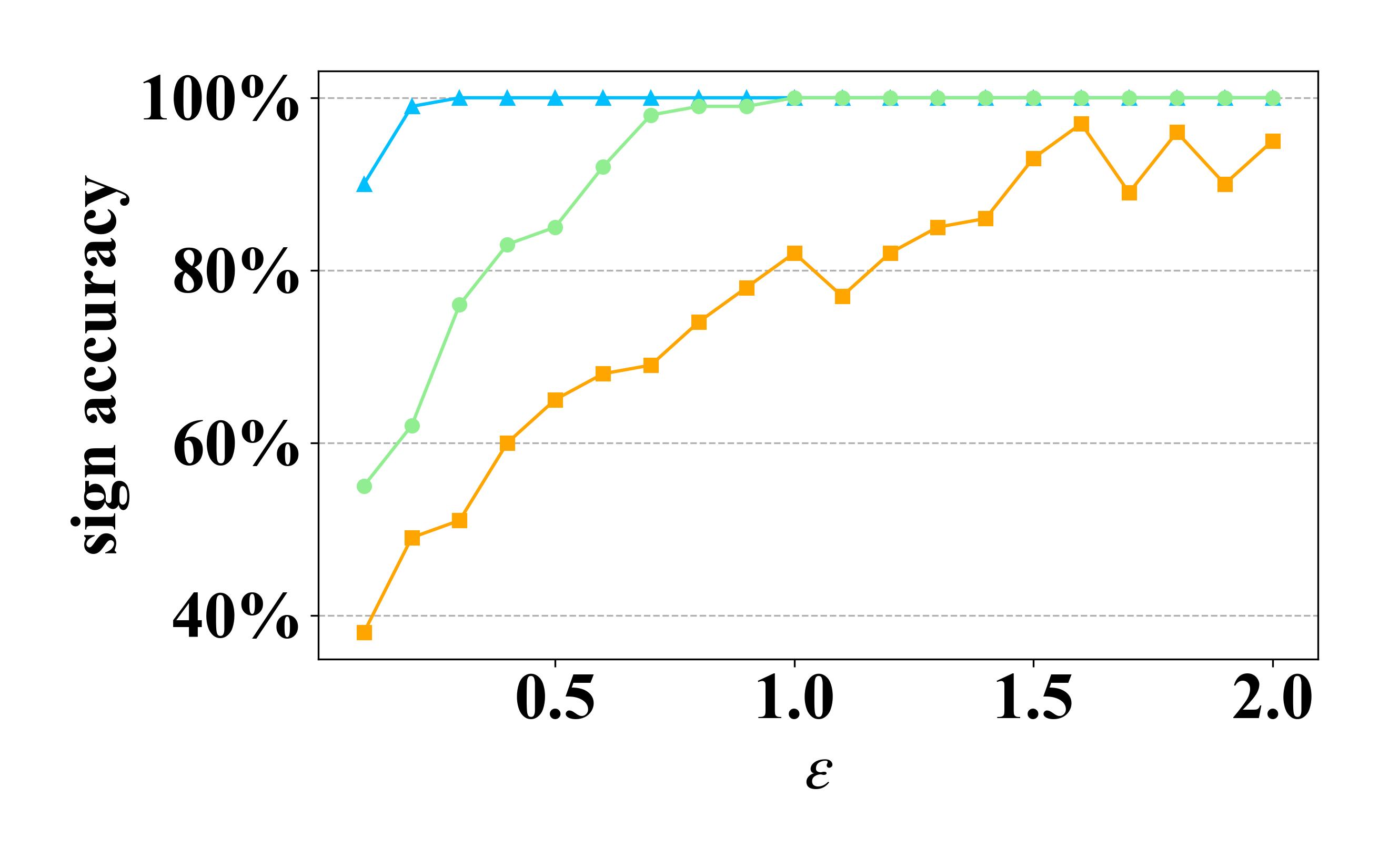}
        \caption{BAGraph-m10}
    \end{subfigure}
    \hfill
    \begin{subfigure}[b]{0.235\textwidth}
        \includegraphics[width=\textwidth]{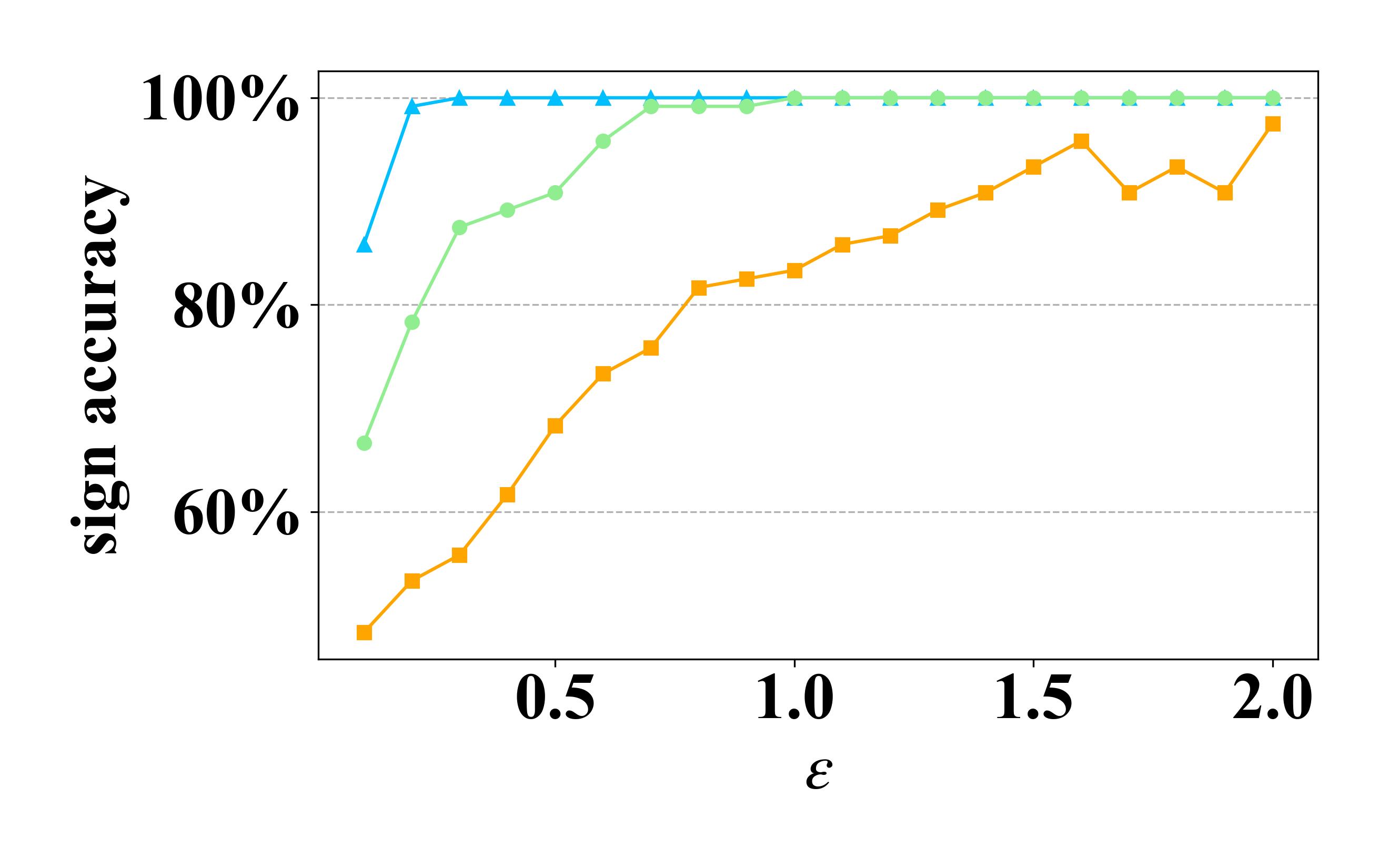}
        \caption{BAGraph-m50}
    \end{subfigure}
    
    \vspace{1em} 
    
    \begin{subfigure}[b]{0.235\textwidth}
        \includegraphics[width=\textwidth]{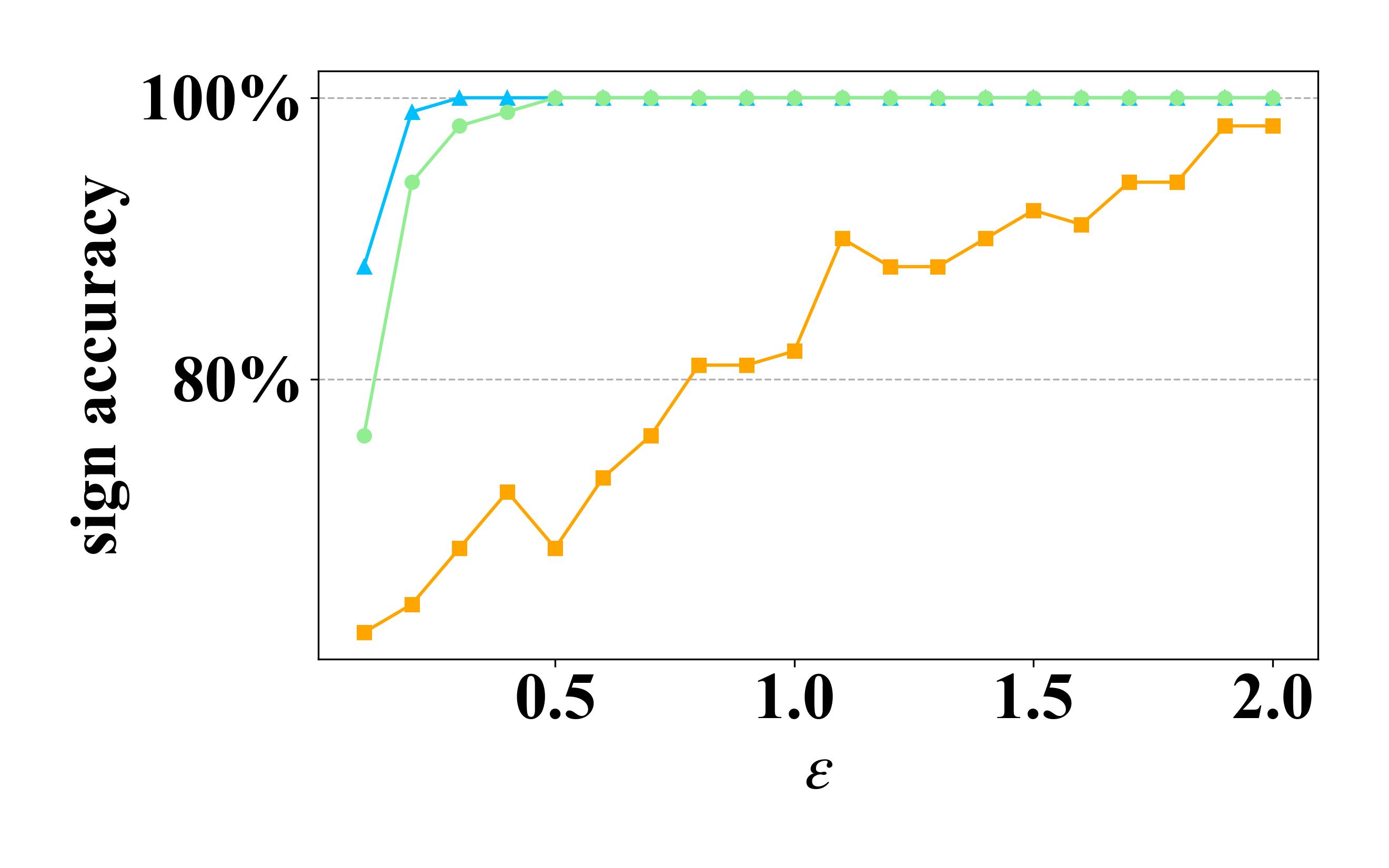}
        \caption{BAGraph-m100}
    \end{subfigure}
    \hfill
    \begin{subfigure}[b]{0.235\textwidth}
        \includegraphics[width=\textwidth]{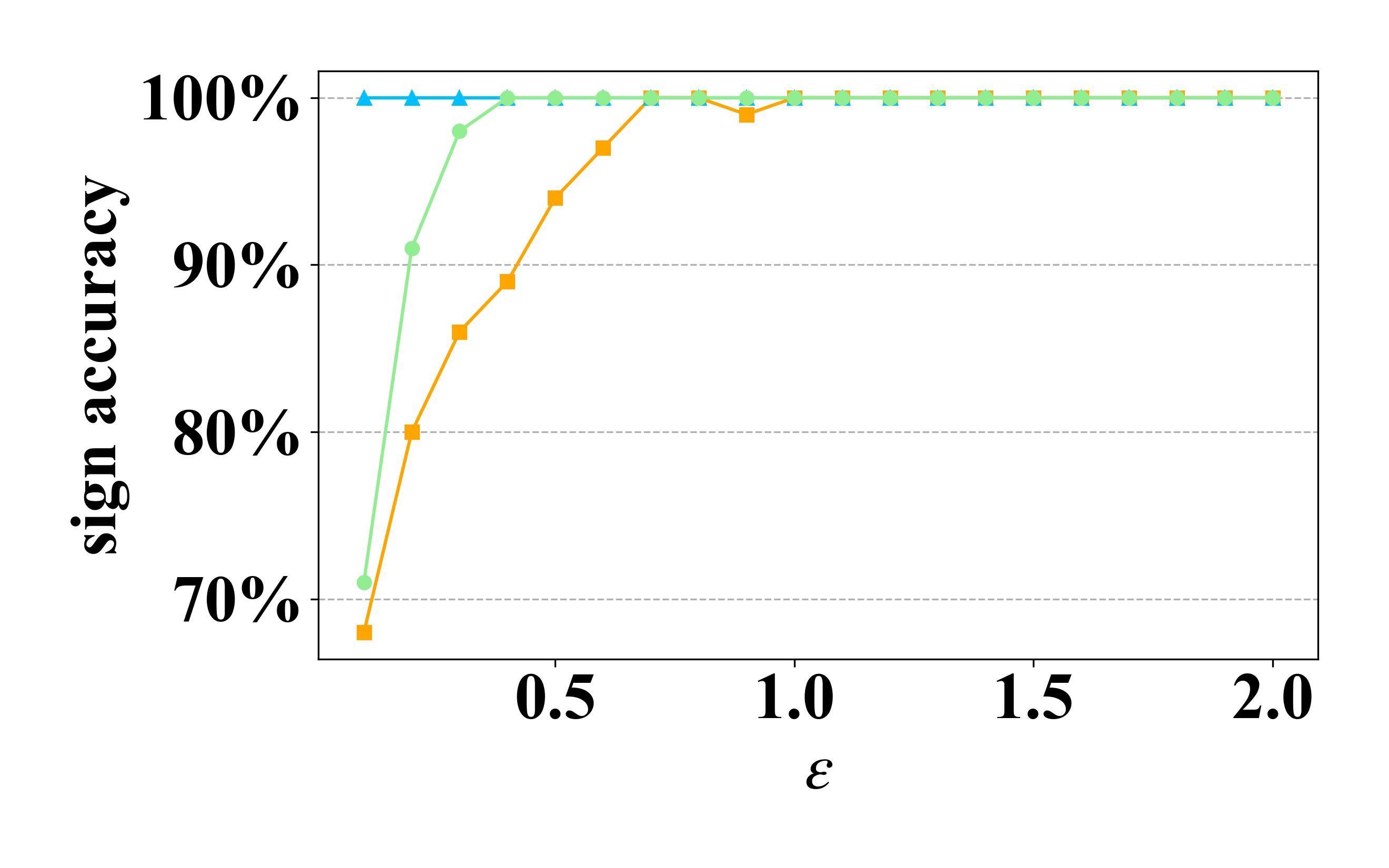}
        \caption{Facebook}
    \end{subfigure}

    \vspace{1em} 
    
    \begin{subfigure}[b]{0.235\textwidth}
        \includegraphics[width=\textwidth]{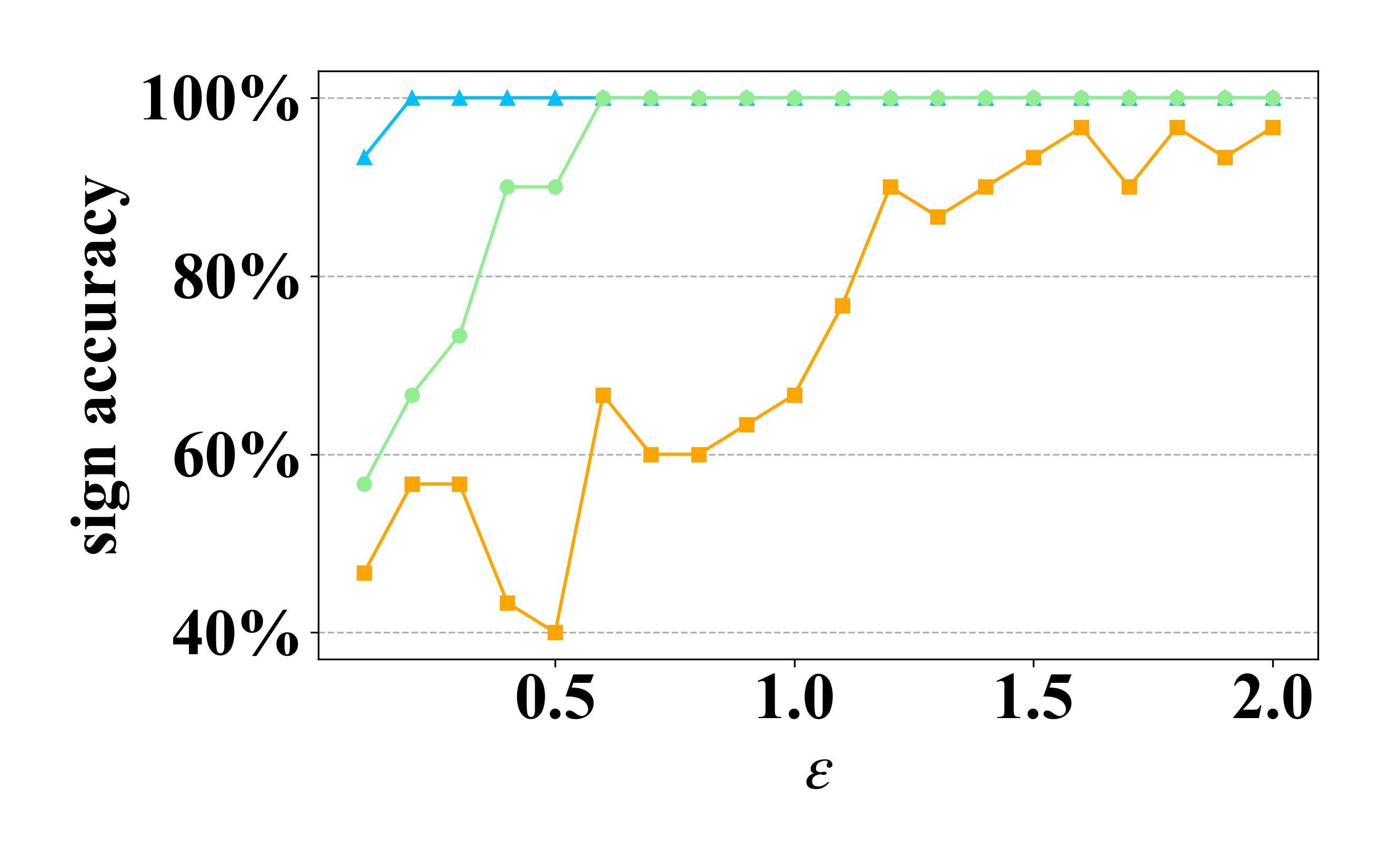}
        \caption{Deezer}
    \end{subfigure}
    \hfill
    \begin{subfigure}[b]{0.235\textwidth}
        \includegraphics[width=\textwidth]{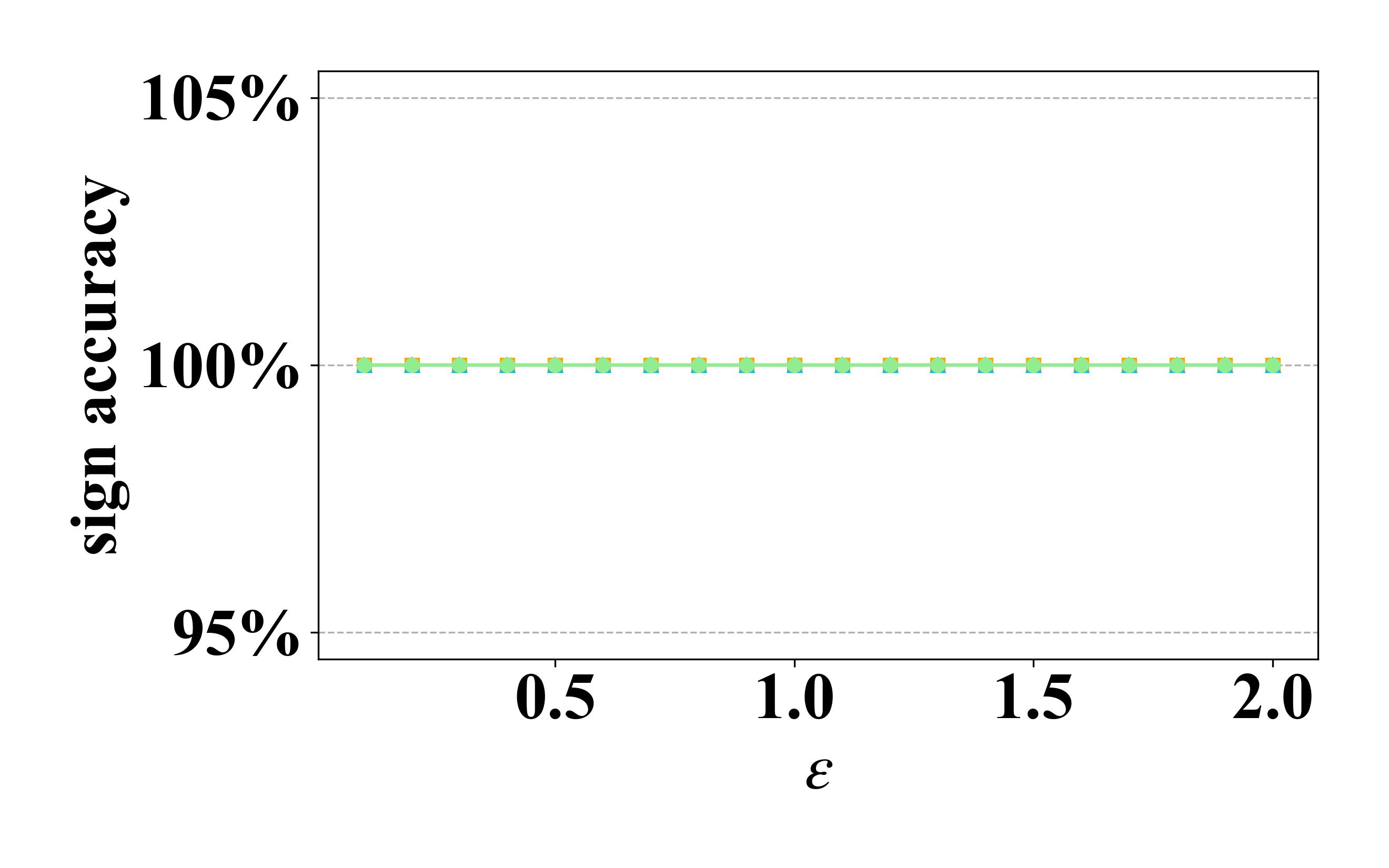}
        \caption{GitHub}
    \end{subfigure}

    \caption{Relation between the sign accuracy and $\varepsilon$ in edge DP.}
    \label{fig2}
\end{figure}

\subsubsection{Numerical bound vs. closed-form bound in $\mathbf{Shuffle_{ru}}$}
As mentioned in subsection 3.4, Feldman et al. provided two types of upper bounds on the privacy level in the shuffle model: a closed-form upper bound (see Eq.(\ref{eq:shuffle})) and a numerical upper bound (computed by the algorithm in \cite{feldman2022hiding}). Below, we compare the performance of $\mathbf{Shuffle_{ru}}$ using these two bounds. Here, we also set $\delta = 10^{-8}$.

\begin{figure}[thbp!]
    \centering
    \begin{subfigure}[b]{0.5\textwidth}
        \includegraphics[width=\textwidth]{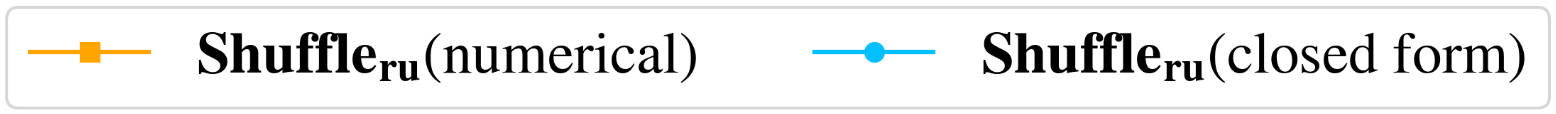}
    \end{subfigure}

    \vspace{1em}
    
    \begin{subfigure}[b]{0.235\textwidth}
        \includegraphics[width=\textwidth]{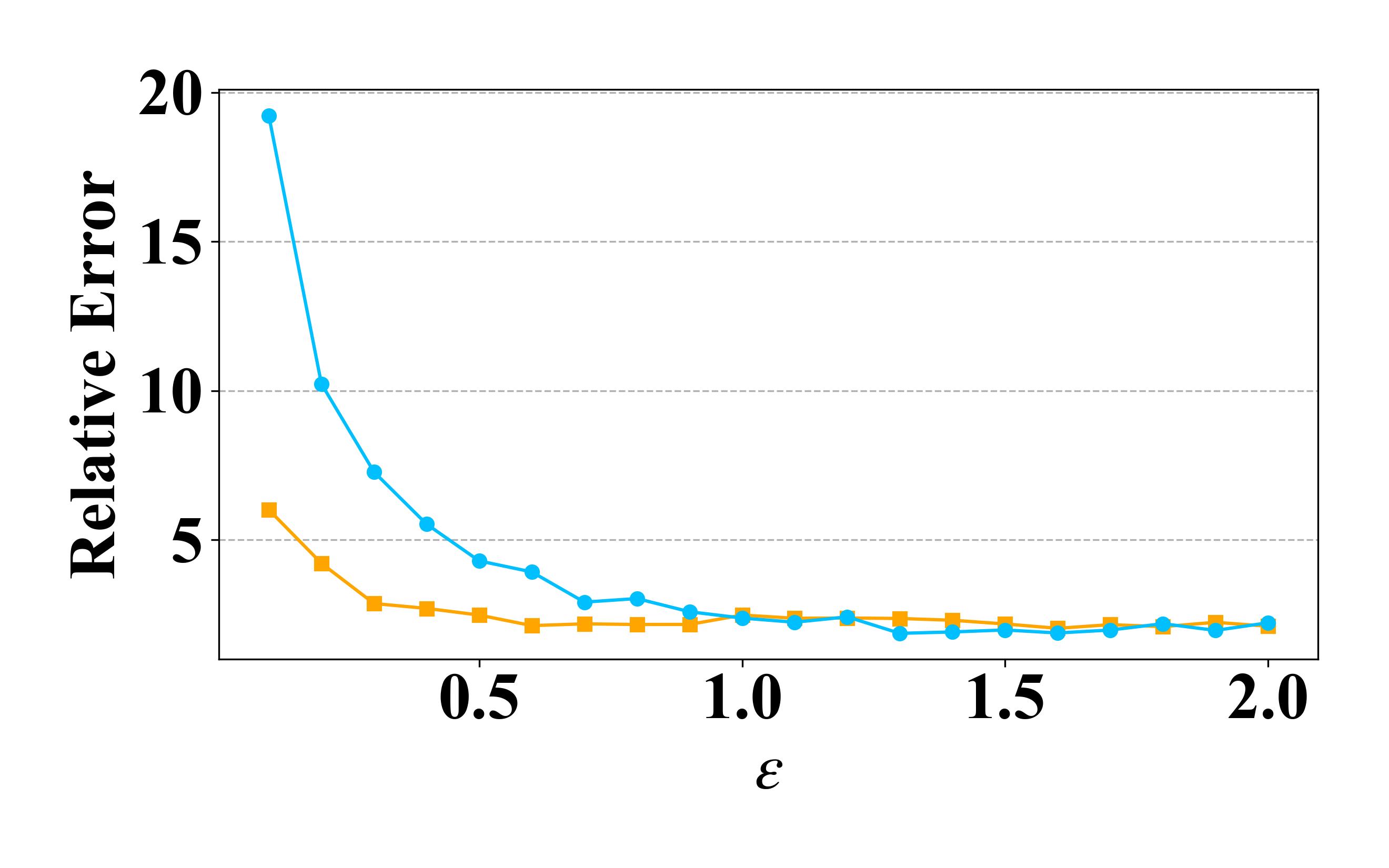}
        \caption{BAGraph-m10}
    \end{subfigure}
    \hfill
    \begin{subfigure}[b]{0.235\textwidth}
        \includegraphics[width=\textwidth]{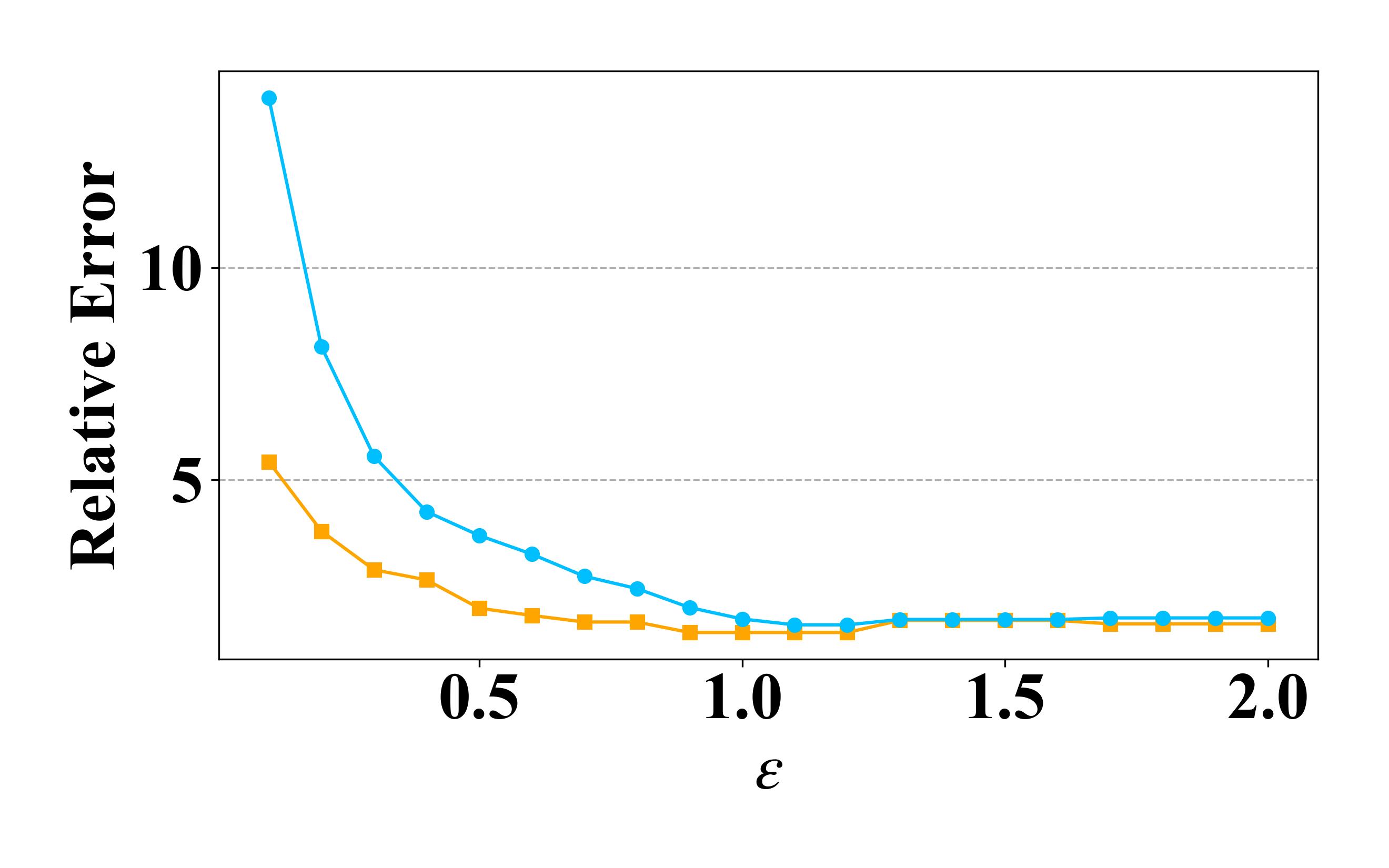}
        \caption{BAGraph-m50}
    \end{subfigure}
    
    \vspace{1em} 
    
    \begin{subfigure}[b]{0.235\textwidth}
        \includegraphics[width=\textwidth]{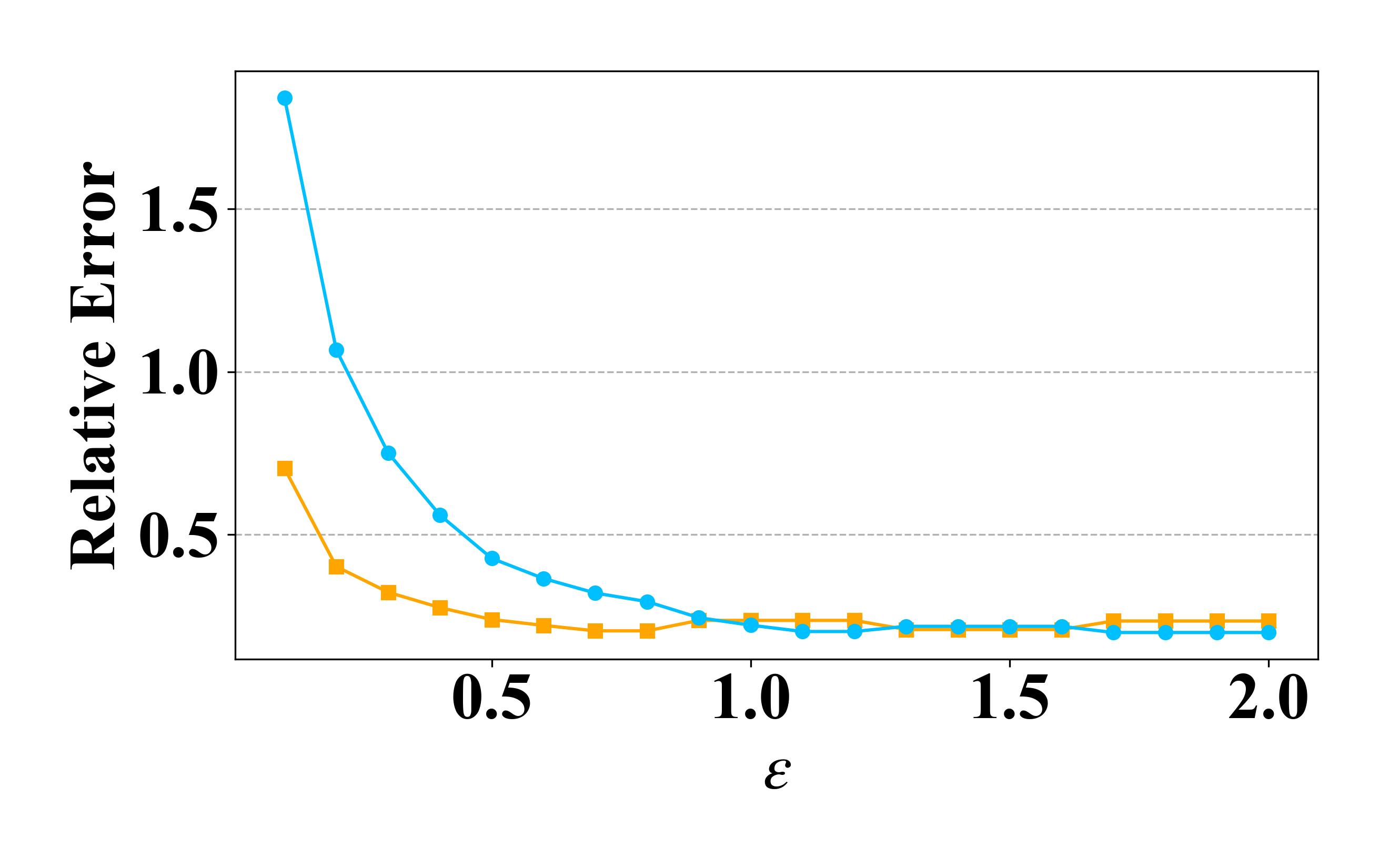}
        \caption{BAGraph-m100}
    \end{subfigure}
    \hfill
    \begin{subfigure}[b]{0.235\textwidth}
        \includegraphics[width=\textwidth]{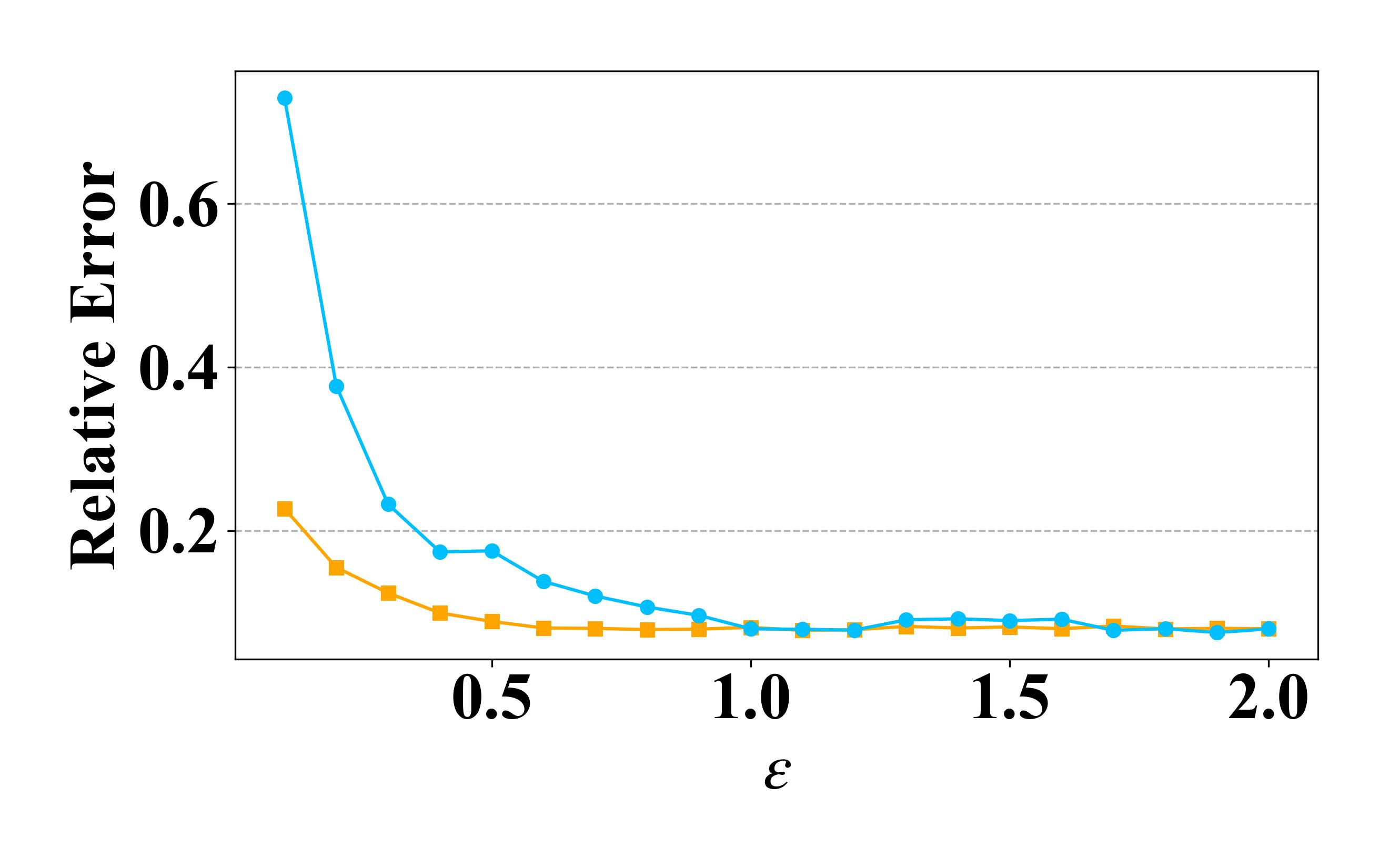}
        \caption{Facebook}
    \end{subfigure}

    \vspace{1em} 
    
    \begin{subfigure}[b]{0.235\textwidth}
        \includegraphics[width=\textwidth]{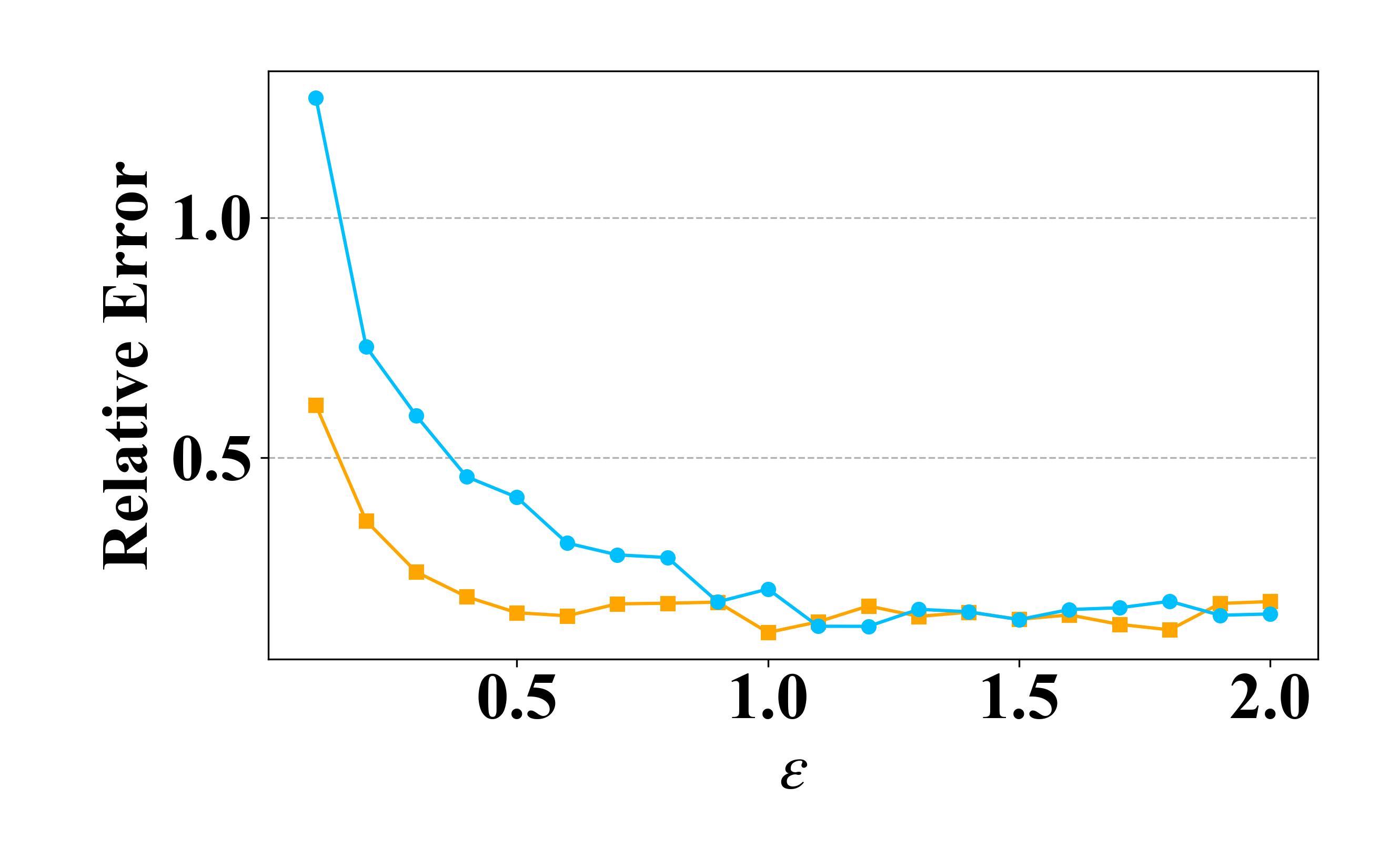}
        \caption{Deezer}
    \end{subfigure}
    \hfill
    \begin{subfigure}[b]{0.235\textwidth}
        \includegraphics[width=\textwidth]{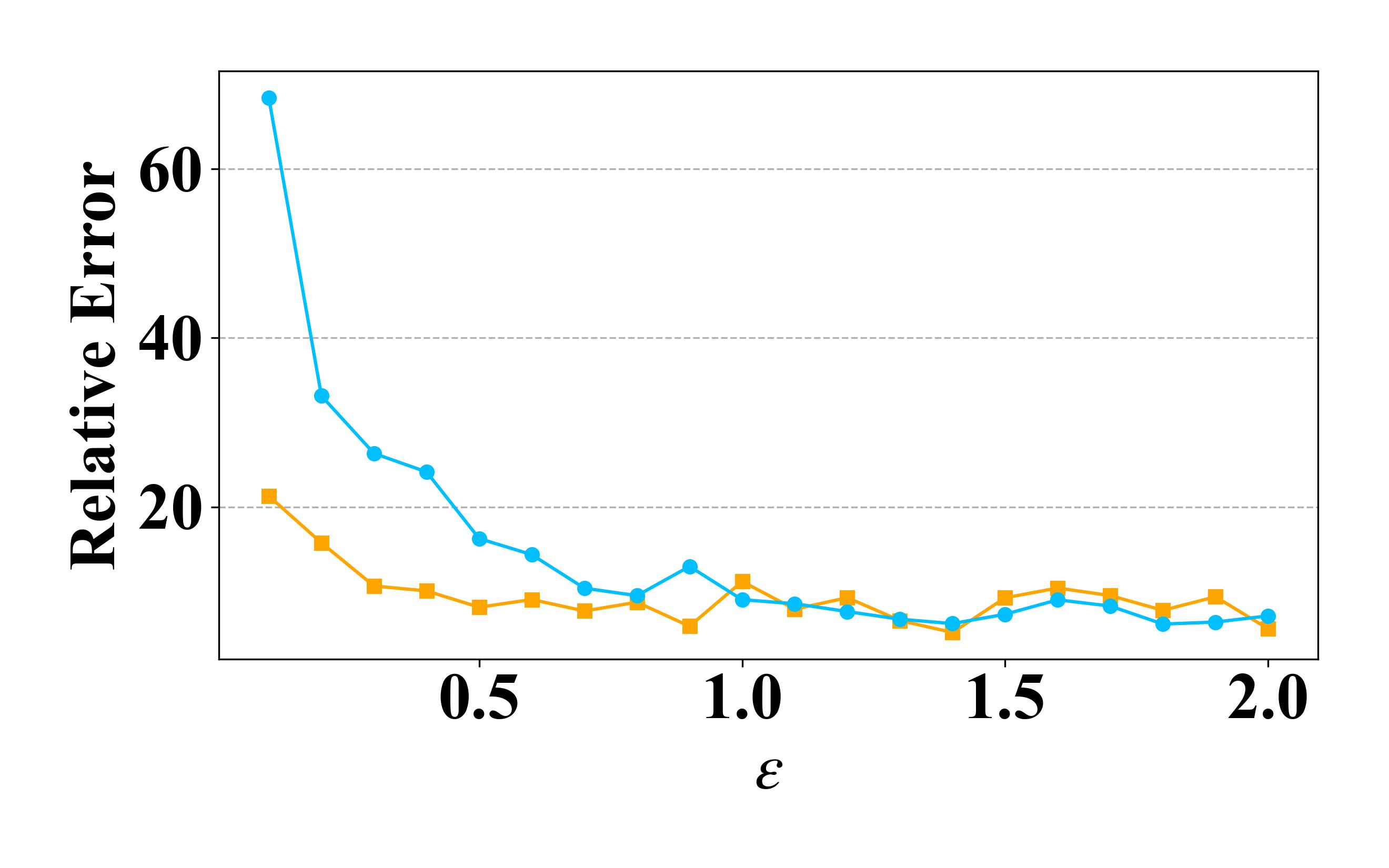}
        \caption{GitHub}
    \end{subfigure}

    \caption{Numerical bound vs. closed-form bound in $\mathbf{DeShuffle_{ru}}$.}
    \label{fig4}
\end{figure}

Figure \ref{fig4} shows that when $\varepsilon$ is small, $\mathbf{Shuffle_{ru}}$ equipped with the numerical upper bound achieves a lower relative error, and when $\varepsilon \ge 1$, the relative errors of both methods are nearly the same. This is because, in both cases, the corresponding local privacy budget $\varepsilon_0$ approaches its maximum value of $\log \left( \frac{n}{16\log \left( 2/\delta \right)} \right) $ when $\varepsilon \ge 1$. In summary, $\mathbf{Shuffle_{ru}}$ with the numerical upper bound generally outperforms the version with the closed-form upper bound.

\section{Discussion}

In this work, we study the estimation of network assortativity under DP for the first time and propose three algorithms $\mathbf{Local_{ru}}$, $\mathbf{Shuffle_{ru}}$ and $\mathbf{Decentral_{ru}}$. As mentioned above, the techniques built in the development and analysis of algorithms have a wide range of applications. For example, these algorithms can be adapted to analyze the degree-degree correlation distribution of networks. Although we focus only on simple networks, the proposed algorithms can be slightly modified to achieve the same task on other types of networks including weighted networks, directed networks. Hence, we would like to see more applications based on our work in the near future.     

On the other hand, our work is just the tip of the iceberg. There is still room for further improvement. For instance, we make use of some relatively rough methods when determining bounds for some quantities including MSE, time and space complexity. Therefore, it is of great interest to capture new bounds using advanced techniques, which is left as our next move. What's more, it is a notable problem of how to derive the tight upper bound of MSE for DP-algorithm addressing the concern we are discussing in this work. In addition, it is also interesting to propose some new schemes for the purpose of accurately estimating network assortativity under LDP.      

It is well known that there is a long history of network structure analysis \cite{newman2018networks,newman2003structure}. At the same time, it has received more attention to analyze structure of network in the requirement of privacy protection, such as subgraph counting \cite{imola2021locally}, community detection \cite{fu2024community} and graph synthesis \cite{qin2017generating}. This work studies another important structural property of network under DP. With the increased awareness of privacy protection, there is growing emphasis on data privacy disclosure issues. Network, as a class of fundamental yet significant structure depicting data, will continue to gain more attention in the field of privacy-preserving data analysis. In the future, therefore, we explore other structural properties of network from the viewpoint of privacy protection.

\section{Conclusion}

In summary, we are the first to consider the problem of how to accurately estimate network assortativity under the demand of privacy protection. Specifically, we propose three DP-based algorithms, i.e., $\mathbf{Local_{ru}}$, $\mathbf{Shuffle_{ru}}$ and $\mathbf{Decentral_{ru}}$, to address this concern in two distinct scenarios. With rigorous theoretical analysis, we show that the proposed algorithms yield an unbiased estimation for network assortativity. Also, we make use of MSE to measure the steadiness of algorithms. At the same time, we determine time and space complexity of three algorithms. Additionally, we conduct extensive computer simulations, which demonstrate that experimental evaluations are in line with theoretical analysis. At last, we declare some potential applications of the light shed in the development of our algorithms and point out our next move.  

\begin{acks}
The research was supported by the National Natural Science Foundation of China No. 62403381, the Fundamental Research Funds for the Central Universities No. G2023KY05105 and the Key Research and Development Plan of Shaanxi Province No. 2024GX-YBXM-021.
\end{acks}


\bibliographystyle{ACM-Reference-Format}
\bibliography{sample}

\section*{Supplementary Material}

Below are more details that are omitted in the main text due to the space limitation. 


\subsection*{Proof of Theorem 4.2}

    \textbf{Theorem 4.2} \emph{The estimate $\hat{q}_{ru}\left( G\right)$ produced by $\mathbf{Local_{ru}}$ is unbiased, i.e., $\mathbb{E}\left( \hat{q}_{ru}\left( G\right)\right)=q_{ru}\left( G\right)$.}

\setcounter{equation}{0}                                                                                                        \renewcommand\theequation{A.\arabic{equation}}   

Before beginning with the detailed proof, we need to introduce a lemma as below. 

\begin{lemma}[\cite{kozubowski2010multitude}]
    Given a random variable $X \sim \text{Lap}\left( x,b\right)$, then
    \begin{equation}
        \mathbb{E}\left( X^r\right)=\sum_{k=0}^r{\left\{ \frac{1}{2} \left[ 1+\left(-1\right)^k\right] \frac{r!}{\left( r-k\right) !}b^k x^{r-k} \right\}}.
    \end{equation}
    \label{lemma:15}
\end{lemma}

Now, let us give the proof of Theorem 4.2.
\begin{proof}
    It is clear to see that $\tilde{a}_{i,j}$ is in fact a Bernoulli random
    variable and $\tilde{d}_i \sim \text{Lap}\left( d_i,\frac{1}{\varepsilon_2} \right)$. Then, we have
\begin{subequations}
\label{eq:whole}
\begin{eqnarray}
\mathbb{E}\left( \tilde{a}_{i,j}\right)=a_{i,j}\left( 1-p \right) +\left( 1-a_{i,j} \right)p,\label{subeq:A-1-1}
\end{eqnarray}
\begin{equation}
\mathbb{E}\left( \tilde{d}_i^2\right)=d_i^2+\frac{2}{\varepsilon_2^2}.\label{subeq:A-1-2} \quad \text{(by Lemma \ref{lemma:15})}
\end{equation}
\begin{equation}
\mathbb{E}\left( \tilde{d}_{i}^{4} \right) =d_{i}^{4}+\frac{12}{\varepsilon_2 ^2}d_{i}^{2}+\frac{24}{\varepsilon_2 ^4} \quad \text{(by Lemma \ref{lemma:15})}
\end{equation}
\end{subequations}

    Next, we move to the proof of unbiasedness of $\hat{q}_{ru}\left( G\right)$. First, we obtain

    \begin{equation}
        \begin{aligned}
        \mathbb{E}\left( X_1 \right) &=\sum_{i=1}^n{\sum_{j=1}^{i-1}{\frac{\left( E\left( \tilde{a}_{i,j} \right) -p \right) E\left( \tilde{d}_i \right) E\left( \tilde{d}_j \right)}{1-2p}}} \\
        &=\sum_{i=1}^n{\sum_{j=1}^{i-1}{\frac{\left[ a_{i,j}\left( 1-p \right) +\left( 1-a_{i,j} \right)p -p \right] d_id_j}{1-2p}}} \\
        &=\sum_{i=1}^n{\sum_{j=1}^{i-1}{a_{i,j}d_id_j}}\\
        &=\sum_{\left( v_i,v_j \right) \in E}{d_id_j}.
        \end{aligned}
    \end{equation}
Similarly, it is easy to check 
\begin{equation}
    \begin{aligned}
        \mathbb{E}\left( Y_1 \right) &=\mathbb{E}\left[ \left( \frac{1}{2}\sum_{i=1}^n{\tilde{d}_{i}^{2}}-\frac{n+2}{\varepsilon _{2}^{2}} \right) ^2-\frac{5n+4}{\varepsilon _{2}^{4}} \right] \\
        &=\mathbb{V}\left( \frac{1}{2}\sum_{i=1}^n{\tilde{d}_{i}^{2}}-\frac{n+2}{\varepsilon _{2}^{2}} \right) -\frac{5n+4}{\varepsilon _{2}^{4}} \\
        &\quad +\left[ \mathbb{E}\left( \frac{1}{2}\sum_{i=1}^n{\tilde{d}_{i}^{2}}-\frac{n+2}{\varepsilon _{2}^{2}} \right) \right] ^2\\
        &=\frac{1}{4}\sum_{i=1}^n{\left[ \mathbb{E}\left( \tilde{d}_{i}^{4} \right) -\left[ \mathbb{E}\left( \tilde{d}_{i}^{2} \right) \right] ^2 \right]}-\frac{5n+4}{\varepsilon _{2}^{4}} \\
        &\quad +\left[ \frac{1}{2}\sum_{i=1}^n{\mathbb{E}\left( \tilde{d}_{i}^{2} \right)}-\frac{n+2}{\varepsilon _{2}^{2}} \right] ^2\\
        &=\frac{1}{4}\sum_{i=1}^n{\left( \frac{8}{\varepsilon _{2}^{2}}d_{i}^{2}+\frac{20}{\varepsilon _{2}^{4}} \right)}-\frac{5n+4}{\varepsilon _{2}^{4}} \\
        &\quad +\left[ \frac{1}{2}\sum_{i=1}^n{\left( \frac{2}{\varepsilon _{2}^{2}}+d_{i}^{2} \right)}-\frac{n+2}{\varepsilon _{2}^{2}} \right] ^2\\
        &=\frac{2}{\varepsilon _{2}^{2}}\sum_{i=1}^n{d_{i}^{2}}+\frac{5n}{\varepsilon _{2}^{4}}+\left[ \frac{1}{2}\sum_{i=1}^n{d_{i}^{2}}-\frac{2}{\varepsilon _{2}^{2}} \right] ^2-\frac{5n+4}{\varepsilon _{2}^{4}} \\
        &=\left[ \sum_{\left( v_i,v_j \right) \in E}{\frac{1}{2}\left( d_i+d_j \right)} \right] ^2.
    \end{aligned}
    \label{eq:proof_Y1}
    \end{equation}
To sum up, we verify that 
    \begin{equation}
        \mathbb{E}\left[ \hat{q}_{ru}\left( G \right) \right] =\frac{\mathbb{E}\left( X_1 \right)}{M}-\frac{\mathbb{E}\left( Y_1 \right)}{M^2}=q_{ru}\left( G \right).
    \end{equation}
\end{proof}

\subsection*{Proof of Theorem 4.3}

\setcounter{equation}{0}                                                                                                        \renewcommand\theequation{B.\arabic{equation}} 

    \textbf{Theorem 4.3}  \emph{When $\varepsilon_1$, $\varepsilon_2$ are constants, the estimate $\hat{q}_{ru}\left( G\right)$ produced by $\mathbf{Local_{ru}}$ provides the following utility guarantee:}
    \begin{align}
        \mathrm{MSE}\left( \frac{n^3d_{\max}^{2}+n^2d_{\max}^{4}}{M^2}+\frac{n^3d_{\max}^{6}}{M^4} \right).
    \end{align}

Let us introduce a lemma to succeed in verifying Theorem 4.3. 
\begin{lemma}[\cite{imola2022differentially}]
    Let $x_1,x_2$ be two random variables, then $\mathbb{V}\left( x_1+x_2\right) \le 4\mathrm{max}\left\{ \mathbb{V}\left( x_1\right), \mathbb{V}\left( x_2 \right) \right\}$.
    \label{lemma:16}
\end{lemma}

From now on, we show the detailed proof of Theorem 4.3 as follows. 
\begin{proof}
    Due to $\mathbb{V}\left( a_{i,j}\right)=\mathbb{E}\left( a_{i,j}^2\right) -\left[ \mathbb{E}\left( a_{i,j}\right) \right] ^2=p(1-p)$, we can obtain
    \begin{equation}\label{eq:B-1}
    \begin{aligned}
        \mathbb{E}\left( \hat{a}_{i,j}^2\right)&=\mathbb{V}(\hat{a}_{i,j})+\left[\mathbb{E}(\hat{a}_{i,j})\right]^2 \\
        &=\frac{p(1-p)}{(1-2p)^2}+a_{i,j}^2 \\
        &=\frac{p(1-p)}{(1-2p)^2}+a_{i,j}.
    \end{aligned}
    \end{equation}
    From Theorem 4.2, it is clear to the eye that the estimate $\hat{q}_{ru}\left( G\right)$ produced by $\mathbf{Local_{ru}}$ is unbiased. By the bias-variance decomposition \cite{murphy2012machine}, the mean squared error (MSE) of $\hat{q}_{ru}\left( G\right)$ is equal to its variance. Let $P=\frac{X_1}{M}$ and $Q=-\frac{Y_1}{M^2}$, then
    \begin{equation}\label{eq:B-2}
    \begin{aligned}
        \text{MSE}\left( \hat{q}_{ru}\left( G \right) \right) &=\mathbb{V}\left( \hat{q}_{ru}\left( G \right) \right) =\mathbb{V}\left( P+Q \right) \\
        &\le 4\max \left\{ \mathbb{V}\left( P \right) ,\mathbb{V}\left( Q \right) \right\}. \\
        &\qquad \qquad \text{(by Lemma \ref{lemma:16})}
    \end{aligned}
    \end{equation}

Now, our task is to calculate $\mathbb{V}\left( P \right)$ and $\mathbb{V}\left( Q \right)$ separately. \\
    
    For $\mathbb{V}\left( P \right)$, since $\mathbb{V}\left( P \right)=M^{-2}\mathbb{V}\left( X_1 \right)$, we need to focus on calculation of $\mathbb{V}\left( X_1 \right)$. For ease of presentation, we define $\tilde{B}_{i,j}=\frac{\left( \tilde{a}_{i,j}-p \right) \tilde{d}_i\tilde{d}_j}{1-2p}$, then \\
    \begin{equation}\label{eq:B-3}
    \begin{aligned}
        &\mathbb{V}\left( X_1 \right) =\mathbb{V}\left( \sum_{i=2}^n{\sum_{j=1}^{i-1}{\tilde{B}_{i,j}}} \right) \\
        &=\sum_{i=2}^n{\mathbb{V}\left( \sum_{j=1}^{i-1}{\tilde{B}_{i,j}} \right)}  +\sum_{2\le k,l\le n,k\ne l}{\text{Cov}\left( \sum_{j=1}^{k-1}{\tilde{B}_{k,j}},\sum_{j=1}^{l-1}{\tilde{B}_{l,j}} \right)} \\
        &=\sum_{i=2}^n{\sum_{j=1}^{i-1}{\mathbb{V}\left( \tilde{B}_{i,j} \right)}} +\sum_{i=2}^n{\sum_{1\le k,l\le i-1,k\ne l}{\text{Cov}\left( \tilde{B}_{i,k},\tilde{B}_{i,l} \right)}} \\
        &\quad +\sum_{2\le k,l\le n,k\ne l}{\sum_{j=1}^{k-1}{\sum_{t=1}^{l-1}{\text{Cov}\left( \tilde{B}_{k,j},\tilde{B}_{l,t} \right)}}} \\
        &=P_1+P_2+P_3,
    \end{aligned}
    \end{equation}
    where $P_1=\sum_{i=2}^n{\sum_{j=1}^{i-1}{\mathbb{V}\left( \tilde{B}_{i,j} \right)}}$, $P_2=\sum_{i=2}^n{\sum_{1\le k,l\le i-1,k\ne l}{\text{Cov}\left( \tilde{B}_{i,k},\tilde{B}_{i,l} \right)}}$ and $P_3=\sum_{2\le k,l\le n,k\ne l}{\sum_{j=1}^{k-1}{\sum_{t=1}^{l-1}{\text{Cov}\left( \tilde{B}_{k,j},\tilde{B}_{l,t} \right)}}}$. \\

    Next, we calculate $P_1$, $P_2$ and $P_3$ respectively, and obtain 
\begin{subequations}
\label{eq:whole}
\begin{eqnarray}
 \begin{aligned}
        P_1&=\sum_{i=2}^n{\sum_{j=1}^{i-1}{\mathbb{V}\left( \tilde{B}_{i,j} \right)}}\\
        &=\sum_{i=2}^n{\sum_{j=1}^{i-1}{\left[ \mathbb{E}\left[ \left( \tilde{B}_{i,j} \right) ^2 \right] -\left[ \mathbb{E}\left( \tilde{B}_{i,j} \right) \right] ^2 \right]}} \\
        &=\sum_{i=2}^n{\sum_{j=1}^{i-1}{\left[ \left( \frac{p\left( 1-p \right)}{\left( 1-2p \right) ^2}+a_{i,j} \right) \left( \frac{2}{\varepsilon _{2}^{2}}+d_{i}^{2} \right) \left( \frac{2}{\varepsilon _{2}^{2}}+d_{j}^{2} \right) \right]}} \\
        &\quad -\sum_{i=2}^n{\sum_{j=1}^{i-1}}a_{i,j}d_{i}^{2}d_{j}^{2}\\
        &=\sum_{i=2}^n{\sum_{j=1}^{i-1}{\left[ \frac{p\left( 1-p \right)}{\left( 1-2p \right) ^2}\left( \frac{2}{\varepsilon _{2}^{2}}+d_{i}^{2} \right) \left( \frac{2}{\varepsilon _{2}^{2}}+d_{j}^{2} \right)  \right]}} \\
        &\quad+\sum_{i=2}^n{\sum_{j=1}^{i-1}}\frac{2}{\varepsilon _{2}^{2}}a_{i,j}\left( \frac{2}{\varepsilon _{2}^{2}}+d_{i}^{2}+d_{j}^{2} \right)\\
        &\le \sum_{i=2}^n{\sum_{j=1}^{i-1}{\left[ \frac{p\left( 1-p \right)}{\left( 1-2p \right) ^2}\left( \frac{2}{\varepsilon _{2}^{2}}+d_{\max}^{2} \right) \left( \frac{2}{\varepsilon _{2}^{2}}+d_{\max}^{2} \right)  \right]}} \\
        &\quad +\sum_{i=2}^n{\sum_{j=1}^{i-1}}\frac{2}{\varepsilon _{2}^{2}}\left( \frac{2}{\varepsilon _{2}^{2}}+2d_{\max}^{2} \right)\\
        &=O\left( n^2d_{\max}^{4} \right), 
    \end{aligned}\label{subeq:B-4-1}
\end{eqnarray}
\begin{equation}
  \begin{aligned}
        P_2&=\sum_{i=2}^n{\sum_{1\le k,l\le i-1,k\ne l}{\text{Cov}\left( \tilde{B}_{i,k},\tilde{B}_{i,l} \right)}} \\
        &=\sum_{i=2}^n{\sum_{1\le k,l\le i-1,k\ne l}{\left[ \mathbb{E}\left( \tilde{B}_{i,k}\tilde{B}_{i,l} \right) -\mathbb{E}\left( \tilde{B}_{i,k} \right) \mathbb{E}\left( \tilde{B}_{i,l} \right) \right]}} \\
        &=\sum_{i=2}^n{\sum_{1\le k,l\le i-1,k\ne l}{\left[ a_{i,k}a_{i,l}\left( \frac{2}{\varepsilon _{2}^{2}}+d_{i}^{2} \right) d_kd_l \right]}} \\
        &\quad -\sum_{i=2}^n{\sum_{1\le k,l\le i-1,k\ne l}}a_{i,k}a_{i,l}d_{i}^{2}d_kd_l\\
        &=\frac{2}{\varepsilon _{2}^{2}}\sum_{i=2}^n{\sum_{1\le k,l\le i-1,k\ne l}{a_{i,k}a_{i,l}d_kd_l}} \\
        &\le \frac{2}{\varepsilon _{2}^{2}}\sum_{i=2}^n{\sum_{1\le k,l\le i-1,k\ne l}{d_{\max}^{2}}} \\
        &=O\left( n^3d_{\max}^{2} \right) ,
    \end{aligned}\label{subeq:B-4-2} 
\end{equation}

\begin{equation}
 \begin{aligned}
        P_3&=\sum_{2\le k,l\le n,k\ne l}{\sum_{j=1}^{k-1}{\sum_{t=1}^{l-1}{\text{Cov}\left[ \tilde{B}_{k,j},\tilde{B}_{l.t} \right]}}} \\
        &=\sum_{2\le k,l\le n,k\ne l}{\sum_{j=1}^{k-1}{\sum_{t=1}^{l-1} \mathbb{E}\left( \tilde{B}_{k,j}\tilde{B}_{l.t} \right) }} \\
        &\quad -\sum_{2\le k,l\le n,k\ne l}{\sum_{j=1}^{k-1}{\sum_{t=1}^{l-1}}}\mathbb{E}\left( \tilde{B}_{k,j} \right) \mathbb{E}\left( \tilde{B}_{l.t} \right)\\
        &=\sum_{2\le k,l\le n,k\ne l}{\sum_{j=1}^{k-1}{\sum_{t=1}^{l-1}{ a_{k,j}a_{l,t}\mathbb{E}\left( \tilde{d}_k\tilde{d}_j\tilde{d}_l\tilde{d}_t \right)}}} \\
        &\quad -\sum_{2\le k,l\le n,k\ne l}{\sum_{j=1}^{k-1}{\sum_{t=1}^{l-1}}}a_{k,j}d_kd_ja_{l,t}d_ld_t\\
        &=2\sum_{2\le k<l\le n,k\ne l}{\sum_{j=1}^{k-1}{\left[ \frac{2}{\varepsilon _{2}^{2}}a_{k,j}a_{l,t}d_kd_l+\frac{2}{\varepsilon _{2}^{2}}a_{k,j}a_{l,t}d_jd_l \right]}} \\
        &=\frac{4}{\varepsilon _{2}^{2}}\sum_{2\le k<l\le n,k\ne l}{\sum_{j=1}^{k-1}{\left[ a_{k,j}a_{l,t}d_kd_l+a_{k,j}a_{l,t}d_jd_l \right]}} \\
        &\le \frac{4}{\varepsilon _{2}^{2}}\sum_{2\le k<l\le n,k\ne l}{\sum_{j=1}^{k-1}{2d_{\max}^{2}}} \\
        &=O\left( n^3d_{\max}^{2} \right). 
    \end{aligned}\label{subeq:B-4-3} 
\end{equation}
\end{subequations} \\

    Thus, \\
    \begin{equation}\label{eq:B-5}
        \mathbb{V}\left( X_1\right)=O\left( n^3d_{max}^3+n^2d_{max}^4\right).
    \end{equation} \\

This leads to the following expression \\ 

\begin{equation}\label{eq:B-6}
    \begin{aligned}
        \mathbb{V}\left( P\right)&= O\left( \frac{n^3d_{max}^3+n^2d_{max}^4}{M^2}\right).
    \end{aligned}
    \end{equation} \\
    
    Below we move on to the calculation of $\mathbb{V}\left( Q\right)$. By Lemma \ref{lemma:15}, we first derive 
    
\begin{subequations}
\label{eq:whole}
\begin{eqnarray}
\mathbb{E}\left( \tilde{d}_{i}^{6} \right) =d_{i}^{6}+\frac{30}{\varepsilon _{2}^{2}}d_{i}^{4}+\frac{360}{\varepsilon _{2}^{4}}d_{i}^{2}+\frac{720}{\varepsilon _{2}^{6}},\label{subeq:B-7-1}
\end{eqnarray}
\begin{equation}
\mathbb{E}\left( \tilde{d}_{i}^{8} \right) =d_{i}^{8}+\frac{56}{\varepsilon _{2}^{2}}d_{i}^{6}+\frac{1680}{\varepsilon _{2}^{4}}d_{i}^{4}+\frac{20160}{\varepsilon _{2}^{6}}d_{i}^{2}+\frac{40320}{\varepsilon _{2}^{8}}.\label{subeq:B-7-2}
\end{equation}
\end{subequations} \\

To make further progress, we write \\
    \begin{equation}\label{eq:B-8}
    \begin{aligned}
        \mathbb{V}\left( \tilde{d}_{i}^{2} \right)&=\mathbb{E}\left( \tilde{d}_{i}^{4} \right) -\left[ \mathbb{E}\left( \tilde{d}_{i}^{2} \right) \right] ^2 \\
        &=\left( d_{i}^{4}+\frac{12}{\varepsilon _{2}^{2}}d_{i}^{2}+\frac{24}{\varepsilon _{2}^{4}} \right) -\left( d_{i}^{2}+\frac{2}{\varepsilon _{2}^{2}} \right) ^2 \\
        &=\frac{8}{\varepsilon _{2}^{2}}d_{i}^{2}+\frac{20}{\varepsilon _{2}^{4}},
    \end{aligned}
    \end{equation} \\
    and \\
    \begin{equation}\label{eq:B-9}
    \begin{aligned}
        \mathbb{V}\left( \tilde{d}_{i}^{4} \right) &=\mathbb{E}\left( \tilde{d}_{i}^{8} \right) -\left[ \mathbb{E}\left( \tilde{d}_{i}^{4} \right) \right] ^2 \\
        &=\frac{32}{\varepsilon _{2}^{2}}d_{i}^{6}+\frac{1488}{\varepsilon _{2}^{4}}d_{i}^{4}+\frac{19584}{\varepsilon _{2}^{6}}d_{i}^{2}+\frac{39744}{\varepsilon _{2}^{8}}.
    \end{aligned}
    \end{equation}
    
Armed with the consequences above, let us derive quantity $\mathbb{V}\left( Q \right) $ as below 
\begin{equation}\label{eq:B-10}
    \begin{aligned}
        &\mathbb{V}\left( Q \right) =\mathbb{V}\left( -\frac{Y_1}{M^2} \right) =\frac{1}{M^4}\mathbb{V}\left( Y_1 \right) \\
        &=\frac{1}{M^4}\mathbb{V}\left[ \left( \frac{1}{2}\sum_{i=1}^n{\tilde{d}_{i}^{2}}-\frac{n+2}{\varepsilon _{2}^{2}} \right) ^2-\frac{5n+4}{\varepsilon _{2}^{4}} \right] \\
        &=\frac{1}{M^4}\mathbb{V}\left[ \frac{1}{4}\left( \sum_{i=1}^n{\tilde{d}_{i}^{2}} \right) ^2-\frac{n+2}{\varepsilon _{2}^{2}}\sum_{i=1}^n{\tilde{d}_{i}^{2}} \right] \\
        &=\frac{1}{M^4}\mathbb{V}\left[ \frac{1}{4}\left( \sum_{i=1}^n{\tilde{d}_{i}^{4}}+\sum_{i=1}^n{\sum_{j=1,j\ne i}^n{\tilde{d}_{i}^{2}\tilde{d}_{j}^{2}}} \right) -\frac{n+2}{\varepsilon _{2}^{2}}\sum_{i=1}^n{\tilde{d}_{i}^{2}} \right] \\
        &=\frac{1}{M^4}\mathbb{E}\left[ \left( \frac{1}{4}\sum_{i=1}^n{\tilde{d}_{i}^{4}}+\frac{1}{4}\sum_{i=1}^n{\sum_{j=1,j\ne i}^n{\tilde{d}_{i}^{2}\tilde{d}_{j}^{2}}}-\frac{n+2}{\varepsilon _{2}^{2}}\sum_{i=1}^n{\tilde{d}_{i}^{2}} \right) ^2 \right] \\
        &\quad -\frac{1}{M^4}\left[ \mathbb{E}\left( \frac{1}{4}\sum_{i=1}^n{\tilde{d}_{i}^{4}}+\frac{1}{4}\sum_{i=1}^n{\sum_{j=1,j\ne i}^n{\tilde{d}_{i}^{2}\tilde{d}_{j}^{2}}}-\frac{n+2}{\varepsilon _{2}^{2}}\sum_{i=1}^n{\tilde{d}_{i}^{2}} \right) \right] ^2 \\
        &=\frac{1}{M^4}Q_1-\frac{1}{M^4}Q_2 ,
    \end{aligned}
\end{equation}
    where 
    \begin{align*}
        &Q_1=\mathbb{E}\left[ \left( \frac{1}{4}\sum_{i=1}^n{\tilde{d}_{i}^{4}}+\frac{1}{4}\sum_{i=1}^n{\sum_{j=1,j\ne i}^n{\tilde{d}_{i}^{2}\tilde{d}_{j}^{2}}}-\frac{n+2}{\varepsilon _{2}^{2}}\sum_{i=1}^n{\tilde{d}_{i}^{2}} \right) ^2 \right] ,\\
        &Q_2=\left[ \mathbb{E}\left( \frac{1}{4}\sum_{i=1}^n{\tilde{d}_{i}^{4}}+\frac{1}{4}\sum_{i=1}^n{\sum_{j=1,j\ne i}^n{\tilde{d}_{i}^{2}\tilde{d}_{j}^{2}}}-\frac{n+2}{\varepsilon _{2}^{2}}\sum_{i=1}^n{\tilde{d}_{i}^{2}} \right) \right] ^2.
    \end{align*}

    Analogously, we need to calculate $Q_1$ and $Q_2$ separately. In essence, it is easy to derive 
    \begin{subequations}
    \label{eq:whole}
    \begin{eqnarray}
        \begin{aligned}
        Q_1&=\mathbb{E}\left[ \left( \frac{1}{4}\sum_{i=1}^n{\tilde{d}_{i}^{4}}+\frac{1}{4}\sum_{i=1}^n{\sum_{j=1,j\ne i}^n{\tilde{d}_{i}^{2}\tilde{d}_{j}^{2}}}-\frac{n+2}{\varepsilon _{2}^{2}}\sum_{i=1}^n{\tilde{d}_{i}^{2}} \right) ^2 \right]  \\
        &=\frac{1}{16}\mathbb{E}\left[ \left( \sum_{i=1}^n{\tilde{d}_{i}^{4}} \right) ^2\right] +\frac{1}{16}\mathbb{E}\left[ \left( \sum_{i=1}^n{\sum_{j=1,j\ne i}^n{\tilde{d}_{i}^{2}\tilde{d}_{j}^{2}}} \right) ^2\right] \\
        &\quad +\frac{\left( n+2 \right) ^2}{\varepsilon _{2}^{4}}\mathbb{E}\left[ \left( \sum_{i=1}^n{\tilde{d}_{i}^{2}} \right) ^2\right] \\
        &\quad +\frac{1}{8}\mathbb{E}\left[ \left( \sum_{i=1}^n{\tilde{d}_{i}^{4}} \right) \left( \sum_{i=1}^n{\sum_{j=1,j\ne i}^n{\tilde{d}_{i}^{2}\tilde{d}_{j}^{2}}} \right) \right] \\
        &\quad -\frac{n+2}{2\varepsilon _{2}^{2}}\mathbb{E}\left[ \left( \sum_{i=1}^n{\tilde{d}_{i}^{4}} \right) \left( \sum_{i=1}^n{\tilde{d}_{i}^{2}} \right) \right] \\
        &\quad -\frac{n+2}{2\varepsilon _{2}^{2}}\mathbb{E}\left[ \left( \sum_{i=1}^n{\sum_{j=1,j\ne i}^n{\tilde{d}_{i}^{2}\tilde{d}_{j}^{2}}} \right) \left( \sum_{i=1}^n{\tilde{d}_{i}^{2}} \right) \right] \\
        &=\frac{1}{16}\sum_{i=1}^n{\mathbb{E}\left( \tilde{d}_{i}^{8} \right)}+\frac{3}{16}\sum_{i=1}^n{\sum_{j=1,j\ne i}^n{\mathbb{E}\left( \tilde{d}_{i}^{4} \right) \mathbb{E}\left( \tilde{d}_{j}^{4} \right)}}\\
        &\quad +\frac{1}{4}\sum_{i=1}^n{\sum_{j=1,j\ne i}^n{\mathbb{E}\left( \tilde{d}_{i}^{6} \right) \mathbb{E}\left( \tilde{d}_{j}^{2} \right)}}\\
        &\quad +\frac{3}{8}\sum_{i=1}^n{\sum_{j=1,j\ne i}^n{\sum_{k=1,k\ne i,j}^n{\mathbb{E}\left( \tilde{d}_{i}^{4} \right) \mathbb{E}\left( \tilde{d}_{j}^{2} \right) \mathbb{E}\left( \tilde{d}_{k}^{2} \right)}}}\\
        &\quad +\frac{1}{16}\sum_{i=1}^n{\sum_{j=1,j\ne i}^n{\sum_{k=1,k\ne i,j}^k{\sum_{s=1,s\ne i,j,k}^n{\mathbb{E}\left( \tilde{d}_{i}^{2} \right) \mathbb{E}\left( \tilde{d}_{j}^{2} \right) \mathbb{E}\left( \tilde{d}_{k}^{2} \right) \mathbb{E}\left( \tilde{d}_{s}^{2} \right)}}}}\\
        &\quad +\frac{\left( n+2 \right) ^2}{\varepsilon _{2}^{4}}\sum_{i=1}^n{\mathbb{E}\left( \tilde{d}_{i}^{4} \right)}+\frac{\left( n+2 \right) ^2}{\varepsilon _{2}^{4}}\sum_{i=1}^n{\sum_{j=1,j\ne i}^n{\mathbb{E}\left( \tilde{d}_{i}^{2} \right) \mathbb{E}\left( \tilde{d}_{j}^{2} \right)}}\\
        &\quad -\frac{n+2}{2\varepsilon _{2}^{2}}\sum_{i=1}^n{\mathbb{E}\left( \tilde{d}_{i}^{6} \right)}-\frac{3\left( n+2 \right)}{2\varepsilon _{2}^{2}}\sum_{i=1}^n{\sum_{j=1,j\ne i}^n{\mathbb{E}\left( \tilde{d}_{i}^{4} \right) \mathbb{E}\left( \tilde{d}_{j}^{2} \right)}}\\
        &-\frac{n+2}{2\varepsilon _{2}^{2}}\sum_{i=1}^n{\sum_{j=1,j\ne i}^n{\sum_{k=1,k\ne i,j}^n{\mathbb{E}\left( \tilde{d}_{i}^{2} \right) \mathbb{E}\left( \tilde{d}_{j}^{2} \right) \mathbb{E}\left( \tilde{d}_{k}^{2} \right)}}},
        \end{aligned}\label{subeq:B-11-1}
    \end{eqnarray} \\

\begin{equation}
    \begin{aligned}
        Q_2&=\left[ \mathbb{E}\left( \frac{1}{4}\sum_{i=1}^n{\tilde{d}_{i}^{4}}+\frac{1}{4}\sum_{i=1}^n{\sum_{j=1,j\ne i}^n{\tilde{d}_{i}^{2}\tilde{d}_{j}^{2}}}-\frac{n+2}{\varepsilon _{2}^{2}}\sum_{i=1}^n{\tilde{d}_{i}^{2}} \right) \right] ^2 \\
        &=\left[ \frac{1}{4}\sum_{i=1}^n{\mathbb{E}\left( \tilde{d}_{i}^{4} \right)}+\frac{1}{4}\sum_{i=1}^n{\sum_{j=1,j\ne i}^n{\mathbb{E}\left( \tilde{d}_{i}^{2} \right) \mathbb{E}\left( \tilde{d}_{j}^{2} \right)}}-\frac{n+2}{\varepsilon _{2}^{2}}\sum_{i=1}^n{\mathbb{E}\left( \tilde{d}_{i}^{2} \right)} \right] ^2 \\
        &=\frac{1}{16}\left[ \sum_{i=1}^n{\mathbb{E}\left( \tilde{d}_{i}^{4} \right)} \right] ^2+\frac{1}{16}\left[ \sum_{i=1}^n{\sum_{j=1,j\ne i}^n{\mathbb{E}\left( \tilde{d}_{i}^{2} \right) \mathbb{E}\left( \tilde{d}_{j}^{2} \right)}} \right] ^2\\
        &\quad +\frac{\left( n+2 \right) ^2}{\varepsilon _{2}^{4}}\left[ \sum_{i=1}^n{\mathbb{E}\left( \tilde{d}_{i}^{2} \right)} \right] ^2\\
        &\quad +\frac{1}{8}\left[ \sum_{i=1}^n{\mathbb{E}\left( \tilde{d}_{i}^{4} \right)} \right] \left[ \sum_{i=1}^n{\sum_{j=1,j\ne i}^n{\mathbb{E}\left( \tilde{d}_{i}^{2} \right) \mathbb{E}\left( \tilde{d}_{j}^{2} \right)}} \right]\\
        &\quad -\frac{n+2}{2\varepsilon _{2}^{2}}\left[ \sum_{i=1}^n{\mathbb{E}\left( \tilde{d}_{i}^{4} \right)} \right] \left[ \sum_{i=1}^n{\mathbb{E}\left( \tilde{d}_{i}^{2} \right)} \right] \\
        &\quad -\frac{n+2}{2\varepsilon _{2}^{2}}\left[ \sum_{i=1}^n{\sum_{j=1,j\ne i}^n{\mathbb{E}\left( \tilde{d}_{i}^{2} \right) \mathbb{E}\left( \tilde{d}_{j}^{2} \right)}} \right] \left[ \sum_{i=1}^n{\mathbb{E}\left( \tilde{d}_{i}^{2} \right)} \right]\\
        &=\frac{1}{16}\sum_{i=1}^n{\left[ \mathbb{E}\left( \tilde{d}_{i}^{4} \right) \right] ^2}+\frac{1}{16}\sum_{i=1}^n{\sum_{j=1,j\ne i}^n{\mathbb{E}\left( \tilde{d}_{i}^{4} \right) \mathbb{E}\left( \tilde{d}_{j}^{4} \right)}}\\
        &\quad +\frac{1}{4}\sum_{i=1}^n{\sum_{j=1,j\ne i}^n{\mathbb{E}\left( \tilde{d}_{i}^{4} \right) \mathbb{E}\left( \tilde{d}_{i}^{2} \right) \mathbb{E}\left( \tilde{d}_{j}^{2} \right)}}\\
        &\quad +\frac{1}{8}\sum_{i=1}^n{\sum_{j=1,j\ne i}^n{\sum_{k=1,k\ne i,j}^n{\mathbb{E}\left( \tilde{d}_{i}^{4} \right) \mathbb{E}\left( \tilde{d}_{j}^{2} \right) \mathbb{E}\left( \tilde{d}_{k}^{2} \right)}}}\\
        &\quad +\frac{1}{8}\sum_{i=1}^n{\sum_{j=1,j\ne i}^n{\left[ \mathbb{E}\left( \tilde{d}_{i}^{2} \right) \right] ^2\left[ \mathbb{E}\left( \tilde{d}_{j}^{2} \right) \right] ^2}}\\
        &\quad +\frac{1}{4}\sum_{i=1}^n{\sum_{j=1,j\ne i}^n{\sum_{k=1,k\ne i,j}^n{\left[ \mathbb{E}\left( \tilde{d}_{i}^{2} \right) \right] ^2\mathbb{E}\left( \tilde{d}_{j}^{2} \right) \mathbb{E}\left( \tilde{d}_{k}^{2} \right)}}}\\
        &\quad +\frac{1}{16}\sum_{i=1}^n{\sum_{j=1,j\ne i}^n{\sum_{k=1,k\ne i,j}^k{\sum_{s=1,s\ne i,j,k}^n{\mathbb{E}\left( \tilde{d}_{i}^{2} \right) \mathbb{E}\left( \tilde{d}_{j}^{2} \right) \mathbb{E}\left( \tilde{d}_{k}^{2} \right) \mathbb{E}\left( \tilde{d}_{s}^{2} \right)}}}}\\
        &\quad +\frac{\left( n+2 \right) ^2}{\varepsilon _{2}^{4}}\sum_{i=1}^n{\left[ \mathbb{E}\left( \tilde{d}_{i}^{2} \right) \right] ^2} +\frac{\left( n+2 \right) ^2}{\varepsilon _{2}^{4}}\sum_{i=1}^n{\sum_{j=1,j\ne i}^n{\mathbb{E}\left( \tilde{d}_{i}^{2} \right) \mathbb{E}\left( \tilde{d}_{j}^{2} \right)}}\\
        &\quad -\frac{n+2}{2\varepsilon _{2}^{2}}\sum_{i=1}^n{\mathbb{E}\left( \tilde{d}_{i}^{4} \right) \mathbb{E}\left( \tilde{d}_{i}^{2} \right)}-\frac{n+2}{2\varepsilon _{2}^{2}}\sum_{i=1}^n{\sum_{j=1,j\ne i}^n{\mathbb{E}\left( \tilde{d}_{i}^{4} \right) \mathbb{E}\left( \tilde{d}_{j}^{2} \right)}}\\
        &\quad -\frac{n+2}{\varepsilon _{2}^{2}}\sum_{i=1}^n{\sum_{j=1,j\ne i}^n{\mathbb{E}\left( \tilde{d}_{i}^{2} \right) \mathbb{E}\left( \tilde{d}_{i}^{2} \right) \mathbb{E}\left( \tilde{d}_{j}^{2} \right)}}\\
        &\quad -\frac{n+2}{2\varepsilon _{2}^{2}}\sum_{i=1}^n{\sum_{j=1,j\ne i}^n{\sum_{k=1,k\ne i,j}^n{\mathbb{E}\left( \tilde{d}_{i}^{2} \right) \mathbb{E}\left( \tilde{d}_{j}^{2} \right) \mathbb{E}\left( \tilde{d}_{k}^{2} \right)}}}.
    \end{aligned}\label{subeq:B-11-2}
\end{equation}
\end{subequations} \\

From Eqs.(\ref{subeq:B-11-1}) and (\ref{subeq:B-11-2}), we obtain 
\begin{equation}\label{eq:B-14}
    \begin{aligned}
        &\mathbb{V}\left( Q \right) =\frac{1}{16M^4}\sum_{i=1}^n{\left[ \mathbb{E}\left( \tilde{d}_{i}^{8} \right) -\left[ \mathbb{E}\left( \tilde{d}_{i}^{4} \right) \right] ^2 \right]}\\
        &\quad +\frac{1}{4M^4}\sum_{i=1}^n{\sum_{j=1,j\ne i}^n{\left[ \mathbb{E}\left( \tilde{d}_{i}^{6} \right) -\mathbb{E}\left( \tilde{d}_{i}^{4} \right) \mathbb{E}\left( \tilde{d}_{i}^{2} \right) \right] \mathbb{E}\left( \tilde{d}_{j}^{2} \right)}}\\
        &\quad +\frac{1}{8M^4}\sum_{i=1}^n{\sum_{j=1,j\ne i}^n{\left[ \mathbb{E}\left( \tilde{d}_{i}^{4} \right) \mathbb{E}\left( \tilde{d}_{j}^{4} \right) -\left[ \mathbb{E}\left( \tilde{d}_{i}^{2} \right) \right] ^2\left[ \mathbb{E}\left( \tilde{d}_{j}^{2} \right) \right] ^2 \right]}}\\
        &\quad +\frac{1}{4M^4}\sum_{i=1}^n{\sum_{j=1,j\ne i}^n{\sum_{k=1,k\ne i,j}^n{\left[ \mathbb{E}\left( \tilde{d}_{i}^{4} \right) -\left[ \mathbb{E}\left( \tilde{d}_{i}^{2} \right) \right] ^2 \right] \mathbb{E}\left( \tilde{d}_{j}^{2} \right) \mathbb{E}\left( \tilde{d}_{k}^{2} \right)}}}\\
        &\quad +\frac{\left( n+2 \right) ^2}{\varepsilon _{2}^{4}M^4}\sum_{i=1}^n{\left[ \mathbb{E}\left( \tilde{d}_{i}^{4} \right) -\left[ \mathbb{E}\left( \tilde{d}_{i}^{2} \right) \right] ^2 \right]}\\
        &\quad -\frac{n+2}{2\varepsilon _{2}^{2}M^4}\sum_{i=1}^n{\left[ \mathbb{E}\left( \tilde{d}_{i}^{6} \right) -\mathbb{E}\left( \tilde{d}_{i}^{4} \right) \mathbb{E}\left( \tilde{d}_{i}^{2} \right) \right]}\\
        &\quad -\frac{n+2}{\varepsilon _{2}^{2}M^4}\sum_{i=1}^n{\sum_{j=1,j\ne i}^n{\left[ \mathbb{E}\left( \tilde{d}_{i}^{4} \right) -\left[ \mathbb{E}\left( \tilde{d}_{i}^{2} \right) \right] ^2 \right] \mathbb{E}\left( \tilde{d}_{j}^{2} \right)}}\\
        &=O\left( \frac{n^3d_{\max}^{6}}{M^4} \right)         .
    \end{aligned}
\end{equation}


Finally, from Eqs.(\ref{eq:B-6}) and (\ref{eq:B-14}) it follows that 
\begin{equation}
    \text{MSE}= O\left( \frac{n^3d_{\max}^{2}+n^2d_{\max}^{4}}{M^2}+\frac{n^3d_{\max}^{6}}{M^4} \right).
\end{equation}
\end{proof}


\subsection*{Proof of Theorem 4.6}

    \textbf{Theorem 4.6} \emph{The estimate $\hat{q}_{ru}\left( G\right)$ produced by $\mathbf{Shuffle_{ru}}$ satisfies $\mathbb{E}\left( \hat{q}_{ru}\left( G\right)\right)=q_{ru}\left( G\right)$.}

\setcounter{equation}{0}                                                                                                        \renewcommand\theequation{C.\arabic{equation}}  
\begin{proof}
    First, we prove that $X_2$ is unbiased, i.e., 
    \begin{equation}
    \begin{aligned}
        \mathbb{E}\left( X_2 \right) &=\mathbb{E}\left( \sum_{i=2}^n{\hat{r}_i} \right)  \\
        &=\mathbb{E}\left[ \sum_{i=2}^n{d_i\sum_{j=1}^{i-1}{\frac{\left( \tilde{a}_{i,j}-p \right) \tilde{d}_j}{1-2p}}} \right] \\
        &=\sum_{i=2}^n{d_i\sum_{j=1}^{i-1}{\mathbb{E}\left( \frac{\tilde{a}_{i,j}-p}{1-2p} \right) \mathbb{E}\left( \tilde{d}_j \right)}} \\
        &=\sum_{i=2}^n{d_i\sum_{j=1}^{i-1}{a_{ij}d_j}} \\
        &=\sum_{\left( v_i,v_j \right) \in E}{d_id_j}.
    \end{aligned}
    \end{equation}

We now consider $Y_2$. Since $Y_2$ has the similar form as $Y_1$, and the unbiasedness of $Y_1$ has been already proven in the proof of Theorem 4.2, it follows that $Y_2$ is also unbiased. Thus, we have
    \begin{equation}
        \mathbb{E}\left( Y_2 \right) = \left[ \sum_{\left( v_i,v_j \right) \in E}{\frac{1}{2}\left( d_i+d_j \right)} \right] ^2.
    \end{equation}
    
Armed with the results above, we come to 
    \begin{equation}
        \mathbb{E}\left( \hat{q}_{ru}\left( G \right) \right) =q_{ru}\left( G \right). 
    \end{equation}
\end{proof}

\subsection*{Proof of Theorem 4.7}

\setcounter{equation}{0}                                                                                                        \renewcommand\theequation{D.\arabic{equation}} 

    \textbf{Theorem 4.7} \emph{When $\varepsilon$, $\delta$ are constants, $\alpha \in (0,1)$, $\varepsilon_0 = \log \left( n\right) + O\left( 1\right)$, the estimate $\hat{q}_{ru}\left( G\right)$ produced by $\mathbf{Shuffle_{ru}}$ provides the following utility guarantee:}
    \begin{align}
        \mathrm{MSE}\left( \hat{q}_{ru}\left( G \right) \right) &=O\left( \frac{n^{1+\alpha}d_{\max}^{4}}{M^2}+\frac{n^3d_{\max}^{2}}{\left( \log n \right) ^2M^2}+\frac{n^3d_{\max}^{6}}{\left( \log n \right) ^2M^4} \right) . 
    \end{align}

\begin{proof}
    Let $U=\frac{X_2}{M}$ and $W=-\frac{Y_2}{M^2}$, then the MSE of $\hat{q}_{ru}\left( G\right)$ by $\mathbf{Shuffle_{ru}}$ can be written as follows
    \begin{equation}
    \begin{aligned}
        \text{MSE}\left( \hat{q}_{ru}\left( G \right) \right) &=\mathbb{V}\left( \hat{q}_{ru}\left( G \right) \right)=\mathbb{V}\left( U+W \right) \\
        &\le 4\max \left\{ \mathbb{V}\left( U \right) ,\mathbb{V}\left( W \right) \right\}. \quad \text{(by Lemma \ref{lemma:16})}
    \end{aligned}
    \end{equation}

 We now need to calculate $\mathbb{V}\left( U \right)$ and $\mathbb{V}\left( W \right)$ separately. Since the expression for $W$ is similar to that of $Q$ in the proof of Theorem 4.3, this leads to $\mathbb{V}\left( W \right) \le O\left( \frac{n^3d_{\max}^{6}}{\left( \log n \right) ^2M^4}\right)$. Next, we only need to compute $\mathbb{V}\left( U \right)$ to establish the upper bound of $\text{MSE}\left( \hat{q}_{ru}\left( G \right) \right)$. For ease of presentation, we define $\tilde{C}_{i,j}=\frac{\left( \tilde{a}_{i,j}-p \right) \tilde{d}_j}{1-2p}$, then

 \begin{equation}
    \begin{aligned}
        &\mathbb{V}\left( U \right) =M^{-2}\mathbb{V}\left( \sum_{i=2}^n{d_i\sum_{j=1}^{i-1}{\tilde{C}_{i,j}}} \right) \\
        &=M^{-2}\sum_{i=2}^n{d_{i}^{2}\mathbb{V}\left( \sum_{j=1}^{i-1}{\tilde{C}_{i,j}} \right)}\\
        &\quad +2M^{-2}\sum_{2\le k<l\le n}{\text{Cov}\left( \sum_{j=1}^{k-1}{d_k\tilde{C}_{k,j}},\sum_{h=1}^{l-1}{d_l\tilde{C}_{l,h}}\right)}\\
        &=M^{-2}\sum_{i=2}^n{\frac{d_{i}^{2}}{\left( 1-2p \right) ^2}\sum_{j=1}^{i-1}{\left[ p\left( 1-p \right) d_{j}^{2}+\frac{2\left( 1-2p \right) ^2a_{i.j}}{\alpha ^2\varepsilon _{0}^{2}}+\frac{2p\left( 1-p \right)}{\alpha ^2\varepsilon _{0}^{2}} \right]}}\\
        &\quad +\frac{4M^{-2}}{\alpha ^2\varepsilon _{0}^{2}}\sum_{2\le k<l\le n}{\sum_{j=1}^{k-1}{a_{k,j}a_{l,j}d_kd_l}}\\
        &\le \frac{p\left( 1-p \right)}{\left( 1-2p \right) ^2M^2}\sum_{i=2}^n{\sum_{j=1}^{i-1}{d_{\max}^{4}}}+\frac{4}{\alpha ^2\varepsilon _{0}^{2}M^2}\sum_{2\le k<l\le n}{\left( k-1 \right) d_{\max}^{2}} \\
        &=\frac{p\left( 1-p \right) n\left( n-1 \right)}{2\left( 1-2p \right) ^2M^2}d_{\max}^{4}+\frac{2n\left( n-1 \right) \left( n-2 \right)}{3\alpha ^2\varepsilon _{0}^{2}M^2}d_{\max}^{2} \\
        &=\frac{n^{\alpha}\left( n-1 \right)}{M^2}d_{\max}^{4}+\frac{n\left( n-1 \right) \left( n-2 \right)}{\left( \log n \right) ^2M^2}d_{\max}^{2}\\
        &=O\left( \frac{n^{1+\alpha}d_{\max}^{4}}{M^2}+\frac{n^3d_{\max}^{2}}{\left( \log n \right) ^2M^2} \right) .
    \end{aligned}
    \end{equation}
    
    Therefore, we gain 
    \begin{equation}
        \text{MSE}\left( \hat{q}_{ru}\left( G \right) \right) = O\left( \frac{n^{1+\alpha}d_{\max}^{4}}{M^2}+\frac{n^3d_{\max}^{2}}{\left( \log n \right) ^2M^2}+\frac{n^3d_{\max}^{6}}{\left( \log n \right) ^2M^4} \right) . 
    \end{equation}
\end{proof}



\subsection*{Proof of Theorem 4.10}

    \textbf{Theorem 4.10} \emph{The estimate $\hat{q}_{ru}\left( G\right)$ produced by $\mathbf{Decentral_{ru}}$ satisfies $\mathbb{E}\left( \hat{q}_{ru}\left( G\right)\right)=q_{ru}\left( G\right)$.}
    \label{theorem:12}

\setcounter{equation}{0}                                                                                                        \renewcommand\theequation{E.\arabic{equation}} 

\begin{proof}
    Since $\tilde{T}_i \sim \text{Lap}\left( \frac{\hat{d}_{max}}{\varepsilon_2}\right)$, by Lemma \ref{lemma:15}, we have
    \begin{equation}
        \mathbb{E}\left( \tilde{T}_i\right)=T_i+\frac{2\hat{d}_{max}^2}{\varepsilon_2^2}.
    \end{equation}
    Then, we can easily obtain
      \begin{equation}
    \begin{aligned}
        \mathbb{E}\left( X_2\right) &=\mathbb{E}\left( \frac{1}{2} \sum_{i=1}^n{\tilde{d}_i \tilde{T}_i}\right) \\
        &=\frac{1}{2} \sum_{i=1}^n{d_i T_i} \\
        &=\sum_{\left( v_i,v_j \right) \in E}{d_id_j},
    \end{aligned}
      \end{equation}
    Following a similar derivation as in Eq.(\ref{eq:proof_Y1}), we have
      \begin{equation}
    \begin{aligned}
        \mathbb{E}\left( Y_2 \right) =\left[ \sum_{\left( v_i,v_j \right) \in E}{\frac{1}{2}\left( d_i+d_j \right)} \right] ^2.
    \end{aligned}
    \end{equation}
    Finally, we show that  
    \begin{equation}
        \mathbb{E}\left( \hat{q}_{ru}\left( G \right) \right) =q_{ru}\left( G \right).
    \end{equation}

\end{proof}

\subsection*{Proof of Theorem 4.11}

\setcounter{equation}{0}                                                                                                        \renewcommand\theequation{F.\arabic{equation}} 

    \textbf{Theorem 4.11} \emph{When $\varepsilon_1$, $\varepsilon_2$ are constants, the estimate $\hat{q}_{ru}\left( G\right)$ produced by $\mathbf{Decentral_{ru}}$ provides the following utility guarantee:}
    \begin{align}
        \mathrm{MSE}\left( \hat{q}_{ru}\left( G \right) \right) &= O\left( \frac{n^3d_{\max}^{6}}{M^4} \right) . 
    \end{align}

\begin{proof}
    Since $\tilde{T}_i\sim \text{Lap}\left( T_i,\frac{\hat{d}_{\max}}{\varepsilon _2} \right) $, then by Lemma \ref{lemma:15}, we can get
\begin{eqnarray}
\mathbb{E}\left( \tilde{T}_{i}^{2} \right) =T_{i}^{2}+\frac{2\hat{d}_{\max}^{2}}{\varepsilon _{2}^{2}},\label{subeq:F-1-1}
\end{eqnarray}
    

    Below we calculate the MSE of $q_{ru}\left( G\right)$ produced by $\mathbf{Decentral_{ru}}$. Let $H=\frac{X_3}{M}$ and $S=-\frac{Y_3}{M^2}$, then
    \begin{equation}\label{eq:F-3}
    \begin{aligned}
        \text{MSE}\left( \hat{q}_{ru}\left( G \right) \right) &=\mathbb{V}\left( \hat{q}_{ru}\left( G \right) \right)=\mathbb{V}\left( H+S \right) \\
        &\le 4\max \left\{ \mathbb{V}\left( H \right) ,\mathbb{V}\left( S \right) \right\}. \quad \text{(by Lemma \ref{lemma:16})}
    \end{aligned}
    \end{equation}
    Next, we calculate $\mathbb{V}\left( H \right)$ and $\mathbb{V}\left( S \right)$ respectively.
    \begin{equation}
    \begin{aligned}\label{eq:F-4}
        \mathbb{V}\left( H\right)&=\mathbb{V}\left( \frac{X_3}{M}\right)=M^{-2}\mathbb{V}\left( X_3\right) \\
        &=\frac{1}{4M^2} \sum_{i=1}^n{\mathbb{V}\left(\hat{d}_i \hat{T}_i\right)} \\
        &=\frac{1}{4M^2} \sum_{i=1}^n{\left[\mathbb{E}\left(\hat{d}_i^2 \hat{T}_i^2\right)-\mathbb{E}\left(\hat{d}_i \hat{T}_i\right)^2\right]} \\
        &=\frac{1}{4M^2} \sum_{i=1}^n{\left[\mathbb{E}\left(\hat{d}_i^2\right)\mathbb{E}\left(\hat{T}_i^2\right)-\mathbb{E}\left(\hat{d}_i\right)^2 \mathbb{E}\left(\hat{T}_i\right)^2\right]} \\
        &=\frac{1}{4M^2} \sum_{i=1}^n{\left[\left(d_i^2+\frac{2}{\varepsilon_1^2}\right)\left(T_i^2+\frac{2(d_{[1]}^\ast+d_{[2]}^\ast)^2}{\varepsilon_2^2}\right)-d_i^2 T_i^2\right]} \\
        &=\frac{1}{4M^2} \sum_{i=1}^n{\frac{2 \varepsilon_2^2 T_i^2+2 \varepsilon_1^2 d_i^2 (d_{[1]}^\ast+d_{[2]}^\ast)^2+4(d_{[1]}^\ast+d_{[2]}^\ast)^2}{\varepsilon_1^2 \varepsilon_2^2}} \\
        & \le \frac{1}{4M^2} \sum_{i=1}^n{\frac{2 \varepsilon_2^2 d_{max}^4+2 \varepsilon_1^2 d_{max}^2(d_{[1]}^\ast+d_{[2]}^\ast)^2+4(d_{[1]}^\ast+d_{[2]}^\ast)^2}{\varepsilon_1^2 \varepsilon_2^2}} \\
        &=O\left( \frac{nd_{max}^2(d_{[1]}^\ast+d_{[2]}^\ast)^2}{M^2}\right).
    \end{aligned}
    \end{equation}

    Since the expression for $S$ is similar to the expression for $Q$ in the proof of Theorem 4.3, we obtain $\mathbb{V}\left( S \right) \le O\left( \frac{n^3d_{\max}^{6}}{M^4}\right)$.

    To sum up, we have 
    \begin{equation}
        \text{MSE}\left( \hat{q}_{ru}\left( G \right) \right)= O\left( \frac{n^3d_{\max}^{6}}{M^4}\right).
    \end{equation}
\end{proof}

\end{document}